\documentclass[11pt]{amsart}
\let\oldphi\varphi \let\varphi\phi \let\phi\oldphi
\let\oldepsilon\varepsilon \let\varepsilon\epsilon \let\epsilon\oldepsilon

\usepackage{geometry}
\usepackage[toc,page]{appendix}
\usepackage{amsmath,mathtools, bm}
\usepackage{tensor}
\usepackage[bbgreekl]{mathbbol}
\usepackage{amsfonts}
\usepackage[utf8]{inputenc}
\usepackage{graphicx}
\usepackage[T1]{fontenc}
\usepackage[english]{babel}
\usepackage{csquotes}
\usepackage{graphicx}
\usepackage[dvipsnames]{xcolor} 
\usepackage{fancyhdr}
\usepackage{float}
\usepackage{afterpage}
\usepackage{placeins}
\usepackage{esvect}
\usepackage{multicol}
\usepackage{setspace}
\usepackage[font=small,labelfont=bf]{caption}
\usepackage{url}
\usepackage{musicography}
\usepackage{comment}
\usepackage[all,cmtip]{xy} 
\usepackage{tikz}
\usepackage{tikz-cd}
\usetikzlibrary{tqft}
\usetikzlibrary{positioning}
\usetikzlibrary{arrows.meta}
\usetikzlibrary{arrows}
\usepackage{xspace}
\usepackage{dsfont}
\usepackage{amssymb}
\usepackage{enumitem}
\usepackage{mdframed}
\usepackage{stmaryrd}

\usepackage{amsthm}
\DeclareMathAlphabet{\mymathbb}{U}{BOONDOX-ds}{m}{n}

\usepackage{xstring}
\usepackage{subfigure}
\usepackage{multirow}
\usepackage{hyperref}
\mathtoolsset{showonlyrefs = true}

\makeatletter
\newcommand{\phantompageref}{\def\@currentHref{page.\@the@H@page}}
\makeatother

\graphicspath{
	{../pics/}
}

\makeindex

\makeatletter
\DeclareFontFamily{OMX}{MnSymbolE}{}
\DeclareSymbolFont{MnLargeSymbols}{OMX}{MnSymbolE}{m}{n}
\SetSymbolFont{MnLargeSymbols}{bold}{OMX}{MnSymbolE}{b}{n}
\DeclareFontShape{OMX}{MnSymbolE}{m}{n}{
    <-6>  MnSymbolE5
   <6-7>  MnSymbolE6
   <7-8>  MnSymbolE7
   <8-9>  MnSymbolE8
   <9-10> MnSymbolE9
  <10-12> MnSymbolE10
  <12->   MnSymbolE12
}{}
\DeclareFontShape{OMX}{MnSymbolE}{b}{n}{
    <-6>  MnSymbolE-Bold5
   <6-7>  MnSymbolE-Bold6
   <7-8>  MnSymbolE-Bold7
   <8-9>  MnSymbolE-Bold8
   <9-10> MnSymbolE-Bold9
  <10-12> MnSymbolE-Bold10
  <12->   MnSymbolE-Bold12
}{}

\let\llangle\@undefined
\let\rrangle\@undefined
\DeclareMathDelimiter{\llangle}{\mathopen}%
                     {MnLargeSymbols}{'164}{MnLargeSymbols}{'164}
\DeclareMathDelimiter{\rrangle}{\mathclose}%
                     {MnLargeSymbols}{'171}{MnLargeSymbols}{'171}
\makeatother

\definecolor{light-gray}{gray}{0.95}
\usepackage[color=light-gray]{todonotes}

\newcommand{\RR}{\ensuremath{\mathbb{R}}}

\newcommand{\AB}{\ensuremath{\mathbb{A}}}

\newcommand{\BB}{\ensuremath{\mathbb{B}}}
\newcommand{\QQ}{\ensuremath{\mathbb{Q}}}

\newcommand{\XX}{\ensuremath{\mathbb{X}}}

\newcommand{\HB}{\ensuremath{\mathbb{H}}}
\newcommand{\NN}{\ensuremath{\mathbb{N}}}
\newcommand{\ab}{\ensuremath{\mathbb{a}}}
\newcommand{\fb}{\ensuremath{\mathbb{f}}}

\newcommand{\EL}{\ensuremath{\mathcal{EL}}}
\newcommand{\ZC}{\ensuremath{\mathcal{Z}}}

\newcommand{\VV}{\ensuremath{\mathcal{V}}}

\newcommand{\BC}{\ensuremath{\mathcal{B}}}
\newcommand{\LL}{\ensuremath{\mathcal{L}}}
\newcommand{\FF}{\ensuremath{\mathcal{F}}}
\newcommand{\FC}{\ensuremath{\mathfrak{F}}}

\newcommand{\HH}{\ensuremath{\mathcal{H}}}
\newcommand{\SC}{\ensuremath{\mathcal{S}}}

\newcommand{\AC}{\ensuremath{\mathcal{A}}}

\newcommand{\VF}{\ensuremath{\mathfrak{X}}}
\newcommand{\gf}{\ensuremath{\mathfrak{g}}}
\newcommand{\ii}{\ensuremath{\mathfrak{i}}}

\newcommand{\so}{\ensuremath{\mathfrak{so}}}
\newcommand{\su}{\ensuremath{\mathfrak{su}}}

\newcommand{\df}{\ensuremath{\mathfrak{d}}}
\newcommand{\mm}{\ensuremath{\mathfrak{m}}}
\newcommand{\kk}{\ensuremath{\mathds{k}}}

\newcommand{\BV}{\ensuremath{\mathfrak{BV}}}

\newcommand{\qBF}{dq\emph{BF}\xspace}
\newcommand{\qGR}{qeGR\xspace}

\newcommand{\Ci}{\ensuremath{\mathcal{C}^{\infty}}}

\newcommand{\ddelta}[2]{\ensuremath{\frac{\delta #1}{\delta #2}}}

\renewcommand{\graph}{\mathrm{graph}}

\DeclareMathOperator\Kill{\mathcal{K}}
\newcommand{\killing}[2]{\ensuremath{\Kill\left( #1, #2 \right)}}
\newcommand{\pair}[2]{\ensuremath{\left\langle #1, #2 \right\rangle}}

\newcommand{\paird}[2]{\ensuremath{\left\llangle #1, #2 \right\rrangle_\df}}
\newcommand{\pairg}[2]{\ensuremath{\left\llangle #1, #2 \right\rrangle_\gf}}
\newcommand{\pairT}[2]{\ensuremath{\left\llangle #1, #2 \right\rrangle_T}}

\newcommand{\lie}[1]{\ensuremath{\LL_{#1}}}
\newcommand{\inp}{\ensuremath{\pair{\cdot}{\cdot}}}
\newcommand{\lieprod}[2]{\ensuremath{\left\langle #1 \wedge #2 \right\rangle}}

\newcommand{\must}{\overset{!}{=}}

\newcommand{\lra}{\longrightarrow}
\newcommand{\lla}{\longleftarrow}
\newcommand{\Lra}{\Longrightarrow}

\newcommand{\lmap}{\longmapsto}
\newcommand{\atob}[1]{\overset{#1}{\lra}}

\DeclareMathOperator\Loc{Loc}

\DeclareMathOperator\Aut{Aut}

\DeclareMathOperator\locf{loc}
\DeclareMathOperator\topf{top}

\DeclareMathOperator\sgn{sgn}
\DeclareMathOperator\id{id}
\DeclareMathOperator\tr{Tr}

\DeclareMathOperator{\im}{im}

\newtheoremstyle{note} 
{1.2em}                
{1.2em}                
{}                     
{}                     
{\scshape\bfseries}    
{ }                    
{.5em}                 
{}                     
\theoremstyle{note}
\newtheorem{theorem}{Theorem}[section]
\newtheorem{definition}[theorem]{Definition}
\newtheorem{notation}[theorem]{{Notation}}
\newtheorem{lemma}[theorem]{Lemma}
\newtheorem{corollary}[theorem]{Corollary}

\newtheorem{proposition}[theorem]{Proposition}
\newtheorem{calculation}[theorem]{Calculation}

\newtheorem{example}[theorem]{Example}
\newtheorem{assumption}[theorem]{Assumption}
\newtheorem{construction}[theorem]{Construction}

\newtheorem{remark}[theorem]{Remark}

\usepackage[backend=biber, giveninits=true, style=alphabetic, sorting=nyt,isbn=false, doi=false, maxalphanames=5,url=false, maxbibnames=99]{biblatex}
\title[deformed $BF$ gravity]{Gravity with torsion as deformed $BF$ theory}

\author{Alberto S. Cattaneo}
\address{Institut f\"ur Mathematik, Universit\"at Z\"urich, Winterthurerstrasse 190, 8057 Z\"urich, Switzerland}
\email{cattaneo@math.uzh.ch}

\author{Leon Menger}
\address{Department of Mathematics, University of Notre Dame, 255 Hurley Bldg, Notre Dame, IN 46556, USA}
\email{leonmenger@nd.edu}

\author{Michele Schiavina}
\address{Department of Mathematics, University of Pavia, Via Ferrata 5, 27100 Pavia, Italy}
\address{INFN Sezione di Pavia, via Bassi 6, 27100 Pavia, Italy}
\email{michele.schiavina@unipv.it}
\date{}

\thanks{
ASC acknowledge partial support of the SNF Grant No. 200020 192080 and of the Simons Collaboration on Global Categorical Symmetries. This research was (partly) supported by the NCCR SwissMAP, funded by the Swiss National Science Foundation. This article is based upon work from COST Action 21109 CaLISTA, supported by COST (European Cooperation in Science and Technology) (www.cost.eu)}

\begin{document}

\maketitle
\thispagestyle{empty}

\begin{abstract}
    We study a family of (possibly non topological) deformations of $BF$ theory for the Lie algebra obtained by quadratic extension of $\mathfrak{so}(3,1)$ by an orthogonal module. The resulting theory, called quadratically extended General Relativity (\qGR), is shown to be classically equivalent to certain models of gravity with dynamical torsion. The classical equivalence is shown to promote to a stronger notion of equivalence within the Batalin--Vilkovisky formalism. In particular, both Palatini--Cartan gravity and a deformation thereof by a dynamical torsion term, called (quadratic) generalised Holst theory, are recovered from the standard Batalin--Vilkovisky formulation of \qGR
    by elimination of generalised auxiliary fields.
\end{abstract}

\tableofcontents

\section*{Introduction}

Topological models like $BF$ theory and deformations thereof,\footnote{These often go under the name of $BF + \Lambda BB$ theories, see \cite{Cattaneo1998}.} while known to have many applications in pure mathematics related to the study of topological invariants (see e.g.\ \cite{Schwarz,Pavelthesis,Cattaneo_2020,Hadfield_2020}), have also proven to be relevant in theoretical physics, especially in the formulation of gauge theories of gravity and their subsequent quantisation. For instance, in three spacetime dimensions, general relativity in the coframe formalism can be cast as a pure $BF$ theory, and is therefore topological. This fact was originally observed by \cite{carlip_1998,Witten3dgrav}, who noticed that the symmetries of $BF$ theory can be used to generate diffeomorphisms on-shell. A strong-equivalence result of the Hamiltonian dg-manifolds underlying the two theories was later proven in the Batalin--Vilkovisky (BV) formalism \cite{CattaneoSchiavinaSelliah2018,CanepaSchiavina19}, a result that also hints at the equivalence of the quantisation of the two theories. Indeed the quantisation of $BF$ theory in the BV formalism, pioneered by Cattaneo, Mnev and
Reshetikhin \cite{CMRCorfu, CattaneoMnevReshetikhin2014, CattaneoMnevReshetikhin2017}, can be used to inform our understanding of the quantisation of triadic gravity. (Note, however, that the equivalence only holds in the \emph{nondegenerate sector} of $BF$ theory.)

In four dimensions the situation is different: Gravity is no longer topological\footnote{By this we mean that there are two propagating, local, degrees of freedom in the phase space.} and can thus not be fully described by a $BF + \Lambda BB$ theory. There have been some attempts at retaining a formulation of gravity as a (deformed) $BF$ theory, via certain conditions on the $B$ field, usually called \emph{simplicity constraints} and enforced dynamically, which break the symmetries of the theory to allow for local degrees of freedom. The two main approaches in this vein are known as Plebanski theory \cite{Plebanski77} and MacDowell--Mansouri theory \cite{MacDowellMansouri1977} (for a review see \cite{Freidel2005, Freidel_2012}). However, it has been suggested that these constraints lead to difficulties when attempting to quantise the system \cite{Mikovi2006}. Consequently, it would be beneficial, both conceptually and for practical purposes, to find a formulation of tetradic gravity based on $BF$ theory together with a symmetry breaking mechanism with better properties.

Another interesting aspect which has received increasing attention lately is the role of torsion in theories of gravity. Together with curvature and non-metricity, it is one of the three geometric quantities that can be used to formulate a fully general theory of tetradic gravity. While both Einstein--Hilbert \cite{Hilbert1915} and Einstein--Palatini \cite{HehlReview76,FFR} models require the torsion to vanish, and thus feature the unique metric-compatible Levi-Civita connection, there exists an equivalent formulation of gravity where curvature is required to vanish, but torsion is allowed to take on non-zero values. This is known as \emph{Teleparallel Gravity} \cite{ArcosPereira2004, Hammond2002, Hohmann2018}. Meanwhile, \emph{Palatini--Cartan gravity}, another classically equivalent theory, imposes neither vanishing torsion nor vanishing curvature from the get-go but rather obtains vanishing torsion ``dynamically'', that is to say: as an algebraic field equation. 

In recent years, the role of torsion terms in tetradic gravity theories has been thoroughly investigated, and several proposals have been put forward, linking torsion effects to the evolution of the early universe and its late accelerated expansion \cite{Poplawski2013Cosmo}, as well as coupling gravity to non-integer spin matter \cite{Poplawski2013Spin, Diether2020} or even investigating dark matter in the presence of quadratic torsion terms \cite{Megier2020}. Thus, while on one hand it is desirable for any model of gravity not to hinge on restrictions of either curvature or torsion, on the other hand it seems like torsion terms, linear or quadratic\footnote{The terms ``linear'' and ``quadratic'' {refer to whether the dynamical torsion field appears linearly or quadratically in the action functional.}}, should receive more attention from a field theoretic perspective. From this point of view, our search for a $BF$-type model for general relativity will naturally output a theory that is sensitive to non-trivial torsion.

In this paper we present a novel formulation of gravity as a symmetry-broken deformation of $BF$ theory. To this end we use the notion of (trivial\footnote{We will present only the most trivial extension possible, subsequent obvious generalisations being straightforward. See Appendix \ref{App:QuadLie}.}) standard quadratic Lie algebra extensions introduced in \cite{KathOlbrich2006, KathOlbrich2007} and consider a $BF + \Lambda BB$ theory for a quadratically extended Lie algebra to be our basic (topological) model. We call the resulting model \emph{quadratically extended $BF$ theory} (qBF). We then introduce a fairly general (quadratic) deformation term, leading us to \emph{deformed, quadratically extended, $BF$ theory} (\qBF), which overall depends on two real parameters and a Lie algebra automorphism $T$, called deformation automorphism. In order to obtain a theory that can reasonably model general relativity, we specialise \qBF to the trivial quadratic extension of the Lie algebra $\so(3,1)$ and call the resulting theory \emph{quadratically extended general relativity} (\qGR).

Our analysis shows that this model is general enough to describe both ``topological sectors'' as well as ``gravitational sectors'', meaning that by varying the deformation (and the parameters), the properties of the model---i.e.\ degrees of freedom, symmetries and equations of motion---change to describe other models of interests, some of which are models of gravity.

Indeed, within certain ranges of the deformation parameters, one is able to break just enough (shift) symmetries to obtain non-topological theories with precisely $2$ local degrees of freedom (in configuration space), corresponding to the propagating degrees of freedom of the graviton. In particular, we show that one can recover both Palatini--Cartan--Holst gravity (up to a boundary term) and a deformation thereof, modified by the presence of torsion as a dynamical field. This justifies the name ``gravitational sector'' of \qGR. We note, in passing, that \qGR implements (a modified version of) the simplicity constraints on the $B$ field.\footnote{This is already true for quadratically extended $BF$ theory.}

In order to substantiate what we mean by ``recovering'' a theory from another one, this work uses the Batalin--Vilkovisky (BV) formalism, which encodes---via cohomological resolutions---the space of inequivalent physical configurations of the theory (the moduli space). To a classical field theory with symmetries, the BV framework associates a Hamiltonian dg manifold $\FC=(\FF, \varpi,S,Q)$, comprised of a $(-1)$-symplectic graded manifold $(\FF, \varpi)$ together with an action functional $S$ over $\FF$ whose Hamiltonian vector field $Q$ has the cohomological property (i.e.\ $[Q,Q]=0$). Using $Q$ one can then define the BV complex of the theory, which resolves the original moduli space.\footnote{By this we mean that the cohomology in degree zero of the BV complex is (a replacement of) the space of functions functions over the moduli space, which is often singular.}

It is natural to phrase equivalence of field theories in this language. One can however formulate several inequivalent notions of equivalence, of increasing strength, three of which are directly important for this work. \emph{Classical BV equivalence} demands that the BV-cohomologies of two  theories agree in degree $0$, i.e.\ for classical fields. This weakest form of equivalence can be investigated using only the classical Lagrangian field theory data and usually provides a starting point for investigating stronger equivalences. A finer notion is \emph{BV equivalence}, which requires the BV complexes of two BV theories to be quasi-isomorphic, i.e.\ that there be a map of Hamiltonian dg manifolds, inducing an isomorphism between their associated cohomologies in all degrees (for the precise notion see Definition \ref{def:weakEQ}). Finally, we mention \emph{strong BV equivalence}, which is an isomorphism of Hamiltonian dg manifolds. A fitting non-trivial example of the latter is the strong BV equivalence between $BF$ theory and triadic gravity, proven in \cite{CattaneoSchiavinaSelliah2018}, which clarifies in what sense general relativity is essentially $BF$ theory in three dimensions. See also \cite{CanepaSchiavina19}.

In order to establish that two theories are BV equivalent, a handy notion is that of auxiliary fields \cite{HENNEAUX1990,Barnich1995}, which essentially encodes the idea that two theories may differ by certain field configurations that are only used to present certain physical content in a particularly useful way, but which are otherwise ``physically irrelevant''. In the language of resolutions, this is translated in the statement that said fields have an associated trivial cohomology and are ``contractible'', so they can be eliminated without changing the underlying physics. {(For further details on this procedure and its applications, see \cite{Barnich2005,Barnich_2011,Grigoriev_2011}.)}

In what constitutes the main part of this work, we identify and eliminate as many auxiliary fields as possible in the BV formulation of \qGR (for certain appropriate choices of parameters and deformations), and find BV-equivalent theories so that the classical equivalences are promoted to BV equivalences. In particular, for a certain choice of parameters and deformations, we find a BV equivalence between \qGR and Palatini--Cartan--Holst theory, modulo a boundary term that corresponds to the Nie-Yan class (cf.\ Proposition \ref{prop:qgH_PC} and Remark \ref{rem:gHlikePCH}). 

Our main result is Theorem \ref{theo:WeakBVdiagram}, which describes general relativity as an instance of \qGR theory, in the sense that the latter is BV-equivalent to both quadratic generalised Holst (qgH) and generalised Holst (gH) theory (Definition \ref{def:ClassicalgH}), both of which are slightly generalised versions of very well-known models for coframe gravity. 

The action functional of gH theory features a quadratic torsion term $\eta(d_A e, d_A e)$, where $e$ is a tetrad and $A$ an $\so(3,1)$ connection. As such, it is a good starting point for further investigations into the role of such terms in relation to known theories of gravity. We will show in Proposition \ref{prop:gH_PCH} that the quadratic torsion term $\eta(d_A e, d_A e)$ corresponds to the topological Nieh--Yan term and can be written as a boundary term plus a topological contribution to the action. Due to the general form of the family of theories defined in this work, this gives a rather general result on the role of quadratic torsion terms of this form.

{We also analyse the classical boundary structure of \qBF (in the sense of Section \ref{subsubsec:KTconstruction}) and compute the associated constraint algebra. We find that the constraint set of the theory is first class only in the topological sector, simultaneously hinting at the fact that: 1.\ including dynamical torsion fields term makes the Cauchy problem more involved, and 2.\ that the simplicity constraint---which reduces the field $B$ to the quadratic tensor $e\wedge e$ for $e$ a tetrad---is second class and as such not easily quantisable. For further details see Section \ref{subsec:qBFboundary} and Remark \ref{rmk:firstvssecondclass}.}

The paper is organised as follows. In Section \ref{subsec:QuadLie} and Section \ref{subsec:BVformalism} we briefly introduce the constructions of standard Lie algebra extensions and the basic notions of the BV formalism. In Section \ref{sec:ClassicalAnalysis} we formulate quadratically extended GR (\qGR), the main subject of this work, and conduct a thorough classical analysis including relations to theories and modifications of gravity. Finally in Section \ref{sec:BVtheories} we formulate the BV theory of \qGR and find BV equivalences providing direct relations to the BV data of PCH gravity.

\subsection*{Torsion in Theories of Gravity}

In a field theory that depends on a principal connection $A$ and a coframe field $e$ (a tetrad or vierbein) there are two derived, gauge covariant, quantities of immediate interest: The curvature $F_A$ and the torsion form $d_A e$.

Depending on the properties of curvature and torsion that are assumed from the get-go (and not as dynamical field equations), we classify theories of gravity in the following way\footnote{For a more in-depth discussion see the introduction of \cite{ArcosPereira2004}.}:
\begin{itemize}
    \item {Impose vanishing torsion $d_A e = 0$:} In this situation one can make use of the fundamental theorem of Riemannian geometry \cite{lee2019introduction} to conclude that the connection must be the unique compatible \textit{Levi-Civita connection}. This is the basis for Einstein--Hilbert gravity. This is sometimes, and perhaps more appropriately, called Einstein--Cartan theory (see e.g. \cite{krasnov_2020}) to highlight the use of coframes, as opposed to metrics.

    \item {Impose vanishing curvature $F_A = 0$:} Such connections are called \textit{Weizenböck connections}. They allow for a force-based description of gravity called \textit{Teleparallel gravity} (for contemporary reviews see \cite{ArcosPereira2004, Hammond2002}).

    \item {No restrictions:} This allows for both torsion and curvature effects. The simplest form of such a theory is Palatini--Cartan gravity. (Note that, when the $SO(3,1)$ connection is left free to vary but within a metric formulation, one often speaks of \emph{Einstein--Palatini} gravity.\footnote{For a historical review of appropriate nomenclature and attribution, see \cite{HehlReview76,FFR}.})
\end{itemize}

In physical literature, it is usually tacitly assumed that the underlying connection in a theory of gravity is the Levi-Civita connection, and theories not assuming this are considered \textit{modifications of gravity}. Over the past years modifications using torsion terms have seen an increase in popularity. From the evolution of the early and late universe \cite{Poplawski2012, Poplawski2013, Ivanov2016, Cubero2019}, the amalgamation of fermionic particles with gravity \cite{Poplawski2013, Tecchiolli19, Diether2020} to attempts at explaining Dark Matter \cite{Megier2020}. There are a plethora of applications where torsion modifications beyond the usual Teleparallel and Palatini--Cartan gravity are considered. Note that other schools like entropic gravity \cite{Carroll2016} or modified Newton dynamics \cite{Scarpa2006} are also rising in popularity. This work however will only yield results about torsion modifications.

\section*{Acknowledgements}
This work stems from one of the authors' master thesis at ETH Zurich.

\section{Background}

This section introduces the concepts underlying the core results of the paper. In Section \ref{subsec:QuadLie} the basic concepts of quadratic Lie algebra extensions are introduced. For this we follow \cite{KathOlbrich2006, KathOlbrich2007} where the (trivial) standard extensions used in this work have first been constructed and discussed.

Section \ref{subsec:ClassicalFormalism} is dedicated to the introduction of classical Lagrangian field theories, their classical boundary structure and reduced phase space. For this part we refer to \cite{AndersonBicomplex, Deligne++} and to \cite{KijTul1979} for the investigation of the boundary.

In Section \ref{subsec:BVformalism} we discuss the fundamentals of the BV formalism, with a special focus on equivalences. More detailed introductions together with motivations for the construction of the BV complex can be found in \cite{BATALIN1977, BATALIN1983, CattaneoMoshayedi2020, CattaneoMnevReshetikhin2017, CattaneoMnevReshetikhin2014, MnevLectures, Stasheff97}.

\subsection{Trivial Standard Extensions}
\label{subsec:QuadLie}

This section presents the construction of a metric Lie algebra from any finite-dimensional Lie algebra $\gf$ and an orthogonal $\gf$-module. We briefly review \cite[Sec. 2]{KathOlbrich2006} whose definitions and results we use throughout.

\begin{definition}[Metric Lie Algebras]
  A metric Lie algebra\index{Metric Lie algebra} over a field $\kk$ is a pair $(\gf, \pair{\cdot}{\cdot})$ where $\gf$ is a Lie algebra and
  \begin{equation}
    \pair{\cdot}{\cdot} \colon \gf \otimes \gf \lra \kk
  \end{equation}
  is a bilinear form that is symmetric, invariant under the adjoint action of $\gf$ on itself and non-degenerate in that the induced morphism $\gf \lra \gf^*$ is a linear isomorphism of Lie algebras.
\end{definition}

\begin{definition}[Orthogonal Lie Algebra Modules]
\label{def:OrthLieModule}
  Let $\gf$ be a finite-dimensional Lie algebra, $V$ a finite-dimensional real vector space equipped with an invariant symmetric bilinear form $\pair{\cdot}{\cdot}_V$ and $\rho \colon \gf \lra \so(V)$ a representation of $\gf$ on $V$. The invariance of $\pair{\cdot}{\cdot}_V$ is expressed algebraically for arbitrary $g \in \gf$ and $v,w \in V$ as
  \begin{equation}
    \pair{\rho(g)v}{w}_V + \pair{v}{\rho(g)w}_V = 0.
  \end{equation}
  If additionally $\pair{\cdot}{\cdot}_V$ is non-degenerate call the triple $(\rho, V, \pair{\cdot}{\cdot}_V)$ an orthogonal $\gf$-module\index{Orthogonal lie algebra module}.
\end{definition}

\begin{remark}
  From now on let $\gf$ be a finite-dimensional Lie algebra and $(\rho, V, \inp_V)$ an orthogonal $\gf$-module, thought of as an abelian metric Lie algebra.
\end{remark}

Combining the wedge product $\wedge$ on Chevalley--Eilenberg cochains $C^\bullet_{CE}(\gf; V)$ of $\gf$ and the non-degenerate symmetric pairing $\pair{\cdot}{\cdot}_V$ on $V$ one obtains a bilinear multiplication map 

\begin{align*}
\langle \cdot \wedge \cdot \rangle \colon C_{CE}^i(\gf; V) \times C_{CE}^j(\gf; V) \atob{\wedge} C_{CE}^{i+j}(\gf; V \otimes V) &\atob{\pair{\cdot}{\cdot}_V} C_{CE}^{i+j}(\gf),\\
  (\alpha, \beta) &\lmap \langle \alpha \wedge \beta \rangle.
\end{align*}

\begin{definition}[Quadratic Extensions]
  A {quadratic extension}\index{quadratic extension} of $\gf$ by the orthogonal $\gf$-module $(\rho, V, \inp_V)$ is given by a tuple $(\mm, \ii, i,p)$ such that
  \begin{itemize}
    \item $(\mm, \inp)$ is a metric Lie algebra,
    \item $\ii \subset \mm$ is an isotropic ideal, i.e. $\inp|_\ii = 0$ and $[\ii, \mm] \subset \ii$
    \item $i$ and $p$ are Lie algebra homomorphisms constituting a short exact sequence of Lie algebras
    \begin{equation}
    \label{eq:abcdefg}
      0 \lra V \atob{i} \mm/\ii \atob{p} \gf \lra 0,
    \end{equation}
    consistent with $\rho$ in that for all $\widetilde{L} \in \mm/\ii$ with $p(\widetilde{L}) = L \in \gf$ and $v \in V$
    \begin{equation}
      i(\rho(L) v) = [\widetilde{L}, i(v)] \in i(v).
    \end{equation}
    Additionally we require $\im(i) = \ii^\perp/\ii$ and $i \colon V \lra \ii^\perp/\ii$ is an isometry.
  \end{itemize}
\end{definition}

\begin{remark}
  Note that the above characterisation of a quadratic Lie algebra extension is ultimately given by a short exact sequence. If we assume $V$ to be an abelian Lie group that additionally is isotropic in $\mm/\ii$ the quadratic extension is a central extension.
\end{remark}

The following construction of quadratic Lie algebra extensions is a special case of the one presented in \cite{KathOlbrich2006} and requires only a Lie algebra and an orthogonal module. For the full procedure, see Section \ref{App:QuadLie}.

\begin{definition}[Trivial Standard Extensions]
\label{def:StandardExtension_triv}
  Define the vector space $\df \coloneqq \gf^* \oplus V \oplus \gf$ and an inner product $\inp_\df \colon \df \times \df \lra \kk$ by
  \begin{equation}
    (z_1 + v_1 + g_1) \times (z_2 + v_2 + g_2) \lmap \pair{v_1}{v_2}_V + z_1(g_2) + z_2(g_1).
  \end{equation}
  Further denote by $[\cdot, \cdot]_\df \colon \df \times \df \lra \df$ the unique antisymmetric bilinear assignment such that
  \begin{enumerate}
    \item $\pair{[X,Y]_\df}{Z}_\df = \pair{X}{[Z,Y]_\df}_\df$ for all $X,Y,Z \in \df$,
    \item $[\df, \gf^*]_\df \subset \gf^*$, $[\gf^*, \gf^*]_\df = 0$, $[V, V]_\df \subset \gf^*$,
    \item $[g_1, g_2]_\df = [g_1, g_2]_\gf$ for all $g_1, g_2 \in \gf$,
    \item $\pair{[g, v_1]_\df}{v_2}_\df = \pair{\rho(g) v_1}{v_2}_\df$ for all $g \in \gf$ and $v_1, v_2 \in V$.
    \end{enumerate}
    We call this the \emph{trivial standard extension} of $\gf$ by $V$ and denote it by $(\df(\gf, V, \rho), \gf^*, i, p)$ or simply as the underlying metric Lie algebra $\df(\gf, V, \rho) \coloneqq (\df, [\cdot, \cdot]_\df, \inp_\df)$, for short.

    Note that $i \colon V \lra V \oplus \gf$ and $p \colon V \oplus \gf \lra \gf$ are given by the canonical injection and projection. (In the above equations the canonical inclusions of $\gf^*, \gf$ and $V$ into $\df$ were omitted to clean up the equations.)
\end{definition}

The definition of trivial standard extensions is justified by the following result:

\begin{lemma}[of {\cite[Proposition 3.1, 3.2]{KathOlbrich2006}}]
  The antisymmetric bilinear assignment $[\cdot, \cdot]_\df$ constructed in Definition \ref{def:StandardExtension_triv} is unique. The tuple $\df(\gf, V, \rho) \coloneqq (\df, [\cdot, \cdot]_\df, \inp_\df)$ defines a metric Lie algebra and a quadratic extension of $\gf$ by $(\rho, V, \inp_V)$.
\end{lemma}

\begin{remark}
  This work makes use of the trivial standard extension of a Lie algebra. There does however exist a more general extension procedure that depends on the choice of elements of the quadratic cohomology of the underlying Lie algebra. This procedure is described in Section \ref{App:QuadLie} and evidently provides an immediate path for a subsequent generalisation of the (trivially-)quadratically extended $BF$ theory we define in this paper.
\end{remark}

\subsection{Classical Lagrangian Field Theory}
\label{subsec:ClassicalFormalism}

In this section we will introduce the basic notions of classical Lagrangian field theory with local symmetries, and the study of their induced boundary data. This will serve as the core for the investigations in Section \ref{sec:ClassicalAnalysis}. As indicated, we follow \cite{AndersonBicomplex, Deligne++, BlohmannLFT} for the classical Lagrangian field theory, and refer to \cite{KijTul1979,CMRCorfu,CattaneoMnevReshetikhin2014,Giachetta1995, Moshayedi2022} for the construction of boundary data, as described in Section \ref{subsubsec:KTconstruction}.

Throughout, $M$ will be a smooth manifold, possibly with boundary. A local Lagrangian field theory is the assignment, to $M$, of the following data: A space of fields $\FF$ and a \emph{local} functional $S\colon \FF\to\RR$, the action functional. Locality is encoded by the fact that $\FF=\Gamma(M,F)$ is the space of sections of a (possibly graded) fibre bundle\footnote{Often one can simply work with vector bundles, but in our applications we need to consider open submanifolds of vector bundles due to the nondegeneracy conditions proper of general relativity. The grading will be useful to encode BV data, see Section \ref{sec:BVtheories}.} $F\to M$, and local forms are intuitively dependent on the fields and their derivatives up to some finite order $k$, as detailed here:

\begin{definition}[Local Forms and Local Functionals \cite{Zuckermann1987,AndersonBicomplex,BlohmannLFT}]
\label{def:local_forms_functionals}
    Let $F \atob{\pi} M$ be a (graded) fibre bundle over a smooth, compact, manifold $M$. Define $\FF \coloneqq \Gamma(M, F)$, $J^k(F)$ as the $k$-th Jet bundle. Further define the jet evaluation maps 
    \[
    j^k \colon \FF \times M \lra J^k(F), \qquad (\phi,x) \mapsto j^k\phi(x)
    \]
    where $j^k\phi$ is the kth jet equivalence class, and denote by $j^\infty$ their projective limit. Similarly construct the {infinite jet bundle} as the projective limit of the following sequence:
    \begin{align}
        F = J^0(F) \lla J^1(F) \lla \cdots \lla J^k(F) \lla \cdots
    \end{align}
    Now finally define the {bicomplex of local forms} on $\FF \times M$ as
    \begin{align}
        (\Omega_{\locf}^{\bullet, \bullet}(\FF \times M), \delta, d) \coloneqq (j^\infty)^* (\Omega^{\bullet, \bullet}(J^\infty(F)), d_V, d_H).
    \end{align}
    Here $d_V$ and $d_H$ denote, respectively the \emph{vertical} and the \emph{horizontal} differentials on the variational bicomplex. Furthermore the differentials $\delta$ and $d$ are defined for $\alpha \in \Omega^{\bullet, \bullet}(J^\infty(F))$ via
    \begin{align}
        \delta (j^\infty)^* \alpha = (j^\infty)^* d_V \alpha,
        d (j^\infty)^* \alpha = (j^\infty)^* d_H \alpha.
    \end{align}
    The complex of integrated local $k$-forms is the image of the map
    \[
    \int_M\colon \Omega_{\locf}^{k, \topf}(\FF\times M) \to \Omega_{\int}^{k}(\FF)
    \]
    with vertical differential $\delta$. Elements of $\Omega_{\locf}^{0, \bullet}(\FF \times M)$ (and $\Omega_{\int}^{0}(\FF)$) will be called (integrated) local functionals.
\end{definition}

\begin{definition}[Classical Field Theories]
    A {classical field theory} on a smooth manifold $M$ is a pair $(\FF, S)$ where the \emph{space of fields} $\FF$ is as above and the {action} is an integrated local functional $S \in \Omega_{\locf}^{0}(\FF)$, i.e.\ there exists a \emph{Lagrangian density} $L \in \Omega_{\locf}^{0, \topf}(\FF \times M)$ such that $S = \int_M L$.
    To any element $L\in\Omega_{\locf}^{0, \topf}(\FF \times M)$ one can associate\footnote{This is a Theorem of Zuckermann \cite[Theorem 3]{Zuckermann1987}. See also \emph{ibid.} for further details.} a local \emph{Euler--Lagrange} form $\mathsf{el}\in \Omega_{\locf}^{1,\topf}(\FF\times M)$ and a \emph{boundary} form $\theta\in\Omega_{\locf}^{1,\topf -1}(\FF\times M)$ such that
    \begin{align}
     \delta L = \mathsf{el} + d \theta.
    \end{align}
    An element $\phi\in\FF$ is called a field configuration, and it is said to be extremal if and only if $\mathsf{el}(\phi)=0$.
    Given the  classical field theory  $(\FF, S)$, its {critical locus} is the zero locus $\EL \coloneqq \Loc_0(\mathsf{el}) \subset \FF$ of the Euler Lagrange form $\mathsf{el}$ associated to the Lagrangian of the theory. We call the elements of $\EL$ {classical solutions}.
\end{definition}

An important concept in field theory that will be relevant for us is that of \emph{local/gauge symmetry}. In accordance with the notion of locality at the core of our framework we define

\begin{definition}[Infinitesimal Local Symmetries]
    Let $(\FF, S)$ be a classical field theory. We call a section $V \in \Gamma(\FF, T\FF)$ an {infinitesimal local symmetry} of the theory if
    \begin{align}
        \lie{V} L = dY_V,
    \end{align}
    for some local form $Y_V\in \Omega_{\locf}^{0,\topf}(\FF\times M)$.
\end{definition}

\begin{remark}\label{rmk:non-trivialsymmetries}
    The set of all symmetries forms a subalgebra of the algebra of vector fields $\widetilde{\mathcal{D}} \subseteq \Gamma(\FF,T\FF)$. Hovever, among such symmetries there are \emph{trivial} ones, denoted $\mathcal{D}_0$, which identically vanish on $\EL$ and form a subalgebra of all symmetries. A choice of non-trivial symmetries for the field theory $(\mathcal{F},S)$ is a choice of embedding of all symmetries modulo trivial ones $\mathcal{D}=\widetilde{\mathcal{D}}/\mathcal{D}_0 \hookrightarrow \widetilde{\mathcal{D}}$. One can show that, in general, this embedding forms an involutive distribution \emph{only on $\EL$} \cite{HenneauxTeitelboim1992}. From its very definition it is easy to see that an (infinitesimal) local symmetry leaves the critical locus $\EL$ invariant. Hence, classical solutions that are related by a symmetry describe the same physics, and the space of classical observables is some suitable model for $C^\infty(\EL)^{\mathcal{D}} \simeq C^\infty(\EL/\mathcal{D})$. Generally, however, both $\EL$ and the quotient $\EL/\mathcal{D}$ can develop singularities, and one does not have easy access to geometry over these spaces.
    This is one of the reasons why we treat Lagrangian field theory with symmetries in the BV formalism, as it instead builds a replacement for $C^\infty(\EL/\mathcal{D})$ in terms of the cohomology of an appropriately constructed complex. (See Section \ref{sec:BVtheories}.)
\end{remark}

Another pivotal point for field theory and this paper is the equivalence of two theories, for now on a classical level:

\begin{definition}[Classical Equivalences]
\label{def:ClassicalEQclass}
    Let $(\FF_i, S_i)$ for $i = 1,2$ be two classical field theories with spaces of classical solutions $\EL_i$ and symmetry spaces $D_i$. We say that the two theories are {classically equivalent} if\footnote{{When a smooth description of the quotient spaces is not available, we resort to the algebraic version of this equivalence, $\Ci(\EL_1 / D_1 |_{\EL_1}) \simeq \Ci(\EL_2 / D_2 |_{\EL_2})$, given an appropriate model for $\Ci(\EL_i / D_i |_{\EL_i})$.}}
    \begin{align}
        \EL_1 / D_1 |_{\EL_1} \simeq \EL_2 / D_2 |_{\EL_2}
    \end{align}
\end{definition}

As we said, before one can make sense of classical equivalence, an appropriate model for $C^\infty(\EL_i/\mathcal{D}_i)$ needs to be specified. In Section \ref{subsec:BVformalism} we will discuss the (standard) BV extensions of classical field theories, which will allow us to answer this question and give a more precise meaning to Definition \ref{def:ClassicalEQclass}. Throughout this paper we will however encounter several situations where one can establish classical equivalence without resorting to the BV formalism:

\begin{remark}
\label{rmk:ClassicalEQsimple}
    In practical scenarios one might prove classical equivalence by checking that $\EL_1 \simeq \EL_2$ and $D_1 |_{\EL_1} \simeq D_2 |_{\EL_2}$. Sometimes, it is possible to find that an even stronger relation holds when comparing field theories. For example, let $(\FF_i, S_i)$ for $i = 1,2$ be two classical field theories with spaces of classical solutions $\EL_i$ and symmetry spaces $D_i$. Let further $C \subset \FF_1$ be a set of fields where a subset of the equations defining the zero locus $\EL_1$ is fulfilled. If one can find an isomorphism $\varphi_{cl} \colon C \lra \FF_2$ such that
    \begin{align}
        S_1 |_C = \varphi_{cl}^* S_2, \qquad \text{and} \qquad D_1 |_C \overset{\varphi_{cl}}{\simeq} D_2 |_{\EL_2},
    \end{align}
    then the two theories are classically equivalent.
\end{remark}

\subsubsection{Kijowksi--Tulczijew boundary data for classical field theory}
\label{subsubsec:KTconstruction}

In this section we will briefly review the construction of boundary data\footnote{Truly, one simply needs an embedded codimension $1$ submanifold, but the boundary scenario allows for further interesting interpretations of the data thus obtained.} for a classical field theory $(\FF,S)$ on a manifold with boundary $(M,\partial M)$, due to Kijowksi and Tulczijew \cite{KijTul1979}. The construction can be summed up as follows:

\begin{enumerate}
\item First we define a na\"ive version of ``fields near the boundary'' by restricting fields and their derivatives along a direction transverse to the boundary. This space can be equipped with a closed two-form in a natural way. 
\item When this form is also pre-symplectic (see below for the precise statement), one can then perform pre-symplectic reduction to obtain a symplectic space of ``boundary fields'', which is often also called  the \emph{geometric phase space} of the theory. 
\item Inside the geometric phase space there is a subset of field configurations that can be extended to solutions of the equations of motion in a small cylinder. This is the Cauchy set of the theory, which in good cases is a coisotropic sumbanifold. Its coisotropic reduction yields the ``reduced phase space'', which is interpreted as the space of admissible boundary configurations modulo symmetries.
\end{enumerate}

More precisely:
\begin{definition}[Preboundary Fields]
\label{def:PreboundaryFields}
    Let $(\FF, S)$ be a classical field theory defined on a smooth manifold with boundary $(M,\partial M)$ and $\epsilon > 0$. Define the {space of pre-boundary fields} $\widetilde{\FF}^\partial$ as the space of germs\footnote{Remember that a germ at $x \in M$ is the equivalence class of functions $f \in \Ci(U), g \in \Ci(V)$ where $U,V$ are open neighbourhoods of $x$, such that $f \sim g$ if $f$ and $g$ agree on any subset of $U \cap V$.} of fields at $\partial M \times \{0\}$ on the manifold $\partial M \times [0,\epsilon]$.
\end{definition}

Let $(\FF, S)$ be a classical field theory and $\widetilde{\FF}^\partial$ its space of pre-boundary fields with $\widetilde{\pi}^\partial \colon \FF \lra \widetilde{\FF}^\partial$ the canonical projection. Then one can always look at $\widetilde{\alpha}^\partial \in \Omega^{1}_{\int}(\widetilde{\FF}^\partial)$ such that\footnote{We denote $\mathsf{EL} = \int_M \mathsf{el}$.}
\begin{align}
    \delta S = \mathsf{EL} + (\widetilde{\pi}^\partial)^* \widetilde{\alpha}^\partial.
\end{align}

The relevant assumption of regularity here is on the null distribution of the associated two form $\widetilde{\omega}^\partial=\delta \widetilde{\alpha}^\partial$

\begin{assumption}
\label{as:GeometricSymplectic}
    The tuple $(\widetilde{\FF}^\partial, \widetilde{\varpi}^\partial)$ where $\widetilde{\varpi}^\partial \coloneqq \delta \widetilde{\alpha}^\partial$ is a pre-symplectic manifold, i.e.\ the kernel $\mathrm{ker}(\widetilde{\varpi}^\flat)$ is a subbundle of $T\widetilde{\FF}^\partial$. Moreover, pre-symplectic reduction by said kernel is the space of sections of a fibre bundle over $\partial M$.
\end{assumption}

\begin{definition}[Geometric Phase Space]
\label{def:GeometricPhase}
    Let $(\FF, S)$ be a classical field theory and $(\widetilde{\FF}^\partial, \widetilde{\varpi}^\partial \coloneqq \delta \widetilde{\alpha}^\partial)$ the associated pre-symplectic space of pre-boundary fields. Define the {geometric phase space} $\FF^\partial$ of $(\FF, S)$ by
    \begin{align}
        \left(\FF^\partial \coloneqq \widetilde{\FF}^\partial / \ker(\widetilde{\varpi}^\partial), \varpi^\partial \right), 
    \end{align}
    where $\varpi^\partial$ is the reduction of $\widetilde{\varpi}^\partial$. We denote the resulting quotient map by $\pi^\partial \colon \FF \lra \FF^\partial$.
\end{definition}

In good cases, one can also find a form $\alpha^\partial \in \Omega^{1}_{\int}(\FF^\partial)$ such that
\begin{align}
    \delta S = \mathsf{EL} + (\pi^\partial)^* \alpha^\partial, \qquad \varpi^\partial = \delta \alpha^\partial,
\end{align}

Let $(\FF, S)$ be a classical field theory, $\EL$ its space of classical solutions and $\FF^\partial$ its geometric phase space together with the surjective submersion $\pi^\partial \colon \FF \lra \FF^\partial$. Inside $\FF^\partial$ we select those configurations that can be extended to a classical solution in a small enough cylinder $\partial M \times [0,\epsilon]$. This is the space of Cauchy data of the theory, and we denote it by $\mathcal{C}\subset \FF^\partial$. In good cases $\mathcal{C}$ is coisotropic. Indeed, it is expected that the projection $\mathcal{EL}^\partial \doteq \pi^\partial(\mathcal{EL})$ be a Lagrangian submanifold in $\FF^\partial$, due to existence and uniqueness of solutions of the field equations.\footnote{This is phrased as intersection of transverse Lagrangians, in this language, where the \emph{other} Lagrangian is a boundary condition.} Since obviously $\mathcal{EL}^\partial \subset \mathcal{C}$, the set of Cauchy data must be coisotropic (because it contains a Lagrangian sumbanifold). See \cite{CMRCorfu} for further details.

\begin{definition}[Reduced Phase Space]
\label{def:ReducedPhase}
Let $(\FF, S)$ be a classical field theory, and let $(\FF^\partial,\varpi^\partial)$ be its geometric phase space. If $\mathcal{C}$ is a (smooth) coisotropic submanifold, the {reduced} (or {physical}) {phase space} is the symplectic reduction $\underline{\mathcal{C}}\doteq \mathcal{C}/\mathcal{C}^{\varpi^\partial}$, where $\mathcal{C}^{\varpi^\partial}$ is the characteristic foliation of $\mathcal{C}\subset \FF^\partial$. 
\end{definition}

The reduced phase space of a classical field theory is the space of equivalence classes of boundary configurations that extend to a solution of the equations of motion. This can be naturally interpreted as the space of possible initial conditions modulo symmetry. In concrete examples, and in particular thoughout this paper, the manifold $\mathcal{C}$ will be presented as the vanishing locus of a set of functions over field configurations, often called \emph{constraints} of the theory. The number of such (independent) constraints, compared to the total (local) degrees of freedom of the theory, informs us on the expected dimension of the reduction. In particular, topological theories are such that, after reduction, there are no residual local degrees of freedom. In gravitational theories, instead, we expect two local degrees of freedom to survive after reduction.

\subsection{Batalin--Vilkovisky Formalism}
\label{subsec:BVformalism}

In this section we introduce a cohomological approach to Lagrangian field theory with symmetries that goes under the name of Batalin--Vilkovisky (BV) formalism \cite{BATALIN1977, BATALIN1981}. Our brief introduction draws from \cite{CattaneoMnevReshetikhin2014, CattaneoMnevReshetikhin2017, MnevLectures, CattaneoMoshayedi2020} and the earlier works of Stasheff \cite{Stasheff97}. The main motivation for the use of this technique is that we want to explore the moduli space of classical solutions, $\mathcal{EL}/D$, which is often singular, and a smooth description is not available. Thus, it is convenient to replace the geometric study of the moduli space with a cohomological description of the quotient.
The interested reader is referred to these sources for a more self-contained introduction and further details.

\begin{definition}[BV Theory]
\label{def:BVtheory}
    A BV theory is a tuple $\FC=(\FF, \varpi, Q, \SC)$ where $\FF$ is a graded manifold (the space of BV fields), $\varpi$ is a $(-1)$-symplectic form (the {BV two-form}) and $\SC \colon \FF \lra \RR$ is a Hamiltonian function of degree $0$ (the {BV action}) with Hamiltonian vector field $Q \in \VF[1](\FF)$ and satisfying the {Classical Master Equation}\footnote{Using that $Q$ is the Hamiltonian vector field of the BV action $\SC$ we can require, equivalently to $\{\SC,\SC\} = 0$, that $[Q,Q] = 0$, i.e.\ that $Q=\{S,\cdot\}$ be \emph{cohomological}.} ({CME}):
  \begin{equation}
    \{\SC,\SC\} = 0.
  \end{equation}
  Here $\{\cdot, \cdot\}$ denotes the graded Poisson bracket induced by $\varpi$. Thus $(\FF, Q)$ is a dg-manifold.\footnote{In this paper, since we do not consider physical fermions, we always assume that the parity of a variable is the mod 2 reduction of its degree.}
\end{definition}

\begin{remark}
    When thinking of BV \emph{theories} one typically additionally wants $\FF$ to be the space of sections of a graded bundle over some (closed\footnote{Extensions to manifolds with boundary are considered in \cite{CattaneoMnevReshetikhin2014}, while for an approach towards noncompact manifolds we refer to \cite{RejznerThesis,RejznerSchiavina21}.}) $d$-dimensional manifold $M$, and assumes $Q,\varpi,S$ to be \emph{local}.
\end{remark}

\begin{example}[Standard BV Extension of a Classical Field Theory (compare \cite{BATALIN1981, HENNEAUX1990})]
    The classical field theories we will study in this work have symmetries that stem from the action of a Lie algebra. With reference to Remark \ref{rmk:non-trivialsymmetries}, this means that $\mathcal{D}$ is the image of a Lie algebra action on $\FF_{cl}$.
    In this case there exists a standard ``BV extension'', which assigns a BV theory to the classical field theory with symmetry $(\FF_{cl}, S_{cl},\mathcal{D})$ with symmetry Lie algebra $(\gf, [\cdot, \cdot])$:\\
    Define the space of BV fields
    \begin{align}
        \FF_{BV} \coloneqq T^*[-1](\FF_{cl} \times \Omega^0[1](M; \gf)),
    \end{align}
    where we denote by $(\phi^i, c^a)$ the fields in the base and by $(\phi_i^\dagger, c_a^\dagger)$ those in the cotangent fibres. Note that the shift induces an internal grading (ghost number); this is in addition to the form degree. Given the non-degenerate bilinear map
    \begin{align}
        \inp \colon \Omega^{i,k}_{\locf}(M) \times \Omega^{j,\topf-k}_{\locf}(M) \lra \Omega^{i+j,\topf}_{\locf}(M),
    \end{align}
    for any $i,j,k$ we can equip $\FF_{BV}$ with a symplectic structure, the BV two-form
    \begin{align}
        \varpi \coloneqq \int_M \pair{\delta \phi^i}{\delta \phi_i^\dagger} + \pair{\delta c^a}{\delta c_a^\dagger}.
    \end{align}
    The cohomological vector field of our BV extension nicely splits into the Koszul--Tate differential $\delta_{KT}$ and the Chevalley--Eilenberg differential $\gamma_{CE}$ as
    \begin{align}
        Q = \delta_{KT} + \check{\gamma}_{CE},
    \end{align}
    where $\check{\gamma}_{CE}$ is the cotangent lift of the Chevalley--Eilenberg cohomological vector field:
    \begin{alignat}{3}
    \gamma_{CE} \phi^i &= c^a v_a^i,\qquad  & \qquad  \delta_{KT} \phi^i &= 0, \qquad    & \qquad\delta_{KT} c^a &= 0, \\
    \gamma_{CE} c^a &= \frac{1}{2} [c, c]^a,  \qquad &\qquad  \delta_{KT} \phi_i^\dagger &= \ddelta{S}{\phi^i},  \qquad & \qquad \delta_{KT} c^\dagger_a &= v^i_a \phi^\dagger_i,
    \end{alignat}
    with $v^i_a$ denoting the fundamental vector fields of $\gf$ on $\FF_{cl}$. We define the BV action as
    \begin{align}
        S_{BV} \coloneqq S_{cl} + \int_M \left\{\pair{\phi^\dagger_i}{\gamma_{CE} \phi^i} + \pair{c^\dagger_a}{\gamma_{CE} c^a}\right\},
    \end{align}
    and the tuple $(\FF_{BV}, \varpi, Q, S_{BV})$ defines a BV theory in the sense of Definition \ref{def:BVtheory}.
\end{example}

\begin{definition}[BV Complexes (compare {\cite[Section 5]{Stasheff97}})]
Let $\FC$ be a BV theory with cohomological vector field $Q$. We define the {BV complex}\index{BV theory!BV complex} of $\FC$ to be the cochain complex $(\BV^\bullet=\Omega_{\int}^{0}(\FF),Q)$, where the grading is the internal ghost number. We denote the {BV cohomology}\index{BV theory!BV cohomology} by $\HH^\bullet(\BV^\bullet)$.
\end{definition}

\begin{remark}
As mentioned, the BV complex provides a cohomological replacement of the moduli space of solutions to the Euler-Lagrange equations of a field theory (modulo their symmetries). This is done by means of a combination of the Koszul--Tate and Chevalley--Eilenberg complexes, so that the BV cohomology in degree zero is isomorphic to the space of functions on the moduli space\footnote{Or a replacement thereof, when the moduli space is not smooth.}. For an interpretation of the higher cohomology groups see \cite[Sections 2 - 5]{Stasheff97}. Note that the BV complex, as defined above, can be straightforwardly generalised to higher integrated local forms. See \cite[2.2.4]{SimaoCattaneoSchiavina} and \cite{Mnev_2019}.
\end{remark}

We can now state a precise definition of classical equivalence between field theories \cite{Canepa_2021,SimaoCattaneoSchiavina}:

\begin{definition}[Classical BV Equivalences]\label{def:classEQ}
  Let $\FC_1$ and $\FC_2$ be two BV theories and $\BV^\bullet_1$ and $\BV^\bullet_2$ their respective BV complexes. The theories are said to be {classically equivalent}\index{BV theory!classical equivalence} if
  \begin{equation}
    H^0(\BV_1^\bullet) = H^0(\BV_2^\bullet).
  \end{equation}
\end{definition}

\begin{remark}
    Note that the above definition, while elegant, requires us to invoke the entire BV package. In practice one can often establish classical equivalence with only the classical field theories at hand using \ref{def:ClassicalEQclass}. In particular we will often encounter the even simpler situation presented in Remark \ref{rmk:ClassicalEQsimple}. This is precisely the approach we will take in Section \ref{sec:ClassicalAnalysis}.
\end{remark}

The BV cohomology beyond degree zero has more accurate information about a given field theory.
Hence, it is natural to consider the (\emph{a priori} finer) notion of BV equivalence:

\begin{definition}[BV Equivalence {\cite{SimaoCattaneoSchiavina}}]\label{def:weakEQ}
  Let $\FC_1$ and $\FC_2$ be two BV theories and $\BV^\bullet_1$ and $\BV^\bullet_2$ their respective BV complexes. The two theories are said to be {BV equivalent} if there exists a quasi-isomorphism between the two BV complexes.
  Namely if there exist two maps $f_1 \colon \FF_1 \lra \FF_2$ and $f_2 \colon \FF_2 \lra \FF_1$ such that their pullbacks define cochain maps, i.e.
  \begin{equation}
    f_1^* Q_2 = Q_1 f_1^*, \qquad f_2^* Q_1 = Q_2 f_2^*,
  \end{equation}
  which satisfy
  \begin{align}
  \label{eq:weakEQ}
    f_1^* [\varpi_2] &= [\varpi_1], \qquad f_2^* [\varpi_1] = [\varpi_2], \qquad f_1^* [\SC_2] = [\SC_1], \qquad f_2^* [\SC_1] = [\SC_2],
  \end{align}
  i.e. preserve the cohomology class of the respective actions and BV two-forms. Lastly they should be such that $f_1^* \circ f_2^*$ and $f_2^* \circ f_1^*$ are identities in cohomology, i.e. the two maps are quasi-inverse of one another.
\end{definition}

\begin{remark}
    It is obvious that we have the implication
      \begin{equation}
    \text{BV equivalence} \ \ \Lra \ \ \text{Classical BV equivalence}.
    \end{equation}
    However, the notion of BV equivalence given in Definition \ref{def:weakEQ} is sometimes referred to as ``weak'' BV equivalence, since it is modelled over that of a weak-equivalence between cochain complexes: quasi-isomorphism. One can thus speak of a ``strong'' BV equivalence if this quasi-isomorphism is an isomorphism. An interesting non-trivial example is the strong BV equivalence between $3$-dimensional $BF$ theory and $3$-dimensional triadic General Relativity.
    Classical equivalence is mostly an indicator that it might be interesting to investigate stronger notions of equivalence. Meanwhile, BV equivalence is a ``strong enough'' notion to infer parallels about the quantisation of two theories. Since their BV complexes are weakly isomorphic, they contain the same information but can operate on different field content. Practical examples are the first and second order formulations of theories (e.g. Yang-Mills) or the reduction of theories with fields that induce acyclic subcomplexes.
\end{remark}

An important tool when establishing BV equivalence between theories is the elimination of \textit{contractible pairs} discussed originally in \cite{HENNEAUX1990, Barnich1995, Brandt1997, Barnich2005} and more recently in \cite{Barnich_2011, Grigoriev_2011, RejznerThesis, SimaoCattaneoSchiavina}.

\begin{definition}[General contractible pairs]\label{def:generalContractible}
Let $\SC_1$ be a BV action of the form\footnote{Here, for simplicity $S$ can be thought of as a functional on some symplectic vector space.}
  \begin{equation}
    \SC_1 (a,b,c^\dagger, d^\dagger) = \SC_2(a,c^\dagger) + \frac{1}{2} \langle b, b \rangle,
  \end{equation}
  where $\langle\cdot, \cdot\rangle$ is a constant non-degenerate bilinear form on the space $V$ in which the fields $v$ take values, $\SC_2$ solves the master equation w.r.t. $a, c^\dagger$ and we can split the fields into two families $(a,c^\dagger)$ and $(b,d^\dagger)$. Let $Q_1$ denote the cohomological vector field of the BV-theory associated to $\SC_1$. If
\begin{equation}
    Q_1 d^\dagger = b, \qquad Q_1 b = 0,
\end{equation}
  we call the collection of fields $(b, d^\dagger)$ a {contractible pair}\index{Contractible pair}. More generally, let $\SC$ be a BV action that depends on the collections of fields $a, b, c^\dagger, d^\dagger$. If $Q$ denotes the Hamiltonian vector field of $\SC$, and
\begin{equation}
    (Q d^\dagger)|_{d^\dagger = 0} = 0,
\end{equation}
  can be uniquely solved for $b$ in terms of the $(a, c^\dagger)$ one calls $(b, d^\dagger)$ a {general contractible pair}.
\end{definition}

\begin{remark}
As was shown in \cite{Barnich2005} contractible pairs are related to the notion of ``auxiliary fields'',  \cite{HENNEAUX1990, Barnich1995}: it is always possible to eliminate such fields via suitable gauge fixing conditions without changing the BV cohomology. As mentioned in \cite{Barnich2005, Barnich_2011, SimaoCattaneoSchiavina} there exist a more general variant of the procedure of reduction of contractible pairs that allows for a construction similar to the one using usual contractible pairs. This can be interpreted as a ``semiclassical'' approximation of the BV-pushforward \cite{MnevLectures, CattaneoMnevReshetikhin2017}.
\end{remark}

\begin{construction}
\label{constr:GeneralContractible}
The appearance of a general contractible pair $(b,d^\dagger)$ reflects the fact that $d^\dagger = 0$ is a good (partial) gauge fixing, which allows us to ``integrate out'' the field $b$. 
Define
\begin{equation}
    B \coloneqq Q d^\dagger, \qquad D^\dagger \coloneqq d^\dagger,
\end{equation}
and note that the change of variables $(b, d^\dagger) \lmap (B, D^\dagger)$ is invertible at least near $d^\dagger = 0$ by definition of general contractible pairs. One can then define two maps $\psi, \varphi$ establishing a quasi-isomorphism between the theories, which are then BV equivalent (cf. \cite[Theorem 3.2.2.]{SimaoCattaneoSchiavina}).  
\end{construction}

The following result is an immediate consequence of the above construction:

\begin{proposition}
\label{prop:GeneralContractibleWeakBV}
  Let $\SC_1$ be a BV action depending on the collections of fields $a, b, c^\dagger, d^\dagger$ where $(b, d^\dagger)$ defines a general contractible pair. Denote by $\Phi$ the coordinate transformation $(b, d^\dagger) \lmap (B \coloneqq Q d^\dagger, D^\dagger \coloneqq d^\dagger)$ and define a map
  \begin{align}
    \varphi^* a = \widetilde{a}, \qquad \varphi^* c^\dagger &= \widetilde{c}^\dagger, \qquad \varphi^* B = 0, \qquad \varphi^* D^\dagger = 0.
  \end{align}
  The BV theory defined by $\SC_1$ is BV equivalent to the one defined by $\varphi^*(\Phi(\SC_1))$.
\end{proposition}

\begin{remark}
    In the case that $(b, d^\dagger)$ is a contractible pair on the nose, one can define two maps $\psi, \varphi$ via:
  \begin{align}
    \varphi^* a = \widetilde{a}, \qquad \varphi^* c^\dagger &= \widetilde{c}^\dagger, \qquad \varphi^* B = 0, \qquad \varphi^* D^\dagger = 0, \\
    \psi^* \widetilde{a} &= a, \qquad \psi^* \widetilde{c}^\dagger = c^\dagger.
  \end{align}
  Then $\alpha^* \coloneqq \varphi^* \circ \psi^*$ is the identity and $\chi^* \coloneqq \psi^* \circ \varphi^*$ is homotopic to the identity acting as
  \begin{align}
    \chi^* a = a, \qquad \chi^* c^\dagger = c^\dagger, \qquad \chi^* B = 0, \qquad \chi^* D^\dagger = 0.
  \end{align}
  This special case of Proposition \ref{prop:GeneralContractibleWeakBV} is discussed in \cite[Lemma 3.2.1 and Theorem 3.2.2.]{SimaoCattaneoSchiavina}.
\end{remark}

\section{Classical \qGR, Classical Equivalences and Boundary analysis}
\label{sec:ClassicalAnalysis}

This section presents a variety of theories and discusses the classical equivalences that can be established among them (Definition \ref{def:classEQ}).

In Section \ref{subsec:qBFclassical} we introduce deformed, quadratically extended $BF$ (\qBF) theory, restrict it to $\so(3,1)$ to formulate what we call quadratically extended General Relativity (\qGR), and we show that the latter is classically equivalent to (resp. quadratic) generalised Holst theory (gH) (resp.\ (qgH)), two classical field theories that we define in Section \ref{subsec:qBFclassicalEQ}. To this end, we employ the strategy outlined in Remark \ref{rmk:ClassicalEQsimple}.

Finally in Section \ref{subsec:qBFboundary} we study the (reduced) phase space of \qGR theory, by means of an analysis of the structure of boundary configurations, following Section \ref{subsubsec:KTconstruction}.

\subsection{Deformed, quadratically extended, \emph{BF} theory}
\label{subsec:qBFclassical}

The following is the standard definition of a class of topological field theories that go by the name of $BF$ theory.

\begin{definition}[$BF$ Theory]
\label{def:ClassicalBFBB}
    Let $\gf$ be any finite-dimensional Lie algebra and $G$ a Lie group integrating $\gf$. Let there further be an invariant inner product $\pairg{\cdot}{\cdot}$ on $\gf$ and a principal $G$-bundle $P \atob{\pi} M$ be defined on over a $4$-dimensional manifold $M$. A connection of this principal bundle is locally described by a local $1$-form $A$ with values in $\gf$. Four-dimensional {$BF$ theory} (abbr.\ {$BF$}) for $\Lambda \in \RR$ is given by the data $(\FF_{BF}, S_{BF}(\Lambda))$, where
    \begin{equation}
        \FF_{BF} \coloneqq \AC_P\times \Omega^2(M; \gf^*){\ni (A,B)},
    \end{equation}
    and
    \begin{align}
        S_{BF}(\Lambda) \coloneqq \int_M \pairg{B}{F_A} + \frac{\Lambda}{2} \pairg{B}{B}.
    \end{align}
    If $\Lambda = 0$ we call the theory {pure $BF$ theory}.
\end{definition}

\begin{remark}
    One of the many interesting features of $BF$ theories, and a main reason for our investigation, stems from the fact that one can deform the theory to obtain theories of gravity. Two recent reviews on this application of $BF$ theories can be found in \cite{Mikovi2006} as well as \cite{Freidel_2012}. We will also review some of these formulations ourselves in Section \ref{subsec:BFgravity} to contextualize the findings of this work. Note that such gravitational theories are commonly used in quantisation schemes like Loop Quantum Gravity (see for example \cite{Rovelli2006, Dupuis2011}).
\end{remark}

\begin{definition}
\label{def:TwistedPairing}
  Let $\gf$ be a finite-dimensional Lie group and $(\df, [\cdot, \cdot]_\df, \paird{\cdot}{\cdot})$ the trivial standard extension thereof by $(\rho, V, \inp_V)$. Let further an automorphism $T \in \Aut(\gf)$ be given. Then one can define a {twisted inner product} on $\df$ via
  \begin{align}
      \pairT{\cdot}{\cdot} \colon (z_1, v_1, \xi_1) \times (z_2, v_2, \xi_2) & \lmap \pair{v_1}{v_2}_V + \xi_1(T(z_2)) + \xi_2(T(z_1)).
  \end{align}
\end{definition}

Now we have all the necessary ingredients to formulate the main subject of this work.

Consider the trivial standard quadratic extension $(\df \coloneqq \gf \oplus V \oplus \gf^*, [\cdot, \cdot]_\df, \paird{\cdot}{\cdot})$ of a finite-dimensional Lie algebra $\gf$ by an orthogonal $\gf$-module $(\rho, V, \inp_V)$ and a Lie group $D$ integrating $\df$. Let further a principal $D$-bundle $P \atob{\pi} M$ be defined on over a $4$-dimensional manifold $M$, and denote by $\mathcal{A}_P$ the associated space of principal connections. 

\begin{definition}[Deformed, quadratically extended $BF$ theory]
\label{def:ClassicalTqBF}
    Given $\beta\in \mathbb{R}$ and $T\in\mathrm{Aut}(\gf)$, we define {deformed, quadratically extended, $BF$ theory} (abbr.\ {\qBF}) to be given by the data $(\FF_{\qBF}, S_{\qBF}(T,\beta,\alpha))$, where
    \begin{equation}
        \FF_{\qBF} \coloneqq \AC_P\times \Omega^2(M; \df^*){\ni (\AB,\BB)},
    \end{equation}
    and, 
    \begin{align}
        S_{\qBF}{(T,\beta,\alpha)} :&= \int_M \paird{\BB}{F_\AB} + \frac{\alpha}{2} \big( \paird{\BB}{\BB} + \beta \pairT{\BB}{\BB} \big) = S_{BF}(\alpha) + \beta \int_M\pairT{\BB}{\BB}.
    \end{align}
    Here $\pairT{\cdot}{\cdot}$ denotes the twisted inner pairing defined in Definition \ref{def:TwistedPairing}. We call $T$ the {deformation automorphism}.
\end{definition}

One immediately finds in Definition \ref{def:ClassicalTqBF} that $\alpha$ controls the deviation from a pure $BF$ theory while $\beta$ controls the deformation term that sets \qBF apart from $BF$ type theories. The following is an immediate consequence:

\begin{corollary}
\label{cor:qBF_topological_cases}
    For $\alpha = 0$, or $\beta = -1$, and $T = \id$, deformed, quadratically extended $BF$ theory is a pure $BF$ theory for the trivial standard extension $\df$. Similarly, for $\beta = 0$, or $T = \id$, it is a $BF$-type theory for the Lie algebra $\df$.
\end{corollary}

\begin{remark}[Specialisation]
    In the following we will set $\alpha = 1$. All results in this paper carry over to any $\alpha \in \RR$, however we found that varying the parameter did not give rise to any sectors that could not be investigated using combinations of $T$ and $\beta$. One can indeed set $\alpha$ to $1$, unless it is equal to zero, by rescaling the field B. Furthermore, from here on we fix $\gf = \so(3,1)$ together with the above identifications. It will be indicated whenever a result does not depend on these choices.
\end{remark}

\subsection{Quadratically extended General Relativity}
Consider 
an orthogonal $\gf \simeq \so(3,1)$-module $(\rho, V, \eta)$ where $\rho \colon \gf \lra O(V)$ and $\eta$ is a Minkowski metric on $V$.
Using the orthogonal module $(\rho, V, \eta)$ we construct the trivial standard extension $(\df \coloneqq \gf \oplus V \oplus \gf^*, [\cdot, \cdot]_\df, \paird{\cdot}{\cdot})$ (Definition \ref{def:StandardExtension_triv}) where 
\[
[V, V]_\df \subset \gf^*, \qquad [\gf^*, \gf^*]_\df = [\gf^*, V]_\df \equiv 0.
\]
Using the metric $\eta$ we can identify $V \simeq V^*$ and $\gf \simeq \bigwedge^2 V$. Moreover, $\gf^*$ is identified (as a vector space) with $\bigwedge^2 \VV$ via the Killing form $\killing{\cdot}{\cdot}$:
\begin{align}
  \gf \simeq \textstyle{\bigwedge}^2 \VV \lra \gf^*, \qquad \chi \lmap \killing{\chi}{\cdot},
\end{align}
(This choice is natural since $\gf=\mathfrak{so}(3,1)$ is semisimple.) Thus, there is a canonical pairing between two elements $(z_1,v_1,\xi_1), (z_2,v_2,\xi_2) \in \df$ (cf.\ Definition \ref{def:StandardExtension_triv}), given by
\begin{align}
  \paird{(z_1,v_1,\xi_1)}{(z_2,v_2,\xi_2)} &= \eta(v_1, v_2) + \killing{\xi_1}{z_2} + \killing{\xi_2}{z_1} \\
  &= \eta(v_1, v_2) + \killing{\xi_1}{z_2} + \killing{\xi_2}{z_1},
\end{align}
where, with a slight abuse of notation, $\killing{\cdot}{\cdot}$ also denotes the natural pairing of $\gf$ and its dual, after the above identification of $\gf\simeq\gf^*$. We can further naturally identify $[V, V]_\df \subset \gf^*$ for $v, w \in V$ as
\begin{align}
  [v,w] = \killing{\cdot}{v \wedge w}.
\end{align}
This yields the following relation between the Minkowski metric $\eta$ and the Killing form for $v, w \in V$ and $\chi \in \gf$:
\begin{align}
  \eta([\chi, v], w) = - \killing{\chi}{[v, w]} = - \killing{\chi}{v \wedge w}.
\end{align}

Denote again by $D$ the Lie group integrating $\df$, and $\mathcal{A}_P$ the space of connections on a principal $D$-bundle.
A field in $\FF_{\qBF}\coloneqq \AC_P \times \Omega^2(M; \df^*)$, 
projected onto the factors of $\df$, can be expressed as $\AB = (A, e, S)$ and $\BB = (B, \tau, c)$ where
\begin{equation}
  A \in \Omega^1(M; \gf), \qquad e \in \Omega^1(M; V), \qquad S \in \Omega^1(M; \gf^*),
\end{equation}
\begin{equation}
  B \in \Omega^2(M; \gf^*), \qquad \tau \in \Omega^2(M; V^*), \qquad c \in \Omega^2(M; \gf).
\end{equation}
Note that using the identification $V \simeq V^*$ one can naturally read $\BB$ as $(c, \tau, B) \in \Omega^2(M; \df)$. In particular we get:
\begin{align}
\label{eq:qBFcurvature}
\begin{split}
  F_\AB &= d \AB + \frac{1}{2} [\AB, \AB]_\df = (dA, de, dS) + \frac{1}{2} [(A,e,S), (A,e,S)]_\df\\
  &= (dA, de, dS) + \frac{1}{2} [(A,e,S), (A,e,S)]\\
  &= (dA, de, dS) + \left(\frac{1}{2} [A, A], [A, e], [A,S] + \frac{1}{2} [e, e]\right) \\
  &= \left( F_A, d_A e, d_A S + \frac{1}{2} [e, e] \right).
\end{split}
\end{align}

\begin{definition}[Quadratically extended GR]\label{def:qeGR}
    We define \emph{quadratically extended General Relativity} (\qGR), to be deformed, quadratically extended $BF$ theory (Definition \ref{def:ClassicalTqBF}) specialised to the Lie algebra $\df=\mathfrak{so}(3,1)^*\oplus V\oplus \mathfrak{so}(3,1)$.
\end{definition}

\begin{calculation}\label{cal:QExtendedEOM}
    Using the above identifications for $\gf = \so(3,1)$, the action of \qGR reads:
\begin{align}
\begin{split}
    S_{\qGR}(T, \beta) &:= \int_M \paird{\BB}{F_\AB} + \frac{1}{2} \big( \paird{\BB}{\BB} + \beta \pairT{\BB}{\BB} \big), \\
    &= \int_M \killing{B}{F_A} + \eta(\tau, d_A e) + \killing{c}{(\id + \beta T) B + \frac{1}{2} e \wedge e + d_A S} + \frac{(1 + \beta)}{2} \eta(\tau, \tau),
    \label{eq:qAction}
\end{split}
\end{align}
    and the field equations of quadratically extended General Relativity take the following form:
  
    \mathtoolsset{showonlyrefs = false}
    \begin{subequations}\label{eq:EOMsub}
    \begin{align}
        F_A + (\id + \beta T) c &= 0, \qquad \text{from} \quad \delta B \label{eq:EOMsub_a}\\
        d_A e + (1+\beta) \tau &= 0, \qquad \text{from} \quad \delta \tau \label{eq:EOMsub_b}\\
        \eta(d_A \tau, \cdot) + \killing{c}{e \wedge \cdot} &= 0, \qquad \text{from} \quad \delta e \label{eq:EOMsub_c}\\
        (\id + \beta T) B + d_A S + \tfrac{1}{2} e \wedge e &= 0, \qquad \text{from} \quad \delta c \label{eq:EOMsub_d}\\
        \killing{d_A B}{\cdot} - \killing{e \wedge \tau}{\cdot} - \killing{[c, S]}{\cdot} &= 0, \qquad \text{from} \quad \delta A \label{eq:EOMsub_e}\\
        d_A c &= 0, \qquad \text{from} \quad \delta S. \label{eq:EOMsub_f}
    \end{align}
    \end{subequations}
    \mathtoolsset{showonlyrefs = true}
    ~\\
\begin{proof}
  The stated form of the action functional is simply found by unfolding the definitions. To obtain the field equations, we calculate the variation of $S_{\qGR}$ (see \eqref{eq:qAction}) and set to zero:
  \begingroup
  \allowdisplaybreaks
  \begin{subequations}
  \begin{align}
    0 \must& \ \delta \int_M \killing{B}{F_A} + \eta(\tau, d_A e) + \killing{c}{(\id + \beta T) B + \frac{1}{2} e \wedge e + d_A S} + \frac{(1 + \beta)}{2} \eta(\tau, \tau) \\
    =& \ \int_M \killing{\delta B}{F_A + (\id + \beta T) c} - \killing{B}{d_A \delta A} + \killing{c}{e \wedge \delta e + [\delta A, S] - d_A \delta S} \\
    &+ \eta(\delta \tau, d_A e + (1 + \beta) \tau) + \eta(\tau, [\delta A, e] - d_A \delta e) + \killing{\delta c}{(\id + \beta T)B + \frac{1}{2} e \wedge e + d_A S}\\
    \simeq& \ \int_M \killing{\delta B}{F_A + (\id + \beta T) c} + \eta(\delta \tau, d_A e + (1 + \beta) \tau) + \eta(d_A \tau, \delta e) + \killing{c}{e \wedge \delta e} \\
    &+ \killing{\delta c}{(\id + \beta T)B + \frac{1}{2} e \wedge e + d_A S} + \killing{d_A B}{\delta A} \\
    &+ \eta(\tau, [\delta A, e]) + \killing{c}{[\delta A, S]} + \killing{d_A c}{\delta S}
  \end{align}
  \end{subequations}
  \endgroup
  Here $\simeq$ denotes equivalence up to boundary terms; in the last step the boundary term $-d \killing{A}{\delta A} - d \killing{c}{\delta S} - d \eta(\tau, d_A e)$ was dropped. From the above form of $\delta S_{\qGR}$ one can immediately infer the claimed field equations.
\end{proof}
\end{calculation}

\begin{remark}
    Since $d_A e$ is the torsion from, Equation \eqref{eq:EOMsub_b} can dynamically identify $\tau$ as the torsion form scaled by $- (1+\beta)$. In particular, one obtains vanishing torsion and no algebraic constraint for $\tau$ if and only if $\beta = -1$. From Equation \eqref{eq:EOMsub_a} one can further infer that having $\id + \beta T$ vanish, which holds for $T = - \frac{1}{\beta} \id$ when $\beta \neq 0$, yields vanishing curvature but not necessarily vanishing torsion, i.e., except for $\beta \neq -1$.
\end{remark}

\begin{table}[H]
\renewcommand{\arraystretch}{2.0}
\begin{center}
\begin{tabular}{|c|c|c|c|}
  \hline
  $\beta$ & $T$ & Type of Theory & $d_A e$\\
  \hline\hline
  $0$ & $\mathrm{any}$ & $BF$ with $\Lambda=1$ & $- \tau$\\
  \hline
  $\mathrm{any}$ & $\id$ & $BF$ with $\Lambda=\beta+1$ & $- (1+\beta) \tau$\\
  \hline
  $-1$ & $\neq \id$ & Torsionless \qGR & $0$\\
  \hline
  $\notin \{-1, 0\}$ & $\neq \id$ & Generic \qGR & $- (1+\beta) \tau$ \\
  \hline
\end{tabular}
\end{center}
\caption[Summary of special parameter choices for classical \qGR]
{
    Table displaying different choices for $\beta$ and $T$, their effect on the type of theory that \qGR produces and the presence of torsion therein. This work will mainly investigate the lower two parameter ranges where one could expect non-topological theories that might be linked to theories of gravity. The last two rows represent the ``gravitational sector'', while the first {two} represent the ``topological sector''. {Note that for $T=\id$ and $\beta= -1$ the theory is \emph{pure} $BF$ and the torsion vanishes.}
} \label{Tab:sectors}
\end{table}

\begin{remark}
\label{rem:untwisted_simpler}
    In the first {two} rows of Table \ref{Tab:sectors} we have that $\beta T = \beta \id$ always holds true.\footnote{Note that $\beta$ cannot be simplified as it may vanish.} 
    As a consequence, in these cases one can completely drop $T$ and write the action of \qGR in the following concise way:
    \begin{align}
        S_{\qGR}(\beta) = S_{BF}(1+\beta) = \int_M \paird{\BB}{F_\AB} + \frac{(1+\beta)}{2} \paird{\BB}{\BB}.
    \end{align}
    The field equations can be further simplified to
    \begin{align}
        F_\AB + (1+\beta) \BB = 0, \qquad d_\AB \BB = 0.
    \end{align}
\end{remark}

{In order to investigate the various sectors systematically, we observe the existence of a  (generic) simplification of the field equations for certain values of the parameters. Indeed,  requiring $\beta T \neq - \id$ implies that $\id + \beta T$ is invertible, since $T$ is an automorphism. We summarise the various sectors-defining conditions as:}

\begin{definition}
\label{def:twisted_untwisted}
    {We distinguish the following conditions on the parameters $\beta\in\mathbb{R}$ and $T\in\Aut(\mathfrak{so}(3,1))$ of \qGR:\footnote{Notice that an analogous decomposition holds for \qBF for any Lie algebra.}
    \begin{enumerate}
    \item ${\beta T = \beta \id}$ is the \emph{topological sector condition}.
    \item ${\beta T \not = \beta \id}$ is the \emph{gravitational sector condition}.
    \item $\beta T\not = - \id$ is the \emph{invertibility sector condition}.
    \end{enumerate}
    When the $\bullet$ sector condition holds, we will say that ``the parameters $(\beta,T)$ are in the $\bullet$ sector'', or simply ``in the $\bullet$ sector''.}
\end{definition}

{
\begin{lemma}\label{lem:sectorrelations}
    We have the following relations between sectors:
    \begin{enumerate}
        \item In the topological sector, the invertibility condition is equivalent to $\beta\not=-1$
        \item In the invertibility sector,
            \begin{itemize}
                \item $\beta=0$ implies the topological sector
                \item $\beta =-1$ implies the gravitational sector
            \end{itemize}
        \item In the complement of the invertibility sector, $\beta T = -\id$,
        \begin{itemize}
            \item $\beta T = \beta \id \iff \beta = -1$ (\emph{topological sector})
            \item $\beta T \not = \beta \id \iff \beta \not = -1$ (\emph{gravitational sector})
        \end{itemize}
    \end{enumerate}
\end{lemma}
\begin{proof}~
    \begin{enumerate}
        \item $\beta T = \beta \id$ holds in the topological sector. Hence, 
        \[
        \beta \id = \beta T \not= -\id \iff \beta \not = -1.
        \]
        \item Assume that $\beta T\not = -\id$. It is evident that $\beta = 0$ is compatible with this condition, but $\beta=0$ also implies $\beta T = \beta \id = 0$. On the other hand, if $\beta = -1$ we have 
        \[
        \beta^2 T \not = -\beta \id \iff -\beta T \not = -\beta \id \iff \beta T \not = \beta \id.
        \]
        \item Straightforward. Note that, in particular $\beta\not=0$.
    \end{enumerate}
\end{proof}
}

\begin{calculation}
\label{cal:TwoAlgebraic}
  {Whenever the invertibility condition ($\beta T \neq - \id$) holds,} we can rewrite the field equations \eqref{eq:EOMsub} as
  
  \mathtoolsset{showonlyrefs = false}
  \begin{subequations}\label{eq:EOMsub1}
  \begin{align}
    c &= - (\id + \beta T)^{-1} F_A,  \label{eq:EOMsub1_a}\\
    B &= - (\id + \beta T)^{-1} (d_A S + \tfrac{1}{2} e \wedge e) , \label{eq:EOMsub1_b}\\
    d_A e &= - (1 + \beta) \tau, \label{eq:EOMsub1_c}\\
    0 &= \killing{d_A B}{\cdot} - \killing{e \wedge \tau}{\cdot} - \killing{[c, S]}{\cdot}, \label{eq:EOMsub1_d}\\
    0 &= \eta(d_A \tau, \cdot) + \killing{c}{e \wedge \cdot} \label{eq:EOMsub1_e}.
  \end{align}
  \end{subequations}
  \mathtoolsset{showonlyrefs = true}
If, additionally, $\beta \neq -1$ we can rewrite Equations \eqref{eq:EOMsub1} as
\mathtoolsset{showonlyrefs = false}
\begin{subequations}\label{eq:EOMsub2}
    \begin{align}
        c &= - (\id + \beta T)^{-1} F_A,  \label{eq:EOMsub2_a}\\
        B &= - (\id + \beta T)^{-1} (d_A S + \tfrac{1}{2} e \wedge e) , \label{eq:EOMsub2_b}\\
        \tau &= - \frac{1}{(1 + \beta)} d_A e, \label{eq:EOMsub2_c}\\
        0 &= \killing{\frac{1}{(1+\beta)} \id - (\id + \beta T)^{-1} F_A}{e \wedge \cdot}, \label{eq:EOMsub2_d}\\
        0 &= \killing{ \left[ \frac{1}{(1+\beta)} \id - (\id + \beta T)^{-1} \right] (d_A e \wedge e)}{\cdot} \label{eq:EOMsub2_e}.
    \end{align}
\end{subequations}
\mathtoolsset{showonlyrefs = true}
\end{calculation}

\begin{remark}
    {When the invertibility condition holds}, the set of Equations \eqref{eq:EOMsub1} yields two purely algebraic constraints on the fields $B, c$ which do not influence the equations of $A$ and $e$.
    In particular the constraint for $B$ \eqref{eq:EOMsub1_b} looks like a deformation of the simplicity constraints that appear naturally in the Hamiltonian analysis of Palatini--Cartan--Holst theory \cite{Peldan1994, Barros01}. This constraint can be used to describe PCH theory as a constrained $BF$ theory (cf., e.g., \cite{Dupuis2011}). We will draw similar parallels in the following.
    Note that there are no constraints on the dynamics of the field $S$. This apparent decoupling will be made precise later by eliminating the field using general contractible pairs. 
\end{remark}

\begin{remark}
\label{rem:BdoesNothing}
    {Following Lemma \ref{lem:sectorrelations}, the invertibility condition is equivalent to $\beta\not=-1$ in the topological sector ($\beta T = \beta \id$). However, the converse is not true: $\beta \not =-1$ together with the invertibility condition ($\beta T  \neq -\id$) does not imply the topological sector condition.
    In that scenario, we can only conclude that} the set of Equations \eqref{eq:EOMsub2} now yields \emph{three} purely algebraic constraints on the fields $B, c$ and $\tau$ which comprise $\BB$.
\end{remark}

\subsection{Classical Equivalences}
\label{subsec:qBFclassicalEQ}

Before investigating the classical equivalences arising from the field equations of \qGR we will provide a definition of Palatini--Cartan--Holst theory which is a theory of gravity that can be linked to \qGR. (Cf. \cite[Definition 11]{Cattaneo2019}.)

\begin{definition}[Palatini--Cartan--Holst Theory]
\label{def:PCH}
    Let $\VV \lra M$ be a Minkowski bundle, a vector bundle with fibre the Minkowski vector space\footnote{With a slight abuse of notation we denote by $\eta$ also the (constant) ``Minkowski metric'' defined fibrewise by $(V,\eta)$.} $(V, \eta)$ over a $4$-dimensional manifold $M$, and let $P \lra M$ be the associated principal $SO(3,1)$ bundle. Let further $\gamma \in \RR$. The \emph{Palatini--Cartan--Holst theory} (abbr.\ PCH) is specified by the pair $(\FF_{PCH}, S_{PCH})$, where
    \begingroup
    \allowdisplaybreaks
    \begin{equation}
        \FF_{PCH} \coloneqq \underbrace{\Omega^1(M; \VV)}_{e} \times \underbrace{\AC_P}_{A},
    \end{equation}
    and
    \begin{align}
        S_{PCH}{(\gamma)} :&= \int_M \frac{1}{2} \tr\left[T_\gamma(F_A) \wedge e \wedge e\right],
    \end{align}
    \endgroup
    where we denoted by $T_\gamma$ the \emph{Holst automorphism} $\alpha \lmap \alpha + \frac{1}{\gamma} \star \alpha$. 
    
    \noindent In the limiting case $\gamma \lra \infty$ one has $T_\gamma \lra \id$ and the resulting theory is called Palatini--Cartan theory.\footnote{The nomenclature is somewhat inconsistent across the literature. See \cite[Nomenclature]{CattaneoSchiavina2018} and references therein for a review.}
\end{definition}

\begin{remark}[Nondegeneracy of coframes]
    Note that, in order to describe gravity, Palatini--Cartan theory (and all other coframe-based models considered here) needs to be supplemented with the requirement that $e\in\Omega^1(M,\mathcal{V})$ be nondegenerate, meaning that $e:TM \to \mathcal{V}$ is an isomorphism, so that $g=e^*\eta$ is a (Lorentzian) metric on $M$. All of our calculations and results are valid for the theories defined without the nondegeneracy requirement, and restrict in particular to their physically relevant nondegenerate versions. {For this reason, Definition \ref{def:PCH} is more general than what is usually considered.}
\end{remark}

The field equations found in Calculation \ref{cal:TwoAlgebraic}, given certain restrictions on the allowed values of $\beta$ and $T$, provide algebraic solutions for $c, B$ and $\tau$ and thus a starting point to investigate tentative classical equivalences between \qGR and classical theories with fewer fields.

\begin{definition}[(Quadratic) Generalised Holst Theory]
\label{def:ClassicalgH}
    Let $\beta \in \RR$, and let $T \in \Aut(\gf)$ be a deformation automorphism such that {the invertibility condition $\beta T \neq -\id$ holds}. 
    \begin{enumerate}
        \item The \emph{quadratic generalised Holst theory} (qgH) is specified by the pair $(\FF_{qgH}, S_{qgH})$, where
    \begingroup
    \allowdisplaybreaks
    \begin{equation}
        \FF_{qgH} \coloneqq \underbrace{\Omega^1(M; \VV)}_{e} \times \underbrace{\AC_P}_{A} \times \underbrace{\Omega^2(M; \VV)}_{\tau},
    \end{equation}
    and
    \begin{align}
        S_{qgH}(T, \beta) :&= \int_M \frac{1}{2} \killing{(\id + \beta T)^{-1} F_A}{e \wedge e} - \eta(\tau, d_A e) - \frac{(1 + \beta)}{2} \eta(\tau, \tau).
    \end{align}
    \endgroup 
    \item Let additionally $\beta\not=-1$. The \emph{generalised Holst theory} (gH) is specified by the pair $(\FF_{gH}, S_{gH})$, where
    \begingroup
    \allowdisplaybreaks
    \begin{equation}
        \FF_{gH} \coloneqq \underbrace{\Omega^1(M; \VV)}_{e} \times \underbrace{\AC_P}_{A},
    \end{equation}
    and
    \begin{align}
        S_{gH}{(T, \beta)} :&= \int_M \frac{1}{2} \left(\killing{(\id + \beta T)^{-1} F_A}{e \wedge e} + \frac{1}{(1 + \beta)} \eta(d_A e, d_A e)\right).
    \end{align}
    \endgroup
    \end{enumerate}
\end{definition}

\begin{remark}
\label{rem:qgH_Holst}
    The identifier ``generalised'' stems from the fact that qgH and gH theories can be seen as a generalisations of Palatini--Cartan--Holst theory (PCH), which is recovered (up to a boundary term) when ${T = \frac{1}{\beta}(\star^{-1} - \id)}$ (see Proposition \ref{prop:gH_PCH}). While certain choices of $T$ and $\beta$ yield topological theories,\footnote{The sense in which we will talk about topological theories will be rendered clear in Section \ref{subsec:qBFboundary}.} we will see that one can indeed model nontopological PCH theory using gH theory.
\end{remark}

\begin{lemma}
    The field equations of qgH are:
    \begin{subequations}\label{eq:EOM_qgH}
    \begin{align}
        \killing{(\id + \beta T)^{-1} F_A}{e \wedge \cdot} - \eta(d_A \tau, \cdot) &= 0, \qquad \text{from} \ \ \delta e \\
        \killing{(\id + \beta T)^{-1} (d_A e \wedge e)}{\cdot} - \eta(\tau, [\cdot, e]) &= 0, \qquad \text{from} \ \ \delta A \\
        d_A e + (1+\beta) \tau &= 0, \qquad \text{from} \ \ \delta \tau.
    \end{align}
    \end{subequations}
    For any $T, \beta$ satisfying the conditions of Definition \ref{def:ClassicalgH} the field equations of gH are:
    \begin{subequations}\label{eq:EOM_gH}
    \begin{align}
        \killing{\left((\id + \beta T)^{-1} - \frac{1}{1+\beta} \id \right) F_A}{e \wedge \cdot} &= 0, \qquad \text{from} \ \ \delta e \\
        \left((\id + \beta T)^{-1} - \frac{1}{1+\beta} \id \right) (d_A e \wedge e) &= 0, \qquad \text{from} \ \ \delta A.
    \end{align}
    \end{subequations}
\begin{proof}
    This follows directly from the variation of $S_{qgH}$ and $S_{gH}$.
\end{proof}
\end{lemma}

\begin{proposition}
\label{prop:qBF_cleq_qgH}
    Consider quadratically extended General Relativity, and {assume that the invertibility condition holds} $\beta T  \neq -\id$, then
    \begin{enumerate}
        \item quadratically extended General Relativity is classically equivalent to quadratic generalised Holst theory.
    \end{enumerate} 
    {Additionally, whenever $\beta\not=-1$,
    \begin{enumerate}[resume]
        \item quadratic generalised Holst theory is classically equivalent to generalised Holst theory.
    \end{enumerate}}
\begin{proof}
    Denote the action of \qGR for $T, \beta$ as assumed, as:
    \begin{align}
        S_{\qGR} &= \int_M \killing{B}{F_A} + \eta(\tau, d_A e) + \killing{c}{(\id + \beta T) B + \frac{1}{2} e \wedge e + d_A S} + \frac{(1 + \beta)}{2} \eta(\tau, \tau).
    \end{align}
    Further, denote the action of qgH as
    \begin{align}
        S_{qgH} = \int_M \frac{1}{2} \killing{(\id + \beta T)^{-1} F_{\widetilde{A}}}{\widetilde{e} \wedge \widetilde{e}} - \eta(\widetilde{\tau}, d_{\widetilde{A}} \widetilde{e}) - \frac{(1 + \beta)}{2} \eta(\widetilde{\tau}, \widetilde{\tau}),
    \end{align}
    where $\widetilde{e} \in \Omega^1(M; \VV)$, $\widetilde{A} \in \AC_P$ and $\widetilde{\tau} \in \Omega^2(M; \VV)$. Recall that $\beta T  \neq -\id$ implies that $\beta T + \id$ is invertible. Define the map
    \begin{align}
        \tilde{F_1} \colon \Omega^1(M; \VV) \times \AC_P \times \Omega^2(M; \VV) &\lra \Omega^2(M; \gf^*) \times \Omega^2(M; \gf) \times \Omega^1(M; \gf^*) \\
        (e, A, \tau) & \lmap (B(e,A,\tau),c(e,A,\tau),S(e,A,\tau)) = - (\id + \beta T)^{-1} \left( \frac{1}{2} e \wedge e, F_A, 0  \right).
    \end{align}
    Then, observing that the domain of $\tilde{F_1}$ is $\FF_{qgH}$ we have the induced map
    \[
    {F_1} \colon \FF_{qgH} \to \FF_{\qGR}; \qquad {F_1} \colon (e,A,\tau) \mapsto (e,A,\tau,\tilde{F_1}(e,A,\tau)).
    \]
    (Observe, in passing, that setting $C\doteq\graph(\tilde{F_1})=\im(F_1)\subset \FF_{\qGR}$ we have that $F_1\colon \FF_{qgH} \to C$ is a diffeomorphism.)
    
    \noindent Then, we have (note in particular the lack of boundary terms):
    \begingroup
    \allowdisplaybreaks
    \begin{align}
        ({F}^*_1 S_{\qGR}) =& \ \int_M \eta(\tau, d_A e) + \frac{(1 + \beta)}{2} \eta(\tau, \tau) - \killing{(\id + \beta T)^{-1} F_A}{\frac{1}{2} e \wedge e}\\
        =& \ \int_M - \frac{1}{2} \killing{(\id + \beta T)^{-1} F_A}{e \wedge e} + \eta(\tau, d_A e) + \frac{1}{2} \eta(\tau, \tau) \\
        =& \ - \left(\int_M \frac{1}{2} \killing{(\id + \beta T)^{-1} F_A}{e \wedge e} - \eta(\tau, d_A e) - \frac{(1 + \beta)}{2} \eta(\tau, \tau) \right) = \ - S_{qgH}.
    \end{align}
    (Note that this matches with the strategy of Remark \ref{rmk:ClassicalEQsimple}.)
    \endgroup
    {Now, let us further assume that $\beta\neq -1$, and denote the action of gH for $\widetilde{e} \in \Omega^1(M; \VV)$, $\widetilde{A} \in \AC_P$ and $T, \beta$ as above by
    \begin{align}
        S_{gH} = \frac{1}{2} \int_M \killing{(\id + \beta T)^{-1} F_{\widetilde{A}}}{\widetilde{e} \wedge \widetilde{e}} + \frac{1}{(1 + \beta)} \eta(d_{\widetilde{A}} \widetilde{e}, d_{\widetilde{A}} \widetilde{e}).
    \end{align}
    
    To prove the second claim, define another map
    \begin{align}
        &\tilde{F_2} \colon \Omega^1(M; \VV) \times \AC_P \lra \Omega^2(M; \gf^*) \times \Omega^2(M; \gf) \times \Omega^1(M; \gf^*) \times \Omega^2(M; \VV) \\
        &(e, A)  \lmap (B(e,A),c(e,A),S(e,A), \tau(e, A) = - (\id + \beta T)^{-1} \left( \frac{1}{2} e \wedge e, F_A, 0, \frac{(\id + \beta T)}{(1+\beta)} d_A e \right).
    \end{align}
    Then, observing that the domain of $\tilde{F_2}$ is $\FF_{gH}$ we have
    \begin{align}
        {F}_2 \colon \FF_{gH} \to \FF_{\qGR}; \qquad {F}_2 \colon (e,A) \mapsto (e, A,\tau(e, A),F(e,A)),
    \end{align}
    a diffeomorphism onto $\mathrm{graph}(\tilde{F_2})$, and
    similarly to the previous case one obtains:
    \begingroup
    \allowdisplaybreaks
    \begin{align}
        {F_2}^*S_{qgH} &= \int_M \frac{1}{2} \killing{(\id + \beta T)^{-1} F_{A}}{e \wedge e} + \frac{1}{2(1+\beta)} \eta( d_A e, d_{A} e) \\
        &= \frac{1}{2} \int_M \killing{(\id + \beta T)^{-1} F_{A}}{e \wedge e} + \frac{1}{(1+\beta)} \eta( d_A e, d_{A} e) = S_{gH}.
    \end{align}
    \endgroup
    }
\end{proof}
\end{proposition}

Using the relations between the Killing form, the trace and $\eta$ one can prove the following relation between gH and Palatini--Cartan--Holst theory:

\begin{proposition}
\label{prop:gH_PCH}
    Whenever $\beta \notin \{-1, 0\}$, generalised Holst theory with $T = \frac{1}{\beta}(\star^{-1} - \id)$ takes the following form:
    \begin{align}
        S_{gH}{\left(\frac{1}{\beta}(\star^{-1} - \id), \beta \right)} = \frac{1}{2} \int_M \tr[T_\gamma(F_A) \wedge e \wedge e] + \frac{1}{\gamma} d \eta(d_A e, e).
    \end{align}
    Where we denoted $T_\gamma=1+\frac{1}{\gamma}\star$ the Holst automorphism (Definition \ref{def:PCH}) and $\gamma \coloneqq 1 + \beta$. Hence, up to a boundary term and scaling, gH theory for the chosen parameters coincides with Palatini--Cartan--Holst theory.
\begin{proof}
    For the above choices of $\beta$ and $T$ we know that the action of generalised Holst theory is:
    \begin{align}
        S_{gH}{\left(\frac{1}{\beta}(\star^{-1} - \id), \beta \right)} = \frac{1}{2} \int_M \tr[F_A \wedge e \wedge e] + \frac{1}{\gamma} \eta(d_A e, d_A e).
    \end{align}
    One can rewrite the term $\eta(d_A e, d_A e)$ as:
    \begin{subequations}
    \begin{align}
        \eta(d_A e, d_A e) &= d \eta(d_A e, e) - \eta([F_A, e], e) = d \eta(d_A e, e) + \killing{F_A}{e \wedge e} \\
        &= d \eta(d_A e, e) + \tr [\star F_A \wedge e \wedge e].
    \end{align}
    \end{subequations}
    Thus one can rewrite the action as
    \begin{align}
        S_{gH}{\left(\frac{1}{\beta}(\star^{-1} - \id), \beta \right)} &= \frac{1}{2} \int_M \tr[F_A \wedge e \wedge e] + \frac{1}{\gamma} \left( d \eta(d_A e, e) + \tr [\star F_A \wedge e \wedge e] \right) \\
        &= \frac{1}{2} \int_M \tr\left[\left(F_A + \frac{1}{\gamma} \star F_A\right) \wedge e \wedge e \right] + \frac{1}{\gamma} d \eta(d_A e, e)
    \end{align}
    This proves the claim.
\end{proof}
\end{proposition}

\begin{remark}[On ``quadratic'' theories]
\label{rem:quadratic_torsion}
    Note that in passing from quadratic generalised Holst theory  and generalised Holst theory (Definitions \ref{def:ClassicalgH} and \ref{def:ClassicalgH}) we have dropped the identifier ``quadratic''. The action of qgH has a term proportional to $\eta(\tau, \tau)$ which justifies its name. Looking at Definition \ref{def:ClassicalgH} we see that gH features the term $\eta(d_A e, d_A e)$, which one could think of as a ``quadratic torsion'' term as well. However, the identity (cf. Proposition \ref{prop:gH_PCH})
    \begin{subequations}
    \begin{align}
        \eta(d_A e, d_A e) &= d \eta(d_A e, e) - \eta([F_A, e], e) = d \eta(d_A e, e) + \killing{F_A}{e \wedge e},
    \end{align}
    \end{subequations}
    shows that $\eta(d_A e, d_A e)$ contributes only with a topological and a boundary term. Hence it makes sense to consider generalised Holst theory as no longer characterised by ``quadratic torsion'', and the terminology we used reflects that.
\end{remark}

\begin{remark}[Palatini--Cartan--Holst theory from \qGR theory]\label{rem:gHlikePCH}
    Proposition \ref{prop:gH_PCH} has several interesting consequences. First of all, note that the chosen parameters belong to the intersection of the gravitational sector ($\beta T\not = \beta \id$) of \qGR  and the invertibility sector ($\beta T \not = - \id$). Linking this with Proposition \ref{prop:qBF_cleq_qgH} we can relate the action functional of \qGR and that of Palatini--Cartan--Holst theory (Definition \ref{def:PCH}), up to a boundary term, via classical equivalence:
    \begin{align}
        \left. S_{\qGR}{\left(\frac{1}{\beta}(\star^{-1} - \id), \beta \right)}\right|_C &= - \frac{1}{2} \int_M \underbrace{\tr\left[T_\gamma(F_A) \wedge e \wedge e \right]}_{\text{Palatini--Cartan--Holst}} + \underbrace{\frac{1}{\gamma} d \eta(d_A e, e)}_{\text{boundary term}},
    \end{align}
    where $C$ is defined as
    \begin{align}
        C \coloneqq \left\{(\AB, \BB) \in \FF_{\qGR} \Big| c = - (\id + \beta T)^{-1} F_A, \tau = - \frac{1}{(1+\beta)} d_A e \right\}.
    \end{align}
    Observe that the boundary term corresponds to the Nieh--Yan class; this, among other topological contributions, is discussed in \cite[eq. 2,3]{Rezende2009}.
\end{remark}

The following diagram summarises the \emph{classical} equivalence results obtained so far:

\begin{figure}[H]
    \[
    \xymatrix{
    &  \qGR \ar[dl]_{\genfrac{}{}{0pt}{}{}{(\id + \beta T) \neq 0}} \ar[dd]^{\genfrac{}{}{0pt}{}{(1 + \beta) \neq 0}{(\id + \beta T) \neq 0}} & & {\small S_{\qGR}=\int_M \paird{\BB}{F_\AB} + \frac{1}{2} \big( \paird{\BB}{\BB} + \beta \pairT{\BB}{\BB} \big)}  \\
        qgH \ar[dr]_{\beta \neq -1} &   & & {\small S_{qgH}= \int_M \frac{1}{2} \killing{(\id + \beta T)^{-1} F_A}{e \wedge e} - \eta(\tau, d_A e) - \frac{(1 + \beta)}{2} \eta(\tau, \tau)}\\
        & gH & & {\small S_{qH}=\frac{1}{2} \int_M \killing{(\id + \beta T)^{-1} F_{\widetilde{A}}}{\widetilde{e} \wedge \widetilde{e}} + \frac{1}{(1 + \beta)} \eta(d_{\widetilde{A}} \widetilde{e}, d_{\widetilde{A}} \widetilde{e})}
    }
    \]
\caption{Diagram of classically equivalent theories displaying the results of Proposition \ref{prop:qBF_cleq_qgH}. The classical equivalence between \qGR and gH is obtained via composition.} \label{fig:classical_diagram}
\end{figure}

\subsection{Classical boundary/Hamiltonian picture}
\label{subsec:qBFboundary}

In the following we review the classical boundary structure of \qGR. This is instrumental to obtaining a Hamiltonian version of the Lagrangian field theories we studied so far, and will give us direct information on the (local) degrees of freedom of the various models we consider and their symmetries, the study of which will be perfected in Section \ref{subsec:BVextensions}. 

The basis of the procedure adopted in the following section is the Kijowksi--Tulczijew construction \cite{KijTul1979} reviewed in Section \ref{subsubsec:KTconstruction}. It outputs a symplectic space of (boundary) fields (Definition \ref{def:GeometricPhase}), and a submanifold of ``Cauchy data'', defined by the field equations, which in turn defines the \emph{reduced phase space} (Definition \ref{def:ReducedPhase}) of the theory. {If the theory has a well-defined Cauchy problem, the submanifold is coisotropic, and is possibly specified by  a first-class constraint set. In certain cases, however, the restriction of the field equations to the boundary yield a constraint set that is not first-class. In that case, the reduced phase space cannot be obtained by means of coisotropic reduction alone, and the procedure is more involved.}

In Definition \ref{def:ClassicalTqBF}, \qGR theory was defined as a family of theories that depend on the choice of deformation automorphism $T$ and a real parameters $\beta$. In this section we want to understand how these choices influence the properties of \qGR by looking at the constraint algebra, through which we will infer the structure of the reduced phase space and learn about the (residual, local) degrees of freedom therein. For an application of these methods to gravity, see \cite{CattaneoSchiavina2018,CanepaCattaneoSchiavina2021,CanepaCattaneoTecchiolli}.

In particular we want to find out how these properties link to known theories of gravity. Thus at each step we will highlight what restrictions we impose on $T$ and $\beta$. In this investigation, the gravitational sector (see \ref{def:twisted_untwisted}) will be of particular interest. 

Recall that the classical space of fields of \qGR on a manifold $M$ is $\FF_{\qGR} \coloneqq \AC_P \times \Omega^2(M; \df^*)$, where $P \lra M$ is a $D$-bundle with $D$ integrating $\df$, which in all our examples is the trivial standard extension of $\gf \simeq \so(3,1)$ by an orthogonal module. When $\partial M\not=\emptyset$, the boundary inclusion map $\imath_\partial \colon \partial M \lra M$ induce the (boundary) principal bundle $P^\partial \coloneqq \imath_\partial^* P \to \partial M$.

\begin{proposition}
    For every $\beta\in \mathbb{R}$ and $T\in\Aut(\so(3,1))$, the geometric phase space of quadratically extended General Relativity is given by the symplectic manifold
    \begin{equation}
        \FF^\partial_{\qGR} \coloneqq \AC_{P^\partial} \times \Omega^2(\partial M; \df),
    \end{equation}
    with the exact symplectic form
    \begin{equation}
        \varpi^\partial_{\qGR} \coloneqq \int_{\partial M} \killing{\delta B}{\delta A} + \eta(\delta \tau, \delta e) + \killing{\delta c}{\delta S}.
    \end{equation}
\begin{proof}
    The space of pre-boundary fields $\widetilde{\FF}_{\qGR}$ is\footnote{In principle one should consider higher transverse jets of fields, but they drop out trivially.}
    \begin{align}
        \widetilde{\FF}_{\qGR} = \AC_{P^\partial} \times \Omega^2(\partial M, \df^*).
    \end{align}
    Denote the restricted fields in $\widetilde{\FF}_{\qGR}$ with the same symbols as in $\FF_{\qGR}$. The variation of $S_{\qGR}$ (compare Calculation \ref{cal:QExtendedEOM}) splits into a bulk term and a boundary term which can be seen as a $1$-form on $\widetilde{\FF}_{\qGR}$. This $1$-form takes the following form:
    \begin{equation}
        \widetilde{\alpha}_{\qGR} = \int_{\partial M} \killing{B}{\delta A} + \eta(\tau, \delta e) + \killing{c}{\delta S}.
    \end{equation}
    From $\widetilde{\alpha}_{\qGR}$ we deduce a pre-symplectic $2$-form on the pre-boundary fields $\widetilde{\FF}_{\qGR}$ via
    \begin{equation}
        \widetilde{\varpi}_{\qGR} \coloneqq \delta \widetilde{\alpha}_{\qGR} = \int_{\partial M} \killing{\delta B}{\delta A} + \eta(\delta \tau, \delta e) + \killing{\delta c}{\delta S}.
    \end{equation}
    One immediately sees that the kernel of the map $\widetilde{\varpi}_{\qGR}^\sharp \colon T\widetilde{\mathcal{F}}_{\qGR} \to T^*\widetilde{\mathcal{F}}_{\qGR}$, denoted for simplicity $\ker(\widetilde{\varpi}_{\qGR})$, is empty. Hence, the geometric phase space is the quotient $\mathcal{F}_{\qGR}^\partial \coloneqq \widetilde{\mathcal{F}}_{\qGR}/\mathrm{ker}(\widetilde{\varpi}_{\qGR})$, the quotient map $\pi\colon \widetilde{\mathcal{F}}_{\qGR} \to \mathcal{F}^\partial_{\qGR}$ is the identity, and the boundary $1$- and $2$-forms are simply
    \begin{subequations}
    \label{eq:BoundaryForms}
    \begin{align}
        \alpha^\partial_{\qGR} &= \int_{\partial M} \killing{B}{\delta A} + \eta(\tau, \delta e) + \killing{c}{\delta S}, \\
        \varpi^\partial_{\qGR} &= \int_{\partial M} \killing{\delta B}{\delta A} + \eta(\delta \tau, \delta e) + \killing{\delta c}{\delta S}.
    \end{align}
    \end{subequations}
\end{proof}
\end{proposition}

\begin{remark}
  It should be stressed again that the above result is completely independent of $T$ and $\beta$. Thus both the gravitational and topological sector have the same (exact) symplectic space of (classical) boundary fields
  \begin{equation}
    \FF^\partial_{\qGR} = \AC_{P^\partial} \times \Omega^2(\partial M; \df); \qquad \varpi^\partial_{\qGR} = \delta \alpha^\partial_{\qGR},
  \end{equation}
  with $\alpha^\partial_{\qGR}$ and $\varpi^\partial_{\qGR}$ as in Equation \eqref{eq:BoundaryForms}.
\end{remark}

Let us now inspect the constraints imposed on the boundary data of \qGR which let us define the reduced phase space of Definition \ref{def:ReducedPhase}. {By restricting the field equations \eqref{eq:EOMsub} to the boundary, one can implement six constraints} using the Lagrange multipliers $\ab = (\theta, \epsilon, \sigma) \in \Omega^0(\partial M; \df)$ and $\HB = (\alpha, \rho, \zeta) \in \Omega^1(\partial M; \df)$ as
\begin{alignat}{2}
    H_\theta \coloneqq& \ \int_{\partial M} \killing{d_A B - e \wedge \tau - [c, S]}{\theta}, &&\qquad K_\alpha \coloneqq \ \int_{\partial M} \killing{(\id + \beta T) B + d_A S + \frac{1}{2} e \wedge e}{\alpha},  \\
    I_\epsilon \coloneqq& \ \int_{\partial M} \eta(d_A \tau, \epsilon) + \killing{c}{e \wedge \epsilon}, &&\qquad L_\rho \coloneqq \ \int_{\partial M} \eta(d_A e + (1+\beta) \tau, \rho), \\
    J_\sigma \coloneqq& \ \int_{\partial M} \killing{d_A c}{\sigma}, &&\qquad M_\zeta \coloneqq \ \int_{\partial M} \killing{F_A + (\id + \beta T) c}{\zeta}.
\end{alignat}

These are local functionals on the symplectic space $\FF^\partial_{\qGR}$. We will determine their mutual dependence, and compute their Poisson brackets to establish which of them are \emph{first class}\footnote{Recall that every first class constraint eliminates \emph{two} local degrees of freedom \cite{Dirac1958, HenneauxTeitelboim1992}. This is because the underlying procedure is symplectic reduction, which is a combination of a restriction to a coisotropic submanifold and a quotient by its characteristic foliation.} and how the choices for $T$ and $\beta$ influence the outcome. 

\begin{calculation}
\label{calc:ShiftS}
    Whenever {the invertibility condition $\beta T  \neq -\id$ holds}, the constraint functions $J_\sigma$ and $M_\zeta$ are dependent, upon choosing $\zeta = -d_A \left( (\id + \beta T)^{-1} \sigma \right)$.
\begin{proof}
Let us set $\zeta = - d_A \left((\id+\beta T)^{-1}\sigma \right)$ and keep in mind that $d_A F_A = 0$:
  \begin{align}
    M_{- d_A \left((\id+\beta T)^{-1}\sigma \right)} &= \ \int_{\partial M} \killing{F_A + (\id + \beta T) c}{- d_A \left((\id+\beta T)^{-1}\sigma \right)} \\
    &= \int_{\partial M} - d\killing{(\id+\beta T)^{-1} F_A + c}{\sigma} + \killing{d_A c}{\sigma} \\
    &\simeq \int_{\partial M} \killing{d_A c}{\sigma} = J_\sigma.
  \end{align}
  In the last line we dropped the boundary term by using that $\partial M = \emptyset$ and used the commutation of $d_A$ with automorphisms.
\end{proof}
\end{calculation}

\begin{definition}
    Let $c_1, c_2\in \mathbb{N}$ respectively denote the number of independent local first- and second-class constraints, and let $d$ be the number of local degrees of freedom in $\FF^\partial$.
    A theory is said to be topological if $d = 2c_1 + c_2$.
\end{definition}

\begin{theorem}
\label{theo:ConstraintAlgebraQBF}
  {In the invertibility sector} the Poisson brackets of the constraint functions $H_{\theta}, I_\epsilon, K_\alpha, L_\rho, M_\zeta$ are given by:
  \begingroup
  \allowdisplaybreaks
  \begin{alignat}{3}
    \{K_\alpha, K_{\overline\alpha}\} &= 0, \quad && \{K_\alpha, L_{\rho}\} = \int_{\partial M} \killing{\beta (\id  - T)\alpha}{e \wedge \rho}, \quad &&\{K_\alpha, M_{\zeta}\} = 0, \\
    \{K_\alpha, H_{\theta}\} &= K_{[\theta, \alpha]}, \quad && \{K_\alpha, I_{\epsilon}\} = L_{[\alpha, \epsilon]} + \int_{\partial M} \killing{ \beta (T - \id) \alpha}{(\tau \wedge \epsilon)} , \quad &&~ \\
    ~& \quad && \{L_\rho, L_{\overline\rho}\} = 0, \quad &&\{L_\rho, M_{\zeta}\} = 0, \\
    \{L_\rho, H_{\theta}\} &= L_{[\theta, \rho]}, \quad && \{L_\rho, I_{\epsilon}\} = M_{\rho \wedge \epsilon} + \int_{\partial M} \killing{ \beta (\id - T) c}{\rho \wedge \epsilon}, \quad &&~ \\
    \{M_\zeta, M_{\overline\zeta}\} &= 0, \quad &&\{M_\zeta, H_{\theta}\} = M_{[\theta, \zeta]}, \quad && \{M_\zeta, I_{\epsilon}\} = 0,\\
    \{H_{\theta}, H_{\overline \theta}\} &= H_{[\theta, \overline \theta]}, \quad &&\{H_{\theta}, I_{\epsilon}\} = I_{[\theta, \epsilon]}, \quad && \{I_{\epsilon}, I_{\overline\epsilon}\} = 0.
  \end{alignat}
  \endgroup
\begin{proof}
    See Appendix \ref{App:CalcBdyConstr}.
\end{proof}
\end{theorem}

\begin{remark}[On first and second class constraint sets]\label{rmk:firstvssecondclass}
    One can immediately see that the above relations do not define a Poisson subalgebra for all choices of $T$ and $\beta$ due to the functional form of $\{K_\alpha, L_{\rho}\}$, $\{K_\alpha, I_{\epsilon}\}$ and $\{L_\rho, I_{\epsilon}\}$ for arbitrary $T$, which means they are not all first-class. The failure of being involutive of some of these constraints (i.e.\ their second-class nature) is controlled by $\beta(\mathrm{id} - T)\in\mathrm{Aut}(\gf)$. {Hence, the constraint set is first class iff $T$ and $\beta$ are in the topological sector $(\beta T = \beta \id)$, in which case the theory reduces to a $BF$ theory (cf.\ Definition \ref{def:ClassicalBFBB}).}

    A theory whose constraint subalgebra is not first class cannot have a well-posed Cauchy problem.\footnote{Our construction---following Kijowski--Tulczyjew---aims at constructing coisotropic submanifolds of the (symplectic) space of boundary fields. These are often given by first-class constraint sets therein.  Dirac’s classical construction \cite{Dirac1958} instead defines constraints in a different ambient space, so also second class constraints are allowed in systems with well-posed Cauchy problems.} This is because existence and uniqueness of solutions corresponds to the projection of the EL locus to the boundary being Lagrangian. Since by construction such projection lies in the constrained set, it being not coisotropic means that it cannot contain any Lagrangian submanifold, and thus in particular the Cauchy problem is not well defined. {An example of this scenario within gravitational theories is given by Palatini--Cartan theory with a null boundary \cite{CanepaCattaneoTecchiolli}.}
    
    {Another way of thinking about this is as follows. Let $M=\Sigma \times [0,1]$, so that $\FF^\partial = \overline{\FF}_{\Sigma\times\{0\}} \times \FF_{\Sigma\times\{1\}}$, where the overline indicates the symplectic manifold with opposite simplectic structure. Denote by $L_M\doteq \pi(EL)$ the projection to the space of boundary fields of the set of solutions to the field equations, as above. (In regular Hamiltonian theories, this is the graph of the Hamiltonian evolution at finite time.) If $L_M$ is Lagrangian (in the sense that its tangent spaces are isotropic with isotropic complement), then $\pi_1(L_M)\subset \overline{\FF}_{\Sigma\times\{0\}}$ must be coisotropic \cite[Lemma A.16]{CattaneoMnevWernli} (see \cite{CMRCorfu} for more details).}
    
    As we will see later on, this problem is resolved ``effectively'', in the sense that elimination of auxiliary fields (BV pushforward or classical equivalence) takes us to theories with better-behaved evolution equations. Indeed, we can further distinguish two classes of examples {within the invertibility sector}:
    \begin{enumerate}
    \item For $\beta \not= -1$ \qGR theories are (BV-)equivalent to generalised Holst theories (Proposition \ref{prop:qBFreduceALL}), which have a well defined Cauchy problem. Such equivalence goes through by enforcing the equations of motion for only a part of the fields (specifically for $B$, $c$ and $\tau$). If one restricts the constraint functions to solutions of said partial EoM's, most of the constraint functions vanish, and the remaining two (namely $I_\epsilon$ and $H_\theta$), coincide with those reported in Equations \eqref{eq:constraintsgH} (cf.\ with \cite{CattaneoSchiavina2018}).
    \item For $\beta = -1$ the above procedure fails because $\tau$ is eliminated from the constraint function $L_\rho$, which is therefore not automatically satisfied, even if we enforced EoM's for $c$ and $B$. In this case three constraint functions remain (namely $I_\epsilon$, $H_\theta$ and $L_\rho$), and the theory reduces to ``half-shell'' Palatini Cartan Holst theory (cf. Proposition \ref{prop:qgH_PC}).
    \end{enumerate}
    Ultimately this suggests that neither classical equivalence nor BV pushforward necessarily preserve the boundary structure of a field theory, an observation that is consistent with \cite{SimaoCattaneoSchiavina}. Field theories are realisations of ``a cohomology'' of physical data. Some realisations may be better suited for certain applications, such as solving field equations, and others for different applications, such as quantisation.
    
\end{remark}

The Lagrange multiplier $\ab = (\theta, \epsilon, \sigma) \in \Omega^0(\partial M; \df)$ has $6+4+6 = 16$ local degrees of freedom, while the field $\HB = (\alpha, \rho, \zeta) \in \Omega^1(\partial M; \df)$ has $18+12+18=48$. The following is a direct consequence, which justifies the nomenclature:

\begin{theorem}
\label{theo:Top_if_untwisted}
In the topological sector, \qGR is topological.
\begin{proof}
  We begin with $\beta = -1$ and $T = \id$. In this case the generators can be written as
  \begin{align}
      J_\ab \coloneqq \int_{\partial M} \paird{d_\AB \BB}{\ab}, \qquad L_\HB \coloneqq \int_{\partial M} \paird{F_\AB}{\HB}.
  \end{align}
  Those are just the generators for a $BF$ theory. We move on to choices of $T, \beta$ for which $\beta T = \beta \id$. This holds for $\beta = 0$ and $T = \id$ respectively. In this case the generators are
  \begin{align}
      J_\ab \coloneqq \int_{\partial M} \paird{d_\AB \BB}{\ab}, \qquad L_\HB \coloneqq \int_{\partial M} \paird{F_\AB + (1+\beta) \BB}{\HB},
  \end{align}
  which are the generators for a $BF+ (1+\beta)BB$ theory. Since $(1+\beta) \neq 0$ one can further see that $J_\ab$ and $L_\HB$ are dependent upon choosing $\HB = - \frac{1}{(1+\beta)} d_\AB \ab$. Hence, $\HB$ and $\ab$ have 54 collective local degrees of freedom, but $16$ of those generate redundant first-class constraints. Thus, the number of independent constraints is 48 and \qGR is topological.
\end{proof}
\end{theorem}

{Theorem \ref{theo:Top_if_untwisted} is not surprising, since we already knew that in the topological sector \qGR theory reduces to a $BF + \Lambda BB$ (topological) theory. This tells us that \qGR theory cannot directly be taken as a model for gravity, as it should instead have two local degrees of freedom. The classical equivalences of Proposition \ref{prop:qBF_cleq_qgH}, together with \ref{prop:gH_PCH} prompt us to study \qGR theory in the gravitational sector.

Note that results derived in Section \ref{subsec:qBFclassicalEQ} (see Figure \ref{fig:classical_diagram}) always required the invertibility condition $\beta T  \neq -\id$. This rules out the special case $\beta \neq 0$, $T = - \frac{1}{\beta} \id$ (the complement of the invertibility sector),  which yields a $BF$ theory with an additional $\beta$-dependent} deformation term (see Proposition \ref{prop:T_beta_id}, below).  {From Lemma \ref{lem:sectorrelations} $\beta=-1$ differentiates the topological from the gravitational sector \emph{within} the invertibility sector and, indeed, from Theorem \ref{theo:ConstraintAlgebraQBF} one can directly infer that the constraint algebra for this special case differs from that of for pure $BF$ theory only by the constraint $L_\rho$, which is sensitive to $\beta=-1$.} {To sum up the above, one has the following special choice of parameters for \qGR:}

\begin{proposition}[Non-{invertible} \qGR theory]
\label{prop:T_beta_id}
    {In the complement of the invertibility sector, $\beta T = -\id$,} we have
    \begin{align}
            S_{\qGR}\left[ - \frac{1}{\beta} \id, \beta \right] &= \int_M \killing{B}{F_A} + \eta(\tau, d_A e) + \killing{c}{\frac{1}{2} e \wedge e + d_A S} + \frac{ (1 + \beta)}{2} \eta(\tau, \tau) \\
            &= \int_M \paird{\BB}{F_\AB} + \frac{ (1 + \beta)}{2} \eta(\tau, \tau).
    \end{align}
    The resulting theory {lies in the topological sector whenever $\beta = -1$, and in the gravitational sector otherwise}.
\end{proposition}

{Following Remark \ref{rmk:firstvssecondclass}, the constraint set is bound to be only second-class in the gravitational sector of \qGR.}
\begin{remark}[Equivalent boundary analysis]\label{rmk:classicalboundaryreplacement}
    {In the following we will not tackle a direct analysis of \qGR in the gravitational sector, but we will analyse the boundary structure of gH and qgH theories instead, which are both classically equivalent\footnote{While this strongly suggests that our results should hold for \qGR, in principle one needs to check if the classical equivalence preserves the boundary structure. {A stronger hint of this will be provided by the discovery that the classical equivalence can actually be promoted to a full BV equivalence, see Section \ref{subsec:MainResults}.}} to \qGR in virtue of Proposition \ref{prop:qBF_cleq_qgH}.

    Note that Proposition \eqref{prop:T_beta_id} covers all combinations of $\beta$ and $T$ in the gravitational sector that are not in the parameter ranges of Proposition \ref{prop:qBF_cleq_qgH} (which instead assumed the invertibility sector, i.e. $\beta T \not= -\id$). Thus, from now on, we restrict ourselves to the invertible gravitational sector by fixing $\beta T  \neq -\id$ and requiring $\beta T\not = \beta \id$, which lets us use said classical equivalences. If $\beta = -1$ one \emph{needs} to use qgH as the equivalent model, while for $\beta \neq -1$ one can use the simpler gH theory by rule of Proposition \ref{prop:qBF_cleq_qgH}.}
\end{remark}

We have: 

\begin{proposition}\label{prop:qgH_PC}
    For $T \neq \id - \star^{-1}$ and $\beta = -1$, qgH reduces to so-called ``half-shell Palatini--Cartan theory''
    \begin{align}
        S_{qgH}\left[T=\id - \star^{-1},\beta=-1 \right] = \int_M \underbrace{\frac{1}{2} \tr[F_A \wedge e \wedge e]}_{\text{Palatini--Cartan}} \underbrace{- \eta(\tau, d_A e)}_{\text{zero torsion}}.
    \end{align}
    Moreover, for any $T\not = \id$ and $\beta = -1$, the space of classical boundary fields as well as the degrees of freedom of quadratic generalised Holst are the same as that of half-shell Palatini--Cartan theory. 
\end{proposition}
\begin{proof}
    {
    Note that for any $T \neq \id$ and $\beta = -1$ the action of $qgH$ looks as follows:
    \begin{align}
        S_{qgH}(T, -1) :&= \int_M \frac{1}{2} \killing{(\id - T)^{-1} F_A}{e \wedge e} - \eta(\tau, d_A e).
    \end{align}
    Denoting the invertible automorphism $(\id - T)^{-1}$ by $\widetilde{T}$ this yields the following pre-boundary two-form:
    \begin{align}
        \widetilde{\varpi}^\partial_{qgH} = \int_{\partial M} \killing{\widetilde{T} [\delta A]}{\delta e \wedge e} - \eta(\delta \tau, \delta e)
    \end{align}
    Its kernel is given by
    \begin{align}
        (X_\tau) &= \star \left[(X_{\widetilde{T}[A]}) \wedge e \right], \\
        (X_e) &= 0.
    \end{align}
    Thus one can fix $A$ using the following vertical vector field:
    \begin{align}
        \XX = (X_{\widetilde{T}[A]}) \ddelta{}{\widetilde{T}[A]} + \star \left[(X_{\widetilde{T}[A]}) \wedge e \right] \ddelta{}{\tau}.
    \end{align}
    Flowing along $\XX$ we can choose a background connection $\underline{A}$ to obtain an explicit map to the space of boundary fields:
    \begin{align}
        \pi_M \colon
        \begin{cases}
            \textbf{t} = \tau + \star \left[ \widetilde{T}[\underline{A} - A] \wedge e \right], \\
            \textbf{e} = e.
        \end{cases}
    \end{align}
    From here on one can follow the proofs of \cite[Theorem 4.28]{CattaneoSchiavina2018} and \cite[Proposition 4.29]{CattaneoSchiavina2018} mostly verbatim while carrying in mind that $\tau$ is an element of $\Omega^2(M; \VV)$ in our case, and hence the internal Hodge star $\star$ in the above formulas is required.
    }
\end{proof}

\begin{remark}
{Half-shell PC theory has been investigated in \cite[Section 4.3]{CattaneoSchiavina2018}. It is however somewhat pathological in that the projection to the boundary of the space of solutions of the Euler--Lagrange field equations is not Lagrangian (which is compatible with the observation that the constraint set is not coisotropic).}
\end{remark}

{We now look at the invertible, gravitational sector where additionally $\beta \not=-1$.
Note that, in this sector \qGR is classically equivalent to generalised Holst theory (Proposition \ref{prop:qBF_cleq_qgH}).} As anticipated in Remark \ref{rmk:classicalboundaryreplacement}, we directly analyse qgH theory. The proof of the following Proposition is analogous to the proof of \cite[Theorem 4.6]{CattaneoSchiavina2018} with minor modifications due to $\beta$ and $T$.

\begin{proposition}[See {\cite[Theorem 4.6, Corollary 4.23 and Remark 4.27]{CattaneoSchiavina2018}}]\label{prop:HolstHodgeFour}
    The geometric phase space of generalised Holst theory {in the invertible, gravitational sector with $\beta\not=-1$} is the symplectic manifold given by the fibre bundle
    \begin{equation}
        \FF^\partial_{gH} \lra \Omega^1_{nd}(\partial M; \VV^\partial),
    \end{equation}
    where $\VV^\partial\coloneqq\imath^*\VV$ denotes the induced vector bundle on $\partial M$, and the generic fibre over $e \in \Omega_{nd}^1(\partial M; \VV^\partial)$ is given by the reduced space of connections $\AC_{P^\partial}^{red} \coloneqq \AC_{P^\partial} /_\sim$ with respect to the equivalence relation
    \begin{equation}
        A \sim A^\prime \quad \Longleftrightarrow \quad A - A^\prime \in \ker(e \wedge \cdot), \qquad e \wedge \cdot \ \colon \Omega^1(\partial M; \textstyle{\bigwedge}^2 \VV^\partial) \lra \Omega^2(\partial M; \textstyle{\bigwedge}^3 \VV^\partial).
    \end{equation}
    We denote equivalence classes by $[A]_e \in \AC_{P^\partial}^{red}$. The symplectic form is given by
    \begin{equation}
        \varpi^\partial_{gH} = \int_{\partial M} 2 \killing{\left[\frac{1}{1 + \beta} \id + (\id + \beta T)^{-1} \right] \delta \overline{A}}{\overline{e} \wedge \delta \overline{e}}.
    \end{equation}
    The surjective submersion $\pi_{gH} \colon \FF_{gH} \lra \FF^\partial_{gH}$ has the explicit form:
    \begin{align}
        \pi_{gH} \colon
        \begin{cases}
            \overline{e} = e, \\
            \overline{A} = [A]_e,
        \end{cases}
    \end{align}
    where $\overline{e} \in \Omega^1_{nd}(\partial M; \VV^\partial)$ and $[A]_e \in \AC_{P^\partial}^{red}$.
    Moreover, the field equations of generalised Holst theory define the following constraints on the space of pre-boundary fields, basic w.r.t. $\pi_{gH}$ and first-class:
    \begin{subequations}
    \mathtoolsset{showonlyrefs = false}
    \label{eq:constraintsgH}
        \begin{align}
            \widetilde L_\alpha &= \int_{\partial M} \killing{\left((\id + \beta T)^{-1} - \frac{1}{1+\beta} \id \right) \alpha}{e \wedge d_{\widetilde A} e}, \\
            \widetilde J_\mu &= \int_{\partial M} \killing{\left((\id + \beta T)^{-1} - \frac{1}{1+\beta} \id \right) F_A}{e \wedge \mu},
        \end{align}
    \end{subequations}
    \mathtoolsset{showonlyrefs = true}
    where $\alpha \in \Omega^0(\partial M; \bigwedge^2 \VV^\partial)$, $\mu \in \Omega^0(\partial M; \VV^\partial)$ and $\widetilde{A}$ is defined as in \cite[Lemma 4.12]{CattaneoSchiavina2018}.
\end{proposition}

\begin{remark}
    {In the invertible gravitational sector with $\beta\not=-1$, the constraints \eqref{eq:constraintsgH} define a coisotropic submanifold in $(\FF^\partial_{gH},\varpi^\partial_{gH})$. Their projection to $\FF^\partial_{gH}$ are first class constraints eliminating $10$ local degrees of freedom. {Consequently} the reduced classical phase space of generalised Holst theory in this sector has exactly $2$ local degrees of freedom.
    Indeed, this constraint algebra can be linked to that of four-dimensional diffeomorphisms\footnote{See \cite[Remark 4.24]{CattaneoSchiavina2018} where this is explicitly discussed for the Palatini--Cartan--Holst action which corresponds to $T = \frac{1}{\beta}(\star^{-1} - \id)$}.}
\end{remark}

\section{BV Equivalences}
\label{sec:BVtheories}

The goal of this section is to formulate the BV extensions of \qGR for both the topological and the gravitational sector as well as the BV extensions of (quadratic) generalised Holst theory. We will look for general contractible pairs (Definition \ref{def:generalContractible}) such that, for suitable parameters $T$ and $\beta$, we can prove that \qGR and (the BV extensions of) qgH and gH are BV equivalent. This clarifies the relation between \qGR and known theories of gravity, extending Proposition \ref{prop:gH_PCH} to the BV framework.\\

In Section \ref{subsec:BVextensions} we will investigate the symmetries of \qGR in both the topological and gravitational sector. We will find that in the topological sector some symmetries are either reducible or entirely redundant. The first case can be dealt with, in the BV formalism, by introducing higher-order ghost fields to the BV extension of a classical theory. The second case indicates that a symmetry can be fully represented by other symmetries (at least) on-shell.

We use these findings in Section \ref{subsec:BVextensions} where we formulate the BV extension of \qGR in the topological and gravitational sector, basing our constructions on the classical analysis from Section \ref{sec:ClassicalAnalysis}. We then move on to identify general contractible pairs in Section \ref{subsubsec:GeneralContractible}. In Section \ref{subsubsec:WeakEQs} these pairs will be eliminated in order to prove several BV equivalences. By doing so, we also draw direct parallels between the gravitational sector of \qGR and known theories of gravity.

\subsection{Symmetries}
\label{subsec:BVextensions}

Quadratically extended GR can be endowed with a variety of local symmetries, as we explain now. 
\begin{calculation}
\label{calc:SymmetriesqBFindep}
  For any choice of $T, \beta$, the action of quadratically extended GR given in \eqref{eq:qAction} is preserved (up to boundary terms\footnote{The contraction $\imath_{\QQ} \delta S$ vanishes modulo boundary terms. This is Noether's definition of local symmetries.}) under the following tangent distributions, for every $\zeta \in \Omega^1[1](M; \gf^*)$, $\sigma \in \Omega^0[1](M; \gf^*)$, $\theta \in \Omega^0[1](M; \gf)$ and $\xi \in \Gamma[1](TM)$:
  \begin{itemize}
    \item Two gauge symmetries generated by $\theta$ and $\sigma$:
    \begin{alignat}{2}
      (\QQ \AB)_\theta:&
      \begin{cases}
        (\QQ A)_\theta  = d_A \theta \\
        (\QQ e)_\theta  = [\theta, e] \\
        (\QQ S)_\theta  = [\theta, S]
      \end{cases},
      \qquad &&
      (\QQ \BB)_\theta :
      \begin{cases}
        (\QQ B)_\theta  = [\theta, B] \\
        (\QQ \tau)_\theta  = [\theta, \tau] \\
        (\QQ c)_\theta  = [\theta, c]
      \end{cases}, \\
      (\QQ \AB)_\sigma:&
      \begin{cases}
        (\QQ A)_\sigma  = 0 \\
        (\QQ e)_\sigma  = 0 \\
        (\QQ S)_\sigma  = d_A \sigma
      \end{cases},
      \qquad &&
      (\QQ \BB)_\sigma :
      \begin{cases}
        (\QQ B)_\sigma  = [c, \sigma] \\
        (\QQ \tau)_\sigma  = 0 \\
        (\QQ c)_\sigma  = 0
      \end{cases}.
    \end{alignat}

    \item One shift symmetry generated by $\zeta$:
    \begin{equation}
      (\QQ \AB)_\zeta:
      \begin{cases}
        (\QQ A)_\zeta  = 0 \\
        (\QQ e)_\zeta  = 0 \\
        (\QQ S)_\zeta  = (\id+ \beta T) \zeta
      \end{cases}
      and \qquad
      (\QQ \BB)_\zeta :
      \begin{cases}
        (\QQ B)_\zeta  = d_A \zeta \\
        (\QQ \tau)_\zeta  = 0 \\
        (\QQ c)_\zeta  = 0
      \end{cases}.
    \end{equation}

    \item One diffeomorphism symmetry generated by $\xi$:
    \begin{equation}
      (\QQ \AB)_\xi:
      \begin{cases}
        (\QQ A)_\zeta  = \imath_\xi F_A \\
        (\QQ e)_\zeta  = \lie{\xi}^A e \\
        (\QQ S)_\zeta  = \lie{\xi}^A S
      \end{cases}
      and \qquad
      (\QQ \BB)_\xi :
      \begin{cases}
        (\QQ B)_\zeta  = \lie{\xi}^A B \\
        (\QQ \tau)_\zeta  = \lie{\xi}^A \tau \\
        (\QQ c)_\zeta  = \lie{\xi}^A c
      \end{cases}.
    \end{equation}
  \end{itemize}
  \begin{remark}
      Here we denote, e.g., $(\QQ\AB)_\theta$ to indicate the component of the vector field $\QQ$ in the direction of a the field $\AB$, generated by $\theta$ seen as the element $(\theta,0,0)$ in $\Omega^0[1](M; \gf^*)\times \Omega^0[1](M; \gf)\times \Gamma[1](TM)$ (which can be easily seen is a Lie algebra action). We will also use the notation $\QQ_\theta$ for the image of $(\theta,0,0)$ under the Lie algebra action, and $(\QQ\AB)_\theta\equiv \QQ_\theta(\AB)$. In the physics literature, this is often denoted as $\delta_\theta\AB$.
  \end{remark}
\begin{proof}
  First note that the action of a general symmetry vector field $\QQ$ on the action of \qGR yields $\lie{\QQ} S_{\qGR}(T, \beta) = \imath_\QQ \delta S_{\qGR}(T, \beta)$, thus
  \begin{align}
    &~ \imath_\QQ \int_M \killing{\delta B}{F_A} - \killing{B}{d_A \delta A} + \eta(\delta \tau, d_A e) + \eta(\tau, [\delta A, e] - d_A \delta e) + (1 + \beta) \eta(\delta \tau, \tau) \\
    &+ \killing{\delta c}{ (\id + \beta T) B + \frac{1}{2} e \wedge e + d_A S} + \killing{c}{\delta e \wedge e + [\delta A, S] - d_A \delta S} + \killing{ (\id + \beta T) c}{\delta B}.
  \end{align}
  Starting with the first gauge symmetry $\QQ_\theta$ yields:
  \begin{align}
    & \ \int_M \killing{[\theta, B]}{F_A} - \killing{B}{[F_A, \theta]} + \eta([\theta, \tau], d_A e) + \eta(\tau, [\theta, d_A e]) + (1 + \beta) \eta([\theta, \tau], \tau) \\
    &+ \killing{[\theta, c]}{(\id + \beta T) B + \frac{1}{2} e \wedge e + d_A S} + \killing{c}{[\theta, e] \wedge e + [\theta, d_A S]} + \killing{ (\id + \beta T) c}{[\theta, B]}.
  \end{align}
  By invariance of $\eta$ and $\killing{\cdot}{\cdot}$ this vanishes identically. For $\QQ_\sigma$ one obtains:
  \begin{align}
    &~ \int_M \killing{[c, \sigma]}{F_A} - \killing{c}{[F_A, \sigma]} + \killing{ (\id + \beta T) c}{[c, \sigma]} \equiv 0.
  \end{align}
  Again the invariance of the pairing was used. For the shift symmetry $\QQ_\zeta$ one obtains
  \begin{align}
    &~ \int_M \killing{d_A \zeta}{F_A} - \killing{c}{ (\id+ \beta T) d_A \zeta} + \killing{ (\id + \beta T) c}{d_A \zeta} \simeq 0.
  \end{align}
  In the last equation $d_A F_A = 0$ was used and the boundary term $d \killing{F_A}{\zeta}$ was dropped. The last symmetry is the diffeomorphism symmetry $\QQ_\xi$:
  \begin{align}
    &~ \int_M \killing{\lie{\xi}^A B}{F_A} + \killing{B}{\lie{\xi}^A F_A} + \eta(\lie{\xi}^A \tau, d_A e) + \eta(\tau, \lie{\xi}^A d_A e) +  (1 + \beta) \eta(\lie{\xi}^A \tau, \tau) \\
    &+ \killing{\lie{\xi}^A c}{ (\id + \beta T) B + \frac{1}{2} e \wedge e + d_A S} + \killing{c}{\lie{\xi}^A e \wedge e + \lie{\xi}^A d_A S} + \killing{ (\id + \beta T) c}{\lie{\xi}^A B} \\
    =& \ \int_M \lie{\xi} \Bigg( \killing{B}{F_A} + \eta(\tau, d_A e) + \frac{ (1 + \beta)}{2} \eta(\tau, \tau) + \killing{c}{ (\id + \beta T) B + \frac{1}{2} e \wedge e + d_A S} \Bigg) \simeq 0.
  \end{align}
  Since $\lie{\xi} L_{\qGR} = d \circ \imath_\xi L_{\qGR}$, it contributes only boundary terms.
\end{proof}
\end{calculation}

\begin{remark}
\label{rem:SecondGhost}
    The ghost fields $\zeta \in \Omega^1[1](M; \gf^*)$, $\sigma \in \Omega^0[1](M; \gf^*)$ are themselves subject to a local symmetry, which is a consequence of the (Lie algebra) gauge action having stabilisers. (Another way of saying this is that the symmetry is reducible.) Within the BV framework, this can be dealt with by introducing a second order ghost $\nu \in \Omega^0[2](M; \gf^*)$ (of ghost degree two) and extending the symmetry vector field $\QQ$ by $(\QQ\sigma)_\nu =  (\id + \beta T) \nu$ and $(\QQ \zeta)_\nu = d_A \nu$.
\end{remark}

In Proposition \eqref{prop:HolstHodgeFour} we observed that, for $T, \beta$ in the gravitational sector, \qGR is not a topological theory. Meanwhile in the topological sector (see Table \ref{Tab:sectors}) there exist additional symmetries known from $BF$ (Definition \ref{def:ClassicalBFBB}) type theories:

\begin{lemma}
\label{calc:ShiftIfKilling}
  In the topological sector quadratically extended GR enjoys the following additional symmetries, where $\epsilon \in \Omega^0[1](M; \VV)$, $\rho \in \Omega^1[1](M; \VV)$ and $\alpha \in \Omega^1[1](M; \gf)$
  \begin{itemize}
    \item Two shift symmetries generated by $\rho$ and $\alpha$:
    \begin{alignat}{2}
      (\QQ\AB)_\rho :&
      \begin{cases}
        (\QQ A)_\rho = 0 \\
        (\QQ e)_\rho =  (1+\beta) \rho \\
        (\QQ S)_\rho = 0
      \end{cases},
      \qquad &&
      (\QQ\BB)_\rho :
      \begin{cases}
        (\QQ B)_\rho = - \rho \wedge e \\
        (\QQ \tau)_\rho = d_A \rho \\
        (\QQ c)_\rho = 0
      \end{cases}, \\
      (\QQ\AB)_\alpha :&
      \begin{cases}
        (\QQ A)_\alpha =  (1 + \beta) \alpha \\
        (\QQ e)_\alpha = 0 \\
        (\QQ S)_\alpha = 0
      \end{cases},
      \qquad &&
      (\QQ\BB)_\alpha :
      \begin{cases}
        (\QQ B)_\alpha = - [\alpha, S] \\
        (\QQ \tau)_\alpha = - [\alpha, e] \\
        (\QQ c)_\alpha = d_A \alpha
      \end{cases}.
    \end{alignat}

    \item One gauge symmetry generated by $\epsilon$:
    \begin{equation}
      (\QQ\AB)_\epsilon :
      \begin{cases}
        (\QQ A)_\epsilon = 0 \\
        (\QQ e)_\epsilon = d_A \epsilon \\
        (\QQ S)_\epsilon = e \wedge \epsilon
      \end{cases}
      and \qquad
      (\QQ\BB)_\epsilon :
      \begin{cases}
        (\QQ B)_\epsilon = - \tau \wedge \epsilon  \\
        (\QQ \tau)_\epsilon = [c, \epsilon] \\
        (\QQ c)_\epsilon = 0
      \end{cases}.
    \end{equation}
  \end{itemize}
\begin{proof}
  First note that in the topological sector $(\id + \beta T) = (1+\beta) \id$ always holds. Applying a general symmetry vector field $\QQ$ to the action of \qGR returns:
  \begin{align}
    &~ \imath_\QQ \int_M \killing{\delta B}{F_A} - \killing{B}{d_A \delta A} + \eta(\delta \tau, d_A e) + \eta(\tau, [\delta A, e] - d_A \delta e) +  (1 + \beta) \eta(\delta \tau, \tau) \\
    &+ \killing{\delta c}{ (\id + \beta T) B + \frac{1}{2} e \wedge e + d_A S} + \killing{c}{\delta e \wedge e + [\delta A, S] - d_A \delta S} + \killing{ (\id + \beta T) c}{\delta B}.
  \end{align}
  Start by verifying the shift symmetry $\QQ_\rho$:
  \begin{align}
    &~ \int_M - \killing{\rho \wedge e}{F_A} + \eta(d_A \rho, d_A e) -  (1 + \beta) \eta(\tau, d_A \rho) \\
    &+ \killing{c}{-  (\id + \beta T) (\rho \wedge e) +  (\id + \beta T) (\rho \wedge e)} +  (1 + \beta) \eta(d_A \rho, \tau) \\
    =& \ \int_M - \killing{\rho \wedge e}{F_A} + \eta(d_A \rho, d_A e) \simeq \int_M - \killing{\rho \wedge e}{F_A} - \eta([F_A, \rho], e) = 0.
  \end{align}
  In the last equality one uses $\eta([F_A, \rho], e) = - \killing{F_A}{\rho \wedge e}$. For the second shift symmetry $\QQ_\alpha$ one obtains:
  \begingroup
  \allowdisplaybreaks
  \begin{align}
    &~ \int_M \killing{- [\alpha, S]}{F_A} - \killing{B}{ (\id + \beta T) d_A \alpha} + \eta(- [\alpha, e], d_A e) + \eta(\tau, [ (\id + \beta T) \alpha, e]) \\
    &-  (1 + \beta) \eta([\alpha, e], \tau) + \killing{d_A \alpha}{ (\id + \beta T) B + \frac{1}{2} e \wedge e + d_A S} + \killing{ (\id + \beta T) c}{- [\alpha, S] + [\alpha, S]} \\
    =& \ \int_M \killing{- [\alpha, S]}{F_A} + \eta(- [\alpha, e], d_A e) + \killing{d_A \alpha}{\frac{1}{2} e \wedge e + d_A S} \\
    \simeq& \ \int_M -\killing{[\alpha, S]}{F_A} + \killing{\alpha}{e \wedge d_A e} - \killing{\alpha}{e \wedge d_A e + [F_A, S]} = 0.
  \end{align}
  \endgroup
  At last verify the gauge symmetry $\QQ_\epsilon$ for which one needs to keep in mind that it is represented by a $0$-form:
  \begin{align}
    &~ \int_M - \killing{\tau \wedge \epsilon}{F_A} + \eta([c, \epsilon], d_A e) - \eta(\tau, [F_A, \epsilon]) \\
    &+ \killing{c}{-  (1 + \beta) (\tau \wedge \epsilon) + d_A \epsilon \wedge e - d_A (e \wedge \epsilon)} +  (1 + \beta) \eta([c, \epsilon], \tau) = 0.
  \end{align}
  The above terms all vanish due to the relation of the Minkowski metric and the Killing form.
\end{proof}
\end{lemma}

\begin{remark}
\label{rem:KillingSecondGhosts}
  In the topological sector, the ghosts $\epsilon$ and $\rho$ are subject to an additional symmetry, which can be similarly taken care of by adding a degree $2$ field, $\gamma \in \Omega^0[2](M; \VV)$. The symmetry vector field is then extended as $(\QQ \epsilon)_\gamma = (1+\beta)\gamma$ and $(\QQ \rho)_\gamma = d_A \gamma$.
  Furthermore, the ghost fields $\theta$ and $\alpha$ are also subject to a symmetry, which can be encoded by second-degree ghost field $f \in \Omega^0[2](M; \gf)$. The cohomological vector field $\QQ$ can be extended by $(\QQ \theta)_f = (\id+\beta T) f$ and $(\QQ \alpha)_f = d_A f$. (Note that this redundancy is expected since we have seen that all three generators $H_{\theta}$, $I_\epsilon$ and $J_\sigma$ are redundant in that scenario (cf.\ Theorem \ref{theo:Top_if_untwisted}).)
\end{remark}

The following result will allow us to disregard the component $\QQ_\xi$ of $\QQ$ under certain conditions:

\begin{lemma}
\label{lem:KillingNoDiffeo}
  Whenever $T, \beta$ are in the topological sector the component $\QQ_\xi$ of $\QQ$ can be represented as a combination of $\QQ_\alpha, \QQ_\rho, \QQ_\zeta, \QQ_\sigma, \QQ_\epsilon$ and a symmetry vector field that vanishes identically on-shell. Thus $\QQ_\xi$ is redundant.
\begin{proof}
  We choose $\alpha = - \imath_\xi c$, $\rho = - \imath_\xi \tau$, $\zeta = - \imath_\xi B$ as well as $\sigma = - \imath_\xi S$ and $\epsilon = - \imath_\xi e$. Let for the following $\overline\QQ \coloneqq \QQ_\alpha + \QQ_\rho + \QQ_\zeta + \QQ_\sigma + \QQ_\epsilon$. The given choices yield
  \begingroup
  \allowdisplaybreaks\mathtoolsset{showonlyrefs = false}
  \begin{subequations}\label{eq:recreate_diffeo}
  \begin{align}
    \overline\QQ A &= \imath_\xi F_A - \imath_\xi (F_A + (\id+\beta T) c), \\
    \overline\QQ e &=  \lie{\xi}^A e - \imath_\xi (d_A e + (1+\beta) \tau), \\
    \overline\QQ S &= \lie{\xi}^A S - \imath_\xi ((\id + \beta T) B + d_A S + \tfrac{1}{2} e \wedge e), \\
    \overline\QQ B &= \lie{\xi}^A B - \imath_\xi (d_A B - e \wedge \tau - [c, S]), \\
    \overline\QQ \tau &= \lie{\xi}^A \tau - \imath_\xi (d_A \tau - [c, e]), \\
    \overline\QQ c &= \lie{\xi}^A c - \imath_\xi (d_A c).
  \end{align}
  \end{subequations}\mathtoolsset{showonlyrefs = true}
  \endgroup
  The additional terms that do not belong to $\QQ_\xi$ are directly dependent on the field equations (compare Calculation \ref{cal:QExtendedEOM}). Thus define a new symmetry vector field $\QQ_\chi$, for $\chi \in \Gamma[1](TM)$, acting as
  \begin{subequations}
  \begin{alignat}{3}
    (\QQ A)_\chi &= \imath_\chi (F_A + (\id+\beta T) c), \qquad && (\QQ e)_\chi = \imath_\chi (d_A e + (1+\beta) \tau), \\
    (\QQ c)_\chi &= \imath_\chi (d_A c), \qquad && (\QQ \tau)_\chi = \imath_\chi (d_A \tau - [c, e]), \\ 
    (\QQ S)_\chi &= \imath_\chi ((\id + \beta T) B + d_A S + \tfrac{1}{2} e \wedge e),  \qquad && (\QQ B)_\chi = \imath_\chi (d_A B - e \wedge \tau - [c, S]).
  \end{alignat}
  \end{subequations}
  The vector field $\QQ_\chi$ is a symmetry of $S_{\qGR}$ since it can be represented as a combination of $\QQ_\xi$ and the other symmetries as seen in Equations \eqref{eq:recreate_diffeo}. However, due to the field equations \eqref{cal:QExtendedEOM}, all components of $\QQ_\chi$ vanish \emph{identically} on-shell, making it a trivial symmetry (cf. Remark \ref{rmk:non-trivialsymmetries}).
  {But we have seen in \eqref{eq:recreate_diffeo} that one can represent $\QQ_\xi$ as $\overline\QQ - \QQ_\chi$ by choosing $\chi = \xi$. Since $\QQ_\xi$ can be represented as a combination of $\QQ_\alpha, \QQ_\rho, \QQ_\zeta, \QQ_\sigma, \QQ_\epsilon$ and the trivial symmetry $\QQ_\chi$ it is redundant.}
\end{proof}
\end{lemma}

\subsubsection{BV extension of \qGR in the topological sector}

In the following pages we provide the BV extension of \qGR in the topological sector. To this end we first show that not only the action, as shown in Remark \ref{rem:untwisted_simpler}, but also the symmetries of topological \qGR allow for a compact writing which significantly compresses the form of the BV extension.

\begin{remark}
    Lemma \ref{lem:KillingNoDiffeo} shows that $\xi \in \Gamma[1](TM)$ generates a redundant symmetry in the topological sector. This result does not extend to the gravitational sector since two of the symmetries, generated by $\rho$ and $\alpha$, do not exist here and thus the component $\QQ_\xi$ cannot be recovered by means of other symmetries as it was done in Lemma \ref{lem:KillingNoDiffeo}.
\end{remark}

\begin{remark}
\label{rem:CollectiveSymmetries}
  Following the concise writing of $S_{\qGR}$ in the topological sector (see Remark \ref{rem:untwisted_simpler}) one might expect that a similar rewriting can be done for the symmetries.\footnote{Note that this is obvious once one realises that topological \qGR returns a $BF$ theory.} Thus define $(\theta, \epsilon, \sigma) =: \ab \in \Omega^0[1](M; \df)$ and $(\alpha, \rho, \zeta) =: \HB \in \Omega^1[1](M; \df)$ and express their symmetry vector fields as:
  \begin{itemize}
    \item The shift symmetry:
    \begin{equation}
      \QQ_\HB \AB = (1+\beta) \HB, \qquad  \QQ_\HB \BB = d_\AB \HB.
    \end{equation}
    \item The gauge symmetry:
    \begin{equation}
      \QQ_\ab \AB = d_\AB \ab, \qquad  \QQ_\ab \BB = [\ab, \BB].
    \end{equation}
  \end{itemize}
  One can easily check that this aligns with the definitions of the individual symmetries from Calculation \ref{calc:SymmetriesqBFindep} and Calculation \ref{calc:ShiftIfKilling}. It further allows for a formulation of the BV extension of topological \qGR using superfields similar to the results of \cite[Section  5]{Cattaneo2001}, and inspired by the seminal work of \cite{AKSZ}.
\end{remark}

Following Calculation \ref{calc:ShiftIfKilling} we inferred in Remark \ref{rem:CollectiveSymmetries} that not only the action (see Remark \ref{rem:untwisted_simpler}) but also the symmetries of \qGR in the topological sector can be written in a very compact form. This allows for the following simple BV extension which is just the known extension of $BF$ theory (see \ref{def:ClassicalBFBB}) and is stated here mainly for completeness:

\begin{definition}[Topological \qGR]
\label{def:BV_qGR_top}
    For $T, \beta$ {within the topological invertible sector\footnote{{By this we mean the intersection of the two.}}} we call {standard BV theory for topological \qGR} a $4$-dimensional manifold $M$ , or simply {topological \qGR}, the BV theory
    \begin{align}
        \FC_{\qGR}^{\beta} \coloneqq (\FF_{\qGR}, \Omega_{\qGR} \coloneqq \delta \alpha_{\qGR}, Q_{\qGR}^{\beta}, S_{\qGR}^{ \beta})
    \end{align}
    where
  \begin{equation}
    \FF_{\qGR} = T^*[-1]\left( \AC_P \times \Omega^2(M; \df^*) \times \Omega^0[1](M; \df) \times \Omega^1[1](M; \df) \times \Omega^0[2](M; \df^*) \right).
  \end{equation}
  We denote the additional fields by
  \begin{align}
    (f, \gamma, \nu) = \fb \in \Omega^0[2](M; \df), \quad \AB^\dagger &\in \Omega^3[-1](M; \df^*), \quad \BB^\dagger \in \Omega^2[-1](M; \df),\\
    \ab^\dagger \in \Omega^4[-2](M; \df^*), \quad \HB^\dagger &\in \Omega^3[-2](M; \df^*), \quad \fb^\dagger \in \Omega^0[-3](M; \df),
  \end{align}
  where the superscript $\dagger$ indicates the corresponding fields in the cotangent fibre. Using this notation the symplectic BV $2$-form is given by
  \begin{align}
    \Omega_{\qGR} \coloneqq& \int_M \paird{\delta \AB}{\delta \AB^\dagger} + \paird{\delta \BB}{\delta \BB^\dagger} + \paird{\delta \ab}{\delta \ab^\dagger} + \paird{\delta \HB}{\delta \HB^\dagger} + \paird{\delta \fb}{\delta \fb^\dagger},
  \end{align}
  the BV action is given by
  \begin{align}
  \label{eq:qBFBVAction}
    S_{\qGR}^{\beta} \coloneqq &\int_M \paird{\BB}{F_\AB} + \frac{(1+\beta)}{2} \paird{\BB}{\BB} + \paird{\BB^\dagger}{[\ab, \BB] + d_\AB \HB - [\fb, \BB^\dagger]} \\
    &+ \paird{\AB^\dagger}{d_\AB \ab + (1+\beta) \HB} + \paird{\HB^\dagger}{[\ab, \HB] + d_\AB \fb} \\
    &+ \paird{\ab^\dagger}{\frac{1}{2} [\ab, \ab] + (1+\beta) \fb} + \paird{\fb^\dagger}{[\ab, \fb]},
  \end{align}
  and the cohomological BV vector field $Q_{\qGR}^{\beta}$ defined by $\imath_{Q_{\qGR}^{\beta}} \Omega_{\qGR} = \delta S_{\qGR}^{\beta}$ can be represented by its components as
  \begingroup
  \allowdisplaybreaks
  \begin{subequations}
  \begin{align}
    Q_{\qGR}^{\beta} \BB &= [\ab, \BB] + d_\AB \HB - [\fb, \BB^\dagger],\\
    Q_{\qGR}^{\beta} {\BB^\dagger} &= F_\AB + (1+\beta) \BB + [\ab, \BB^\dagger],\\
    Q_{\qGR}^{\beta} \AB &= d_\AB \ab + (1+\beta) \HB, \\
    Q_{\qGR}^{\beta} {\AB^\dagger} &= d_\AB \BB + [\BB^\dagger, \HB] + [\ab, \AB^\dagger] - [\fb, \HB^\dagger], \\
    Q_{\qGR}^{\beta} \ab &= \frac{1}{2} [\ab, \ab] + (1+\beta) \fb,\\
    Q_{\qGR}^{\beta} {\ab^\dagger} &= d_\AB \AB^\dagger + [\BB^\dagger, \BB] + [\HB^\dagger, \HB] + [\ab, \ab^\dagger] + [\fb^\dagger, \fb],\\
    Q_{\qGR}^{\beta} \HB &= [\ab, \HB] + d_\AB \fb,\\
    Q_{\qGR}^{\beta} {\HB^\dagger} &= d_\AB \BB^\dagger + (1+\beta) \AB^\dagger + [\ab, \HB^\dagger], \\
    Q_{\qGR}^{\beta} {\fb} &= [\ab, \fb], \\
    Q_{\qGR}^{\beta} {\fb^\dagger} &= d_\AB \HB^\dagger + [\ab, \fb^\dagger] + (1+\beta) \ab^\dagger + \frac{1}{2} [\BB^\dagger, \BB^\dagger].
  \end{align}
  \end{subequations}
  \endgroup
\end{definition}

Notice again that such theories have been shown to allow an even more concise formulation in terms of superfields (see \cite[Section 5]{Cattaneo2001}). Since this work aims to investigate the relation of \qGR to theories of gravity, we will not focus on the topological sector any further.

\subsubsection{BV extension of \qGR in the gravitational sector}

In Calculation \ref{calc:SymmetriesqBFindep} it was shown that in the gravitational sector there exist four symmetries that are not redundant. This allows for the formulation of the BV extension of \qGR in the gravitational sector.

\begin{definition}[Gravitational \qGR]
\label{def:BV_qBF}
    For $T, \beta$ {within the invertible gravitational sector}, we call standard BV theory for gravitational \qGR on a $4$-dimensional manifold $M$, or simply gravitational \qGR, the BV theory
    \begin{align}
        \FC_{\qGR}^{T, \beta} \coloneqq (\FF_{\qGR}, \Omega_{\qGR} \coloneqq \delta \alpha_{\qGR}, Q_{\qGR}^{T, \beta}, S_{\qGR}^{T, \beta})
    \end{align}
    where
    \begin{align}
        \FF_{\qGR} \coloneqq& \ T^*[-1]\Big( \AC_P \times \Omega^2(M; \df^*) \times \Omega^0[1](M; \gf) \times \Omega^1[1](M; \gf^*) \\
        &\qquad \qquad \times \Omega^0[1](M; \gf^*) \times \Gamma[1](TM) \times \Omega^0[2](M; \gf^*) \Big),
    \end{align}
    with the fields in the base denoted by $(A, e, S, B, \tau, c, a, \zeta, \sigma, \xi, \nu)$ and their correspondents in the cotangent fibre by
    \begin{subequations}
    \begin{alignat}{3}
        (A^\dagger, e^\dagger, S^\dagger) &\in \Omega^3[-1](M; \df^*), &&\qquad (B^\dagger, \tau^\dagger, c^\dagger) \in \Omega^2[-1](M; \df), &&\qquad a^\dagger \in \Omega^4[-2](M; \gf^*),\\
        \zeta^\dagger &\in \Omega^3[-2](M; \gf), &&\qquad \sigma^\dagger \in \Omega^4[-2](M; \gf), &&\qquad \xi^\dagger \in \Omega^4[-2](M; T^*M), \\
        \nu^\dagger &\in \Omega^4[-3](M; \gf).
    \end{alignat}
    \end{subequations}
    And further
    \begin{align}
        \Omega_{\qGR} \coloneqq& \int_M \paird{\delta \AB}{\delta \AB^\dagger} + \paird{\delta \BB}{\delta \BB^\dagger} + \killing{\delta a}{\delta a^\dagger} + \killing{\delta \zeta}{\delta \zeta^\dagger} \\
        & + \killing{\delta \sigma}{\delta \sigma^\dagger}  + \killing{\delta \nu}{\delta \nu^\dagger} + \tr[\imath_{\delta \delta\xi} \xi^\dagger] ,
    \end{align}
    \begingroup
    \allowdisplaybreaks
    \begin{subequations}
    \begin{align}
        S_{\qGR}^{T, \beta} \coloneqq &\int_M \killing{B}{F_A} + \killing{c}{(\id+\beta T) B + \frac{1}{2} e \wedge e + d_A S} + \eta(\tau, d_A e) \\
        &+ \frac{(1+\beta)}{2} \eta(\tau, \tau) + \killing{A^\dagger}{d_A a + \imath_\xi F_A} + \eta(e^\dagger, [a, e] + \lie{\xi}^A e) \\
        &+ \killing{c^\dagger}{[a,c] + \lie{\xi}^A c} + \killing{S^\dagger}{d_A \sigma + [a, S] +  (\id + \beta T) \zeta + \lie{\xi}^A S} \\
        &+ \eta(\tau^\dagger, [a, \tau] + \lie{\xi}^A \tau) + \killing{B^\dagger}{[a,B] + [c, \sigma] + d_A \zeta + [B^\dagger, \nu] + \lie{\xi}^A B} \\
        &+ \killing{\zeta^\dagger}{[a, \zeta] + d_A \nu + \lie{\xi}^A \zeta} + \killing{a^\dagger}{\frac{1}{2} [a,a] - \frac{1}{2} \imath_\xi^2 F_A} \\
        &+ \killing{\sigma^\dagger}{[a,\sigma] +  (\id + \beta T) \nu + \lie{\xi}^A \sigma} + \killing{\nu^\dagger}{[a, \nu] + \lie{\xi}^A \nu} + \frac{1}{2} \tr[\imath_{[\xi, \xi]} \xi^\dagger].
    \end{align}
    \end{subequations}
    \endgroup
    The cohomological BV vector field $Q_{\qGR}^{T, \beta}$ defined by $\imath_{Q_{\qGR}^{T, \beta}} \Omega_{\qGR} = \delta S_{\qGR}^{T, \beta}$ can be explicitly expressed as:
    \begingroup
    \allowdisplaybreaks
    \begin{subequations}
    \begin{align}
        Q_{\qGR}^{T, \beta} {A} &= d_A a + \imath_\xi F_A, \\
        Q_{\qGR}^{T, \beta} {A^\dagger} &= d_A B  - [c, S] - e \wedge \tau + [\zeta, B^\dagger] + [a, A^\dagger] + [\sigma, S^\dagger] + [\nu, \zeta^\dagger] - \imath_\xi (e^\dagger \wedge e) - d_A \imath_\xi A^\dagger \\
        &- \imath_\xi \big(\tau^\dagger \wedge \tau + [c^\dagger, c] + [S^\dagger, S] + [B^\dagger, B] + [\zeta^\dagger, \zeta]\big) - \frac{1}{2} d_A \imath_\xi^2 a^\dagger - \imath_\xi [\sigma^\dagger, \sigma] - \imath_\xi [\nu^\dagger, \nu], \nonumber \\
        Q_{\qGR}^{T, \beta} {e} &= [a, e] + \lie{\xi}^A e, \\
        Q_{\qGR}^{T, \beta} {e^\dagger} &= d_A \tau - [c, e] + [a, e^\dagger] + \lie{\xi}^A e^\dagger, \\
        Q_{\qGR}^{T, \beta} {S} &= d_A \sigma + [a, S] +  (\id + \beta T) \zeta + \lie{\xi}^A S, \\
        Q_{\qGR}^{T, \beta} {S^\dagger} &= d_A c + [a, S^\dagger] + \lie{\xi}^A S^\dagger, \\
        Q_{\qGR}^{T, \beta} {B} &= [a,B] + [c, \sigma] + d_A \zeta + [B^\dagger, \nu] + \lie{\xi}^A B, \\
        Q_{\qGR}^{T, \beta} {B^\dagger} &= F_A + (\id + \beta T) c + [a, B^\dagger] + \lie{\xi}^A B^\dagger, \\
        Q_{\qGR}^{T, \beta} {\tau} &= [a, \tau] + \lie{\xi}^A \tau, \\
        Q_{\qGR}^{T, \beta} {\tau^\dagger} &= d_A e + (1 + \beta) \tau + [a, \tau^\dagger] + \lie{\xi}^A \tau^\dagger, \\
        Q_{\qGR}^{T, \beta} {c} &= [a,c] + \lie{\xi}^A c, \\
        Q_{\qGR}^{T, \beta} {c^\dagger} &= (\id + \beta T) B + \frac{1}{2} e \wedge e + d_A S + [a, c^\dagger] + [B^\dagger, \sigma] + \lie{\xi}^A c^\dagger, \\
        Q_{\qGR}^{T, \beta} {a} &= \frac{1}{2} [a,a] - \frac{1}{2} \imath_\xi^2 F_A, \\
        Q_{\qGR}^{T, \beta} {a^\dagger} &= [c, c^\dagger] + [B, B^\dagger] + d_A A^\dagger - [S, S^\dagger] + [\zeta, \zeta^\dagger] - [a^\dagger, a] + [\sigma^\dagger, \sigma] + \tau^\dagger \wedge \tau + e^\dagger \wedge e + [\nu, \nu^\dagger], \\
        Q_{\qGR}^{T, \beta} {\sigma} &= [a,\sigma] + (\id + \beta T) \nu + \lie{\xi}^A \sigma, \\
        Q_{\qGR}^{T, \beta} {\sigma^\dagger} &= [c, B^\dagger] + d_A S^\dagger + [a, \sigma^\dagger] + \lie{\xi}^A \sigma^\dagger, \\
        Q_{\qGR}^{T, \beta} {\zeta} &= [a, \zeta] + d_A \nu + \lie{\xi}^A \zeta, \\
        Q_{\qGR}^{T, \beta} {\zeta^\dagger} &= d_A B^\dagger + (\id + \beta T) S^\dagger + [a, \zeta^\dagger] + \lie{\xi}^A \zeta^\dagger, \\
        Q_{\qGR}^{T, \beta} \xi &= \frac{1}{2} [\xi, \xi], \\
        Q_{\qGR}^{T, \beta} \xi^\dagger_\bullet &= - e_\bullet^\dagger \wedge d_A e - d_A e^\dagger \wedge e_\bullet - A^\dagger_\bullet F_A - \imath_\xi a_\bullet^\dagger \wedge F_A - c_\bullet^\dagger \wedge d_A c - d_A c^\dagger \wedge c_\bullet - S_\bullet^\dagger \wedge d_A S \\
        &- d_A S^\dagger \wedge S_\bullet - \tau_\bullet^\dagger \wedge d_A \tau - d_A \tau^\dagger \wedge \tau_\bullet - B_\bullet^\dagger \wedge d_A B - d_A B^\dagger \wedge B_\bullet - \zeta_\bullet^\dagger \wedge d_A \zeta  \nonumber \\
        &- d_A \zeta^\dagger \wedge \zeta_\bullet - \sigma_\bullet^\dagger \wedge d_A \sigma - d_A \sigma^\dagger \wedge \sigma_\bullet - \nu_\bullet^\dagger \wedge d_A \nu - d_A \nu^\dagger \wedge \nu_\bullet + \partial_\bullet \xi^a \xi_a^\dagger + \partial_a \xi^a \xi^\dagger_\bullet \nonumber ,\\
        Q_{\qGR}^{T, \beta} {\nu} &= [a, \nu] + \lie{\xi}^A \nu, \\
        Q_{\qGR}^{T, \beta} {\nu^\dagger} &= d_A \zeta^\dagger + [a, \nu^\dagger] + (\id + \beta T) \sigma^\dagger + \lie{\xi}^A \nu^\dagger.
    \end{align}
    \end{subequations}
    \endgroup
    In $Q_{\qGR}^{T, \beta} \xi^\dagger_\bullet$ the index $\bullet$ is used to highlight that $\xi^\dagger$ is a $1$-form with values in densities.
\end{definition}

\subsubsection{BV extension of gravitational generalised Holst theory}
\label{subsubsec:BVgH}

We state here the BV extension of gravitational generalised Holst theory since it will be pivotal in linking gravitational \qGR to known theories of gravity. Note that qgH as well as the topological sector of gH will not receive the same treatment since they are of minor interest at this point.

\begin{definition}[Gravitational Generalised Holst Theory]
\label{def:BV_gH}
    For $T, \beta$ within the gravitational sector and $M$ a $4$-dimensional manifold we call {standard BV theory for gravitational gH}
    the data
    \begin{align}
        \FC_{gH}^{T, \beta} \coloneqq (\FF_{gH}, \Omega_{gH} \coloneqq \delta \alpha_{gH}, Q_{gH}^{T, \beta}, S_{gH}^{T, \beta}),
    \end{align}
    where
    \begin{align}
        \FF_{gH} \coloneqq& \ T^*[-1]\bigg( \Omega^1(M; \VV) \times \AC(M) \times \Omega^0[1](M; \gf) \times \Gamma[1](TM)\bigg),
    \end{align}
    with the fields in the base denoted by $(A, e, a, \xi)$ and their correspondents in the cotangent fibre by
    \begin{alignat}{4}
        A^\dagger &\in \Omega^3[-1](M; \gf), &&\qquad e^\dagger \in \Omega^3[-1](M; \VV), &&\qquad a^\dagger \in \Omega^4[-2](M; \gf^*), &&\qquad \xi^\dagger \in \Omega^4[-2](M; T^*M).
    \end{alignat}
    The BV two-form is given by
    \begin{align}
        \Omega_{gH} \coloneqq& \int_M \killing{\delta A}{\delta A^\dagger} + \eta(\delta e, \delta e^\dagger) + \killing{\delta a}{\delta a^\dagger} + \tr[\imath_{\delta \xi} \xi^\dagger],
    \end{align}
    and BV action of $\FC_{gH}^{T, \beta}$ by
    \begingroup
    \allowdisplaybreaks
    \begin{subequations}
    \begin{align}
      S_{gH}^{T, \beta} \coloneqq &\int_M \frac{1}{2}\left(\killing{(\id + \beta T)^{-1} F_A}{e \wedge e} + \frac{1}{(1+\beta)} \eta(d_A e, d_A e) \right) + \killing{A^\dagger}{d_A a + \imath_\xi F_A} \\
      &+ \eta(e^\dagger, [a, e] + \lie{\xi}^A e) + \killing{a^\dagger}{\frac{1}{2} [a,a] - \frac{1}{2} \imath_\xi^2 F_A} + \frac{1}{2} \tr[\imath_{[\xi, \xi]} \xi^\dagger],
    \end{align}
    \end{subequations}
    \endgroup
    The cohomological BV vector field $Q_{\qGR}^{T, \beta}$ defined by $\imath_{Q_{\qGR}^{T, \beta}} \Omega_{\qGR} = \delta S_{\qGR}^{T, \beta}$ can be explicitly expressed as:
    \begingroup
    \allowdisplaybreaks
    \begin{subequations}
    \begin{align}
        Q_{\qGR}^{T, \beta} {A} &= d_A a + \imath_\xi F_A, \\
        Q_{\qGR}^{T, \beta} {A^\dagger} &= \left((\id + \beta T)^{-1} - \frac{1}{1+\beta} \id \right) (d_A e \wedge e) + [a, A^\dagger] - d_A \circ \imath_\xi A^\dagger - \imath_\xi [e^\dagger, e] - \frac{1}{2} d_A \circ \imath_\xi^2 a^\dagger, \\
        Q_{\qGR}^{T, \beta} {e} &= [a, e] + \lie{\xi}^A e, \\
        Q_{\qGR}^{T, \beta} {e^\dagger} &= \left(\left((\id + \beta T)^{-1} - \frac{1}{1+\beta} \id \right) F_A \right) \wedge e + [a, e^\dagger] + \lie{\xi}^A e^\dagger, \\
        Q_{\qGR}^{T, \beta} {a} &= \frac{1}{2} [a,a] - \frac{1}{2} \imath_\xi^2 F_A, \\
        Q_{\qGR}^{T, \beta} {a^\dagger} &= d_A A^\dagger - [a^\dagger, a] + [e^\dagger, e], \\
        Q_{\qGR}^{T, \beta} \xi &= \frac{1}{2} [\xi, \xi], \\
        Q_{\qGR}^{T, \beta} \xi^\dagger_\bullet &= - e_\bullet^\dagger \wedge d_A e - d_A e^\dagger \wedge e_\bullet - A^\dagger_\bullet F_A - \imath_\xi a_\bullet^\dagger \wedge F_A + \partial_\bullet \xi^a \xi_a^\dagger + \partial_a \xi^a \xi^\dagger_\bullet \nonumber.
    \end{align}
    \end{subequations}
    \endgroup
    In $Q_{gH}^{T, \beta} \xi^\dagger_\bullet$ the index $\bullet$ is used to highlight that $\xi^\dagger$ is a $1$-form with values in densities.
\end{definition}

\begin{remark}[Topological Generalised Holst Theory]
\label{rem:BV_gH_top}
    As discussed, we will not state the full data for the standard BV theory for topological generalised Holst theory and instead restrict to stating the BV action:
    \begingroup
    \allowdisplaybreaks
    \begin{subequations}
    \begin{align}
        S_{gH}^{\beta} \coloneqq &\int_M \frac{1}{2(1+\beta)} \left(\killing{F_A}{e \wedge e} +\eta(d_A e, d_A e) \right) + \killing{A^\dagger}{d_A a + \alpha} \\
        &+ \eta(e^\dagger, d_A \epsilon + [a, e] + (1+\beta) \rho) + \killing{\alpha^\dagger}{[a, \alpha] + d_A f} \\
        &+ \eta(\rho^\dagger, [a, \rho] + [\alpha, \epsilon] + d_A \gamma - [f, e]) + \killing{a^\dagger}{\frac{1}{2} [a,a] + (1+\beta) f} \\
        &+ \eta(\epsilon^\dagger, [a, \epsilon] + \gamma) + \killing{f^\dagger}{[a,f]} + \eta(\gamma^\dagger, [a, \gamma] + [f, \epsilon]).
    \end{align}
    \end{subequations}
    \endgroup
\end{remark}

\subsection{BV equivalence: Main results}
\label{subsec:MainResults}

Over the following pages we will state the core results of this work. We first identify general contractible pairs among the fields of the BV extensions of topological and gravitational \qGR in Section \ref{subsubsec:GeneralContractible}. We then provide two possible paths for BV equivalences, based opn their elimination. This is carried out in Section \ref{subsubsec:WeakEQs} where we finally link (under certain parameter choices) gravitational \qGR to Palatini--Cartan--Holst gravity.

\subsubsection{General Contractible Pairs}
\label{subsubsec:GeneralContractible}

In this section we will identify general contractible pairs within the BV extensions of both topological and gravitational \qGR. Since these fields contribute trivially to the cohomology, they can be eliminated to give rise to simpler theories, described with different but quasi-isomorphic cochain complexes within the BV setting.

Note ahead that for topological \qGR the goal is not to identify \emph{every} general contractible pair, since some of them are not pertinent for our investigation, which instead aims to relate \qGR to known theories of gravity.

The first result will clarify what has already been foreshadowed by the lack of field equations for the field $S$ (compare Equation \eqref{eq:EOMsub}):

\begin{lemma}
\label{lem:ContractibleClusterS}
  In the invertibility sector $(\beta T \neq \id)$, $(\{B, c, S^\dagger, \zeta, \sigma^\dagger, \nu\},\{c^\dagger, B^\dagger, \zeta^\dagger, S, \nu^\dagger, \sigma\})$ is a general contractible pair.
\begin{proof}
  Starting with the topological sector of \qGR, where one can set $T = \id$ without loss of generality (see Remark \ref{rem:untwisted_simpler}), the chosen pairs yield:
  \begingroup
  \allowdisplaybreaks
  \begin{align}
    Q_{\qGR}^{\beta} c^\dagger|_{c^\dagger = 0} &= (1+\beta) B + \frac{1}{2} e \wedge e + d_A S + [B^\dagger, \sigma] + \epsilon \wedge \tau^\dagger, \\
    Q_{\qGR}^{\beta} B^\dagger|_{B^\dagger = 0} &= F_A + (1+\beta) c, \\
    Q_{\qGR}^{\beta} \zeta^\dagger|_{\zeta^\dagger = 0} &= d_A B^\dagger + (1 + \beta) S^\dagger, \\
    Q_{\qGR}^{\beta} S|_{S = 0} &= d_A \sigma + e \wedge \epsilon + (1 + \beta) \zeta, \\
    Q_{\qGR}^{\beta} \nu^\dagger|_{\nu^\dagger = 0} &= d_A \zeta^\dagger + (1 + \beta) \sigma^\dagger, \\
    Q_{\qGR}^{\beta} \sigma|_{\sigma = 0} &= (1 + \beta) \nu.
  \end{align}
  \endgroup
  When set to zero these yield equations that can be uniquely solved for $B, c, S^\dagger, \zeta, \sigma^\dagger$ and $\nu$ respectively:
  \begin{align}
    B &= - \frac{1}{(1+\beta)} \left(\frac{1}{2} e \wedge e + d_A S + [B^\dagger, \sigma] + \epsilon \wedge \tau^\dagger\right), \\
    c &= - \frac{1}{(1+\beta)} F_A , \\
    S^\dagger &= - \frac{1}{(1+\beta)} d_A B^\dagger , \\
    \zeta &= - \frac{1}{(1+\beta)} (d_A \sigma + e \wedge \epsilon) , \\
    \sigma^\dagger &= - \frac{1}{(1+\beta)} d_A \zeta^\dagger, \\
    \nu &= 0.
  \end{align}
  Meanwhile for the gravitational sector of \qGR the same procedure yields:
  \begin{align}
    Q_{\qGR}^{T, \beta} c^\dagger|_{c^\dagger = 0} &= (\id + \beta T) B + \frac{1}{2} e \wedge e + d_A S + [B^\dagger, \sigma], \\
    Q_{\qGR}^{T, \beta} B^\dagger|_{B^\dagger = 0} &= F_A + (\id + \beta T) c, \\
    Q_{\qGR}^{T, \beta} \zeta^\dagger|_{\zeta^\dagger = 0} &= d_A B^\dagger + (\id + \beta T) S^\dagger, \\
    Q_{\qGR}^{T, \beta} S|_{S = 0} &= d_A \sigma + (\id + \beta T) \zeta, \\
    Q_{\qGR}^{T, \beta} \nu^\dagger|_{\nu^\dagger = 0} &= d_A \zeta^\dagger + (\id + \beta T) \sigma^\dagger, \\
    Q_{\qGR}^{T, \beta} \sigma|_{\sigma = 0} &= (\id + \beta T) \nu.
  \end{align}
  Setting the above terms to zero yields equations that can be uniquely solved, by fixing $S^\dagger, \zeta, \sigma^\dagger$ and $\nu$ respectively as
  \begin{align}
    B &= - (\id + \beta T)^{-1} \left(\frac{1}{2} e \wedge e + d_A S + [B^\dagger, \sigma]\right), \\
    c &= - (\id + \beta T)^{-1} F_A , \\
    S^\dagger &= - (\id + \beta T)^{-1} d_A B^\dagger, \\
    \zeta &= - (\id + \beta T)^{-1} d_A \sigma , \\
    \sigma^\dagger &= - (\id + \beta T)^{-1} d_A \zeta^\dagger, \\
    \nu &= 0.
  \end{align}
  Following Definition \ref{def:generalContractible} we obtain the claimed general contractible pair.
\end{proof}
\end{lemma}

\begin{lemma}
\label{lem:ContractibleClusterTau}
  Whenever $\beta \neq -1$, $(\tau, \tau^\dagger)$ is a general contractible pair.
\begin{proof}
  Again we start with topological \qGR, where
  \begin{align}
      Q_{\qGR}^{\beta}|_{\tau^\dagger = 0} = d_A e + (1+\beta) \tau + [B^\dagger, \epsilon].
  \end{align}
  Setting this to zero yields an equation that can be uniquely solved for $\tau$ with the given restrictions to $\beta$:
  \begin{align}
      \tau &= - \frac{1}{(1+\beta)} (d_A e + [B^\dagger, \epsilon]).
  \end{align}
  A similar situation is found for gravitational \qGR, where
  \begin{align}
      Q_{\qGR}^{T, \beta}|_{\tau^\dagger = 0} = d_A e + (1+\beta) \tau.
  \end{align}
  Setting this expression to zero allows us to uniquely solve for $\tau$:
  \begin{align}
      \tau &= - \frac{1}{(1+\beta)} d_A e.
  \end{align}
  This proves the claim.
\end{proof}
\end{lemma}

\begin{remark}
\label{rem:SmallPairCompoundPair}
  From Lemma \ref{lem:ContractibleClusterS} and Lemma \ref{lem:ContractibleClusterTau} one can infer two important results:
  \begin{itemize}
    \item In the invertibility sector the following general contractible pairs exist independent from each other: $(B, c^\dagger)$, $(c, B^\dagger)$, $(S^\dagger, \zeta^\dagger)$, $(\zeta, S)$, $(\sigma^\dagger, \nu^\dagger)$, $(\nu, \sigma)$.
    \item In the invertibility sector, and when $\beta \neq -1$, the pair
    $$(\{B, \tau, c, S^\dagger, \zeta, \sigma^\dagger, \nu\},\{c^\dagger, \tau^\dagger, B^\dagger, \zeta^\dagger, S, \nu^\dagger, \sigma\})$$
    is a general contractible pair.
  \end{itemize}
\end{remark}

\begin{remark}
\label{rem:BisFA}
    By rule of Lemma \ref{lem:sectorrelations}, since invertibility {in the topological sector} implies $\beta \neq -1$, one can combine the smaller contractible pairs from Remark \ref{rem:SmallPairCompoundPair}. This yields the following unique solutions for the general contractible pair $(\{B, \tau, c\}, \{c^\dagger, \tau^\dagger, B^\dagger\})$:
  \begin{align}
    c &= - (\id + \beta T)^{-1} F_A , \\
    B &= - (\id + \beta T)^{-1} \left(\frac{1}{2} e \wedge e + d_A S \right), \\
    \tau &= - \frac{1}{(1+\beta)} d_A e.
  \end{align}
  This amounts to setting $\BB = - \frac{1}{(1+\beta)} (F_\AB)$.
\end{remark}

The following is an immediate corollary of the previous results:

\begin{corollary}
\label{cor:BiggestContractibleTransformations}
    Let $T, \beta$ be in the intersection of the topological and invertibility sector. Then on the Lagrangian submanifold $\LL$ defined by
    \begin{equation}
        \zeta^\dagger = S = \nu^\dagger = \sigma = c^\dagger = \tau^\dagger = B^\dagger = 0,
    \end{equation}
    the following unique solutions for the fields $\{S^\dagger, \zeta, \sigma^\dagger, \nu, B, \tau, c\}$ arise:
    \begin{subequations}
    \begin{align}
        c &= - (\id + \beta T)^{-1} F_A, \qquad &&B = - \frac{1}{2} (\id + \beta T)^{-1} (e \wedge e), \qquad &&\tau =  - \frac{1}{(1+\beta)} d_A e,\\
        S^\dagger &= 0, \qquad &&\sigma^\dagger = 0, \qquad &&\nu = 0.
    \end{align}
    \end{subequations}
    Meanwhile for $\zeta$ one obtains $\zeta = - e \wedge \epsilon$ for topological \qGR and $\zeta = 0$ for gravitational \qGR.
\end{corollary}

The above results provide several conditional paths for reduction of gravitational and topological \qGR by eliminating general contractible pairs (see Proposition \ref{prop:GeneralContractibleWeakBV}). They can be summarised in the following diagrams:

\begin{figure}[H]
  \[
  \begin{tikzcd}[sep = large]
    \text{Topological \qGR} \arrow[d, Rightarrow, ""]  & \FC_{\qGR}^{\beta} \arrow[dr, "\text{$(\id + \beta T) \neq 0$}"] \arrow[dd, dashed, ""'] & ~ \\
    \text{(Topological qgH)} \arrow[d, Rightarrow, ""] & ~ & \left(\FC_{qgH}^{\beta}\right) \arrow[dl, "\text{$(1+\beta) \neq 0$}"] \\
    \text{Topological gH} & \FC_{gH}^{\beta} & ~\\
    \text{Gravitational \qGR} \arrow[d, Rightarrow, ""] & \FC_{\qGR}^{T, \beta} \arrow[dr, "\text{$(\id + \beta T) \neq 0$}"] \arrow[dd, dashed, ""'] & ~ \\
    \text{(Gravitational qgH)} \arrow[d, Rightarrow, ""] & ~ & \left(\FC_{qgH}^{T, \beta}\right) \arrow[dl, "\text{$(1+\beta) \neq 0$}"] \\
    \text{Gravitational gH} & \FC_{gH}^{T, \beta} & ~
  \end{tikzcd}
  \]
  \caption[Elimination paths for gravitational/topological \qGR]{Paths for elimination of general contractible pairs for both topological and gravitational \qGR. By rule of Proposition \ref{prop:GeneralContractibleWeakBV} the resulting BV theories are BV equivalent. Note that for brevity we did not introduce the standard BV theories for qgH.}\label{fig:ReductionBranches}
\end{figure}
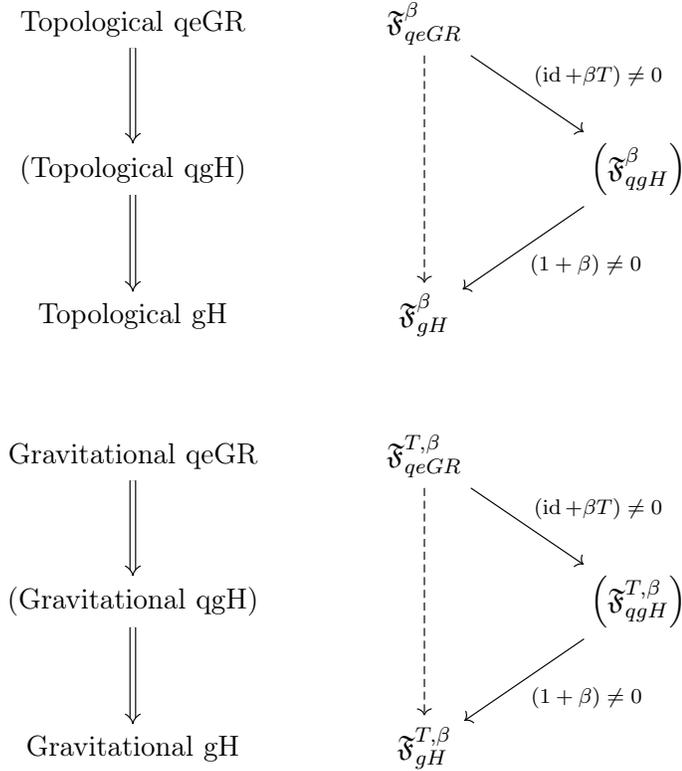

\begin{enumerate}[label=\roman*)]
  \item {Vertical branch:} In the invertibility sector ($\beta T \neq - \id$) with $\beta \neq -1$, one can immediately eliminate
  $$(\{S^\dagger, \zeta, \sigma^\dagger, \nu, B, \tau, c\},\{\zeta^\dagger, S, \nu^\dagger, \sigma, c^\dagger, \tau^\dagger, B^\dagger\}).$$
  For topological \qGR this leaves only $A,e,a,\epsilon,\alpha, \rho, f, \gamma$ and their antifields while for gravitational \qGR one retains $A, e, a, \xi$ and their antifields.

  \item {Right branch:} In the invertibility sector with no additional assumptions on $\beta$ one can eliminate
  $$ (\{S^\dagger, \zeta, \sigma^\dagger, \nu, B, c\},\{\zeta^\dagger, S, \nu^\dagger, \sigma, c^\dagger, B^\dagger\}), $$
  i.e. all fields of i) but $\tau$ and $\tau^\dagger$. Following Remark \ref{rem:BisFA}  and \ref{lem:sectorrelations} this suffices to also eliminate $(\tau, \tau^\dagger)$ in topological \qGR which yields the situation of i). For gravitational \qGR $(1+\beta) \neq 0$ needs to be assumed separately.
\end{enumerate}

\subsubsection{BV Equivalences}
\label{subsubsec:WeakEQs}
In this section we will follow the diagrams presented in Figure \ref{fig:ReductionBranches} to construct BV equivalent theories stemming from either gravitational and topological \qGR. This will enable us to clarify the relation of gravitational \qGR to theories of gravity within the BV formalism.

\begin{notation}
    Denote in the following by $\FC_\star^{T, \beta}$ the BV extensions associated to gravitational \qGR, qgH and gH. Similarly denote the BV extensions associated to the topological theories by $\FC_\star^{\beta}$.
\end{notation}

In what follows we will only explicitly present the proof of one of our BV equivalences, since the remaining ones, while lengthy, are completely analogous and straight-forward. Indeed, the strategy to prove BV equivalence by eliminatiion of general contractible pairs follows Proposition \ref{prop:GeneralContractibleWeakBV}, and the strategy of proof we will use is outlined in Proposition \ref{prop:qBFreduceALL}.

\begin{proposition}
\label{prop:qBFreduceALL}
  In the invertibility sector ($\beta T \neq - \id$) with $\beta \neq -1$ there exists a BV equivalence between gravitational quadratically extended GR (Definition \ref{def:BV_qBF}) and gravitational generalised Holst theory (Definition \ref{def:BV_gH}), $\FC_{\qGR}^{T, \beta} \simeq \FC_{gH}^{T, \beta}$.
\begin{proof}
  Let $T, \beta$ be within the gravitational sector such that one can use the BV extension of gravitational \qGR given in Definition \ref{def:BV_qBF}.
  From Remark \ref{rem:SmallPairCompoundPair} we know that for $\beta \neq -1$ and $\beta T  \neq -\id$, $(\{S^\dagger, \zeta, \sigma^\dagger, \nu, B, c, \tau\},\{\zeta^\dagger, S, \nu^\dagger, \sigma, c^\dagger, B^\dagger, \tau^\dagger\})$ is a general contractible pair.
  We use the unique solutions of the equations of motion for $\{S^\dagger, \zeta, \sigma^\dagger, \nu, B, c, \tau\}$ from Corollary \ref{cor:BiggestContractibleTransformations} to define the map
  {\small\begin{align}
    \Phi \colon (S^\dagger, \zeta, \sigma^\dagger, \nu, B, c, \tau, \zeta^\dagger, S, \nu^\dagger, \sigma, c^\dagger, B^\dagger, \tau^\dagger) \lmap (P^\dagger, U^\dagger, V^\dagger, W^\dagger, X^\dagger, Y^\dagger, Z^\dagger, P, U, V, W, X, Y, Z),
  \end{align}}
  where
  \begingroup
  \allowdisplaybreaks
  \begin{alignat}{7}
    P^\dagger &= \zeta^\dagger, &&\quad U^\dagger = S, &&\quad V^\dagger = \nu^\dagger, &&\quad W^\dagger = \sigma, &&\quad X^\dagger = c^\dagger, &&\quad Y^\dagger = B^\dagger, && \quad Z^\dagger = \tau^\dagger,\\
    P &= Q\zeta^\dagger, &&\quad U = QS, &&\quad V = Q\nu^\dagger, &&\quad W = Q\sigma, &&\quad X = Qc^\dagger, &&\quad Y = QB^\dagger, && \quad Z = Q\tau^\dagger,
  \end{alignat}
  \endgroup
  and all other fields are mapped via the identity, such that:
  \begin{alignat}{1}
    S^\dagger & \lmap P - (\id + \beta T)^{-1} d_A Y^\dagger - [a, P^\dagger] - \lie{\xi}^A P^\dagger, \\
    \zeta &\lmap U - (\id + \beta T)^{-1} d_A W^\dagger - [a, U^\dagger] - \lie{\xi}^A U^\dagger, \\
    \sigma^\dagger &\lmap V - (\id + \beta T)^{-1} d_A P^\dagger - [a, V^\dagger] - \lie{\xi}^A V^\dagger, \\
    \nu &\lmap W - [a, W^\dagger], \\
    B &\lmap X - (\id + \beta T)^{-1} \left(\frac{1}{2} e \wedge e + d_A U^\dagger \right) - [a, X^\dagger] - [Y^\dagger, W^\dagger] - \lie{\xi}^A X^\dagger, \\
    c &\lmap Y - (\id + \beta T)^{-1}(F_A) - [a, Y^\dagger] - \lie{\xi}^A Y^\dagger, \\
    \tau &\lmap Z - \frac{1}{(1+\beta)} d_A e - [a, Z^\dagger] - \lie{\xi}^A Z^\dagger.
  \end{alignat}
  Applying this transformation to $S_{\qGR}^{T, \beta}$ yields:

    \begingroup
    \allowdisplaybreaks
    \begin{subequations}
    \begin{align}
        \Phi\left(S_{\qGR}^{T, \beta}\right) &= \int_M -\frac{1}{2 } \left(\killing{(\id + \beta T)^{-1} F_A}{\left( \frac{1}{2} e \wedge e + d_A U^\dagger \right)} + \frac{1}{(1+\beta)} \eta(d_A e, d_A e) \right) \\
        &+ \killing{A^\dagger}{d_A a + \imath_\xi F_A} + \eta(e^\dagger, [a, e] + \lie{\xi}^A e) + \killing{a^\dagger}{\frac{1}{2} [a,a] - \frac{1}{2} \imath_\xi^2 F_A} + \frac{1}{2} \tr[\imath_{[\xi, \xi]} \xi^\dagger] \\
        & + \eta\left(Z^\dagger, ([a, \cdot] + \lie{\xi}^A) \left(Z - \frac{1}{(1+\beta)} d_A e - [a, Z^\dagger] - \lie{\xi}^A Z^\dagger\right)\right) \\
        &+ \frac{(1+\beta)}{2} \quad \eta\left(Z - [a, Z^\dagger] - \lie{\xi}^A Z^\dagger, Z - [a, Z^\dagger] - \lie{\xi}^A Z^\dagger\right) \\
        &+ \mathcal{K}\Bigg(X^\dagger, (\left[a, \cdot \right] + \lie{\xi}^A) \left(Y - (\id + \beta T)^{-1}(F_A) - [a, Y^\dagger] - \lie{\xi}^A Y^\dagger\right)\Bigg) \\
        &+ \killing{(Y - [a, Y^\dagger] - \lie{\xi}^A Y^\dagger)}{(\id+\beta T) (X - [a, X^\dagger] - [Y^\dagger, W^\dagger] - \lie{\xi}^A X^\dagger)} \\
        &+ \mathcal{K}\Bigg(P - (\id + \beta T)^{-1} d_A Y^\dagger - [a, P^\dagger] - \lie{\xi}^A P^\dagger, [a, U^\dagger] \\
        &+ (\id + \beta T) (U - [a, U^\dagger] - \lie{\xi}^A U^\dagger) + \lie{\xi}^A U^\dagger\Bigg) \\
        &+ \mathcal{K}\Bigg(Y^\dagger, \left[Y - (\id + \beta T)^{-1}(F_A) - [a, Y^\dagger] - \lie{\xi}^A Y^\dagger, W^\dagger \right] \\
        &+ (\left[a, \cdot \right] + \lie{\xi}^A) \left(X -  (\id + \beta T)^{-1} \left(\frac{1}{2} e \wedge e + d_A U^\dagger \right) - [a, X^\dagger] - [Y^\dagger, W^\dagger] - \lie{\xi}^A X^\dagger \right)\\
        &+ d_A \Big(U -  (\id + \beta T)^{-1} d_A W^\dagger - [a, U^\dagger] - \lie{\xi}^A U^\dagger\Big) + [Y^\dagger, (W - [a, W^\dagger])] \Bigg) \\
        &+ \mathcal{K}\Bigg(P^\dagger, (\left[a, \cdot \right] + \lie{\xi}^A) \left(U -  (\id + \beta T)^{-1} d_A W^\dagger - [a, U^\dagger] - \lie{\xi}^A U^\dagger\right)\Bigg) \\
        &+ \killing{V - [a, V^\dagger] - \lie{\xi}^A V^\dagger}{[a,W^\dagger] + (\id + \beta T) (W - [a, W^\dagger]) + \lie{\xi}^A W^\dagger} \\
        &-  \killing{(\id + \beta T)^{-1} d_A P^\dagger}{[a,W^\dagger] + \lie{\xi}^A W^\dagger} + \killing{V^\dagger}{([a, \cdot] + \lie{\xi}^A) (W - [a, W^\dagger])}.
    \end{align}
    \end{subequations}
    \endgroup
    While this action looks convoluted, the next step will drastically reduce it. Following what was done in Proposition \ref{prop:GeneralContractibleWeakBV} define a map $\varphi \colon \FF_{gH}^{T, \beta} \lra \FF_{\qGR}^{T,\beta}$ such that
    \begin{align}
    \varphi^* \kappa &=
    \begin{cases}
        0, \qquad \kappa \in \{P^\dagger, U^\dagger, V^\dagger, W^\dagger, X^\dagger, Y^\dagger, Z^\dagger, P, U, V, W, X, Y, Z\}, \\
        \kappa, \qquad \kappa \in \{A^\dagger, e^\dagger, a^\dagger, \xi^\dagger, A, e, a, \xi\}.
    \end{cases}         
    \end{align}   
    Applying it to $\Phi(S_{\qGR}^T)$ yields:
    \begingroup
    \allowdisplaybreaks
    \begin{subequations}
    \begin{align}
        \varphi^*(\Phi(S_{\qGR}^{T, \beta})) = &\int_M -\frac{1}{2 } \left(\killing{(\id + \beta T)^{-1} F_A}{e \wedge e} + \frac{1}{(1+\beta)} \eta(d_A e, d_A e) \right) + \killing{A^\dagger}{d_A a + \imath_\xi F_A} \\
        &+ \eta(e^\dagger, [a, e] + \lie{\xi}^A e) + \killing{a^\dagger}{\frac{1}{2} [a,a] - \frac{1}{2} \imath_\xi^2 F_A} + \frac{1}{2} \tr[\imath_{[\xi, \xi]} \xi^\dagger].
    \end{align}
    \end{subequations}
    \endgroup
    Rescaling the fields in the cotangent fibre, $A^\dagger, e^\dagger, a^\dagger, \xi^\dagger$, by $-1$ yields $-S_{gH}^{T, \beta}$. This rescaling can be promoted to a symplectomorphism. Comparing with the construction in Proposition \ref{prop:GeneralContractibleWeakBV} and Definition \ref{def:BV_gH} this proves the claim.
\end{proof}
\end{proposition}

The following results will remain without explicit proof, as it is an adaptation of the previous one.

\begin{proposition}
\label{prop:KqBFreduceALL}
  If $(1+\beta) \neq 0$ then there exists a BV equivalence between topological quadratically extended GR \ref{def:BV_qGR_top} and topological generalised Holst theory \ref{rem:BV_gH_top}, $\FC_{\qGR}^{\beta} \simeq \FC_{gH}^{\beta}$.
\end{proposition}

\begin{proposition}
\label{prop:qBFreduceTAU}
  In the invertibility sector there exists a BV equivalence $\FC_{\qGR}^{T, \beta} \simeq \FC_{qgH}^{T, \beta}$. For $(1+\beta) \neq 0$ there exists a BV equivalence $\FC_{qgH}^{T, \beta} \simeq \FC_{gH}^{T, \beta}$. This statement holds for both the gravitational and topological sector.
\end{proposition}

Using the above results yields the following theorem that theorem that summarises the relationship
between the different theories considered in this work:

\begin{theorem}
\label{theo:WeakBVdiagram}
  For any $T, \beta$ the following diagram of BV equivalences commutes:
  \begin{figure}[H]
    \[
    \begin{tikzcd}[sep = huge]
      \FC_{\qGR}^{T, \beta} \arrow[dr, leftrightarrow, "{(\id + \beta T) \neq 0}"] \arrow[dd, leftrightarrow, "{(1+\beta) \neq 0, \ (\id+\beta T) \neq 0}"'] & ~ \\
     ~ & \FC_{qgH}^{T, \beta} \arrow[dl, leftrightarrow, "{(1 + \beta) \neq 0}"] \\
     \FC_{gH}^{T, \beta} & ~
    \end{tikzcd}
    \]
  \end{figure}
\begin{proof}
  One needs to prove the claim for topological and gravitational \qGR separately, however the proof has been completely carried out in its parts. By combining Proposition \ref{prop:qBFreduceALL}, Proposition \ref{prop:KqBFreduceALL} and Proposition \ref{prop:qBFreduceTAU} all required components are provided.
\end{proof}
\end{theorem}

An immediate corollary of Theorem \ref{theo:WeakBVdiagram} is the following:

\begin{corollary}
    For $\beta \notin \{-1, 0\}$ and $T = \frac{1}{\beta}(\star^{-1} - \id)$ the BV theory of {deformed,} quadratically extended $BF$ (Definition Definition \ref{def:BV_qBF}) is BV equivalent to that of Palatini--Cartan--Holst \cite{CattaneoSchiavina2018} (also compare Remark \ref{rem:gHlikePCH}) up to a boundary term. Whenever $M$ is closed they are BV equivalent.
\end{corollary}

\section{Discussion and Outlook}
\label{sec:DiscussionOutlook}

In this work we have proposed a novel model for gauge theories and torsion modifications of gravity using a symmetry-broken $BF+\Lambda BB$ theory and without imposing external constraints. To this end we have utilised the trivial standard Lie algebra extensions presented in \cite{KathOlbrich2006, KathOlbrich2007} together with a deformation term. It was subsequently shown that the deformations of the resulting theory, dubbed quadratically extended GR (\qGR), can be divided into two sectors of interest, the topological sector for which \qGR is topological and the gravitational sector for which it has exactly $2$ local degrees of freedom and diffeomorphism symmetry. Through classical and later BV equivalences it was then shown that, for certain restrictions on the deformation and parameters, one can find significantly simpler, (weakly) BV equivalent, effective theories of \qGR. In particular it was shown that for specific deformation choices the quadratic torsion term that arises in these models directly relates \qGR to Palatini--Cartan--Holst theory (up to a boundary term that corresponds to the Nieh--Yan class). This result is stated in the sense of BV equivalences.

\subsection{\emph{BF} Theories of Gravity}
\label{subsec:BFgravity}

We want to conclude by briefly highlighting other known paths to obtaining established theories of gravity from $BF$ theory. In the following exposition we mostly follow \cite{Freidel_2012} and the sources presented therein for the introduction of (self-dual) Plebanski and MacDowell--Mansouri gravity. After each exposition we will highlight the relation and differences to the model of \qGR presented in this work.

\subsubsection{Self-dual Plebanski gravity}

The self-dual Plebanski formalism was first presented in \cite{Plebanski77} and operates on a $BF+BB$ type theory with Lie algebra $\gf = \su(2)$. Further fix the adjoint representation for the fields $A$ and $B$. A theorem due to Urbantke \cite{Urbantke84} states that if $B$ is non-degenerate then the Urbantke tensor $g_{\mu\nu}$, defined by
\begin{align}
    \sqrt{|g|} g_{\mu\nu} = \frac{1}{12} \epsilon_{ijk} \epsilon^{\alpha \beta \gamma \delta} B^i_{\mu \alpha} B^j_{\beta \gamma} B^k_{\delta \nu},
\end{align}
is a (pseudo-Riemannian) metric. 
The action of self-dual Plebanski is written as
\begin{align}
\label{eq:modelPlebanski}
    S[A, B, \Phi] \coloneqq \int_M B_i \wedge F(A)^i + \frac{1}{2} (\Lambda \delta_{ij} + \Phi_{ij}) B^i \wedge B^j.
\end{align}
The Lagrange-multiplier field $\Phi$ imposes the following condition on the field $B$:
\begin{align}
\label{eq:MetricityConstraint}
    B^i \wedge B^j = \frac{1}{3} \delta^{ij} B_k \wedge B^k.
\end{align}
Assuming non-degenerate $B$ one now uses the induced Urbantke metric $g_{\mu\nu}$ to define the associated tetrads $e$ such that $g = e^*\eta$ where $\eta$ is the Minkowski metric. One can then show (compare \cite[3]{Freidel_2012}) that $B$ is self-dual with respect to $g_{\mu\nu}$ and that, by rule of \eqref{eq:MetricityConstraint}, $B$ can be written in terms of the tetrads $e$ as either of
\begin{align}
    B^i_\pm = \sgn(\eta) (P_\pm)^i_{IJ} e^I \wedge e^J,
\end{align}
where $P_\pm$ denotes the projection on the first or second factor of the chiral decomposition $\so(3,\sgn(\eta)) \simeq \su(2) \oplus \su(2)$ respectively. This constraint is sometimes called the {simplicity constraint}. When working with the self-dual Plebanski action one usually fixes one of the two possible solutions.

Plugging this constraint into the Plebanski action one immediately recovers the form of Einstein--Cartan gravity in self-dual variables. Note however that if we work with $\su(2)$ and a Lorentzian metric, all fields are a priori complex. To obtain real gravity one needs to impose further {reality constraints} (compare \cite{Capovilla1991}). While classically unproblematic, they obstruct the quantisation of the theory (compare \cite[5]{Freidel_2012},\cite{Mikovi2006}).

Plebanski gravity is quite different from the family of \qGR theories considered here: Instead of an explicit symmetry breaking using additional terms in the fields $A$ and $B$, Plebanski gravity uses the algebraically complicated interplay of $\su(2)$ and $\so(3, \sgn(\eta))$ together with additional constraints (cf. \ref{eq:MetricityConstraint}) to reproduce a theory of gravity. Note that these ``external'' constraints, on which \qGR does not rely, become harder to handle in the quantisation procedure. However, inspecting \eqref{eq:modelPlebanski} it could be an interesting exercise to define a \qGR family over $\su(2)$ and find constraints similar to the simplicity and reality constraints to investigate possible direct parallels to Plebanski gravity.

\subsubsection{MacDowell--Mansouri Gravity}

An approach that differs from Plebanski's formulation of gravity can be found in that of MacDowell and Mansouri (MS), first introduced in \cite{MacDowellMansouri1977} without any relation to $BF$ theory. As proposed in \cite{Freidel_2012} we will here present an Euclidean version of MS gravity with positive cosmological constant to showcase its main features. Instead of splitting $\so(4)$ with a chiral decomposition as is done in self-dual Plebanski gravity, we make use of the split
\begin{align}
    \so(5) \simeq \so(4) \oplus \RR^4,
\end{align}
such that we can embed $\so(4)$ within the larger $\so(5)$. Note that in this context we should interpret the factor $\RR^4$ as the symmetry algebra for translations. Using this split we can write a connection $A$ on an $\so(5)$-bundle as $A = (\omega, e)$\footnote{Note that for simplicity we omit here some dimensionful constants introduced in \cite{Freidel_2012}.}. This split translates to the associated curvature as
\begin{align}
\label{eq:MScurvature}
    F(A) = (F(\omega) - e \wedge e, d_\omega e).
\end{align}
The action underlying MS gravity for $\Lambda > 0$ and $V \in \RR$ is
\begin{align}
\label{eq:ModelMMgravity}
    S[A,B] = \int_M B_{ij} \wedge F(A)^{ij} - \frac{1}{2} (\Lambda \delta_{\alpha\gamma} \delta_{\beta\delta} + \epsilon_{\alpha\beta\gamma\delta 4} \ V) B^{\alpha\beta} \wedge B^{\gamma\delta}.
\end{align}
As for usual $BF+\Lambda BB$ type theories, $B$ is an auxiliary field which one can eliminate, at least classically, to obtain the second order action:
\begin{align}
    S_{MM}[A] = \frac{1}{2(V^2 - \Lambda^2)} \int_M V \tr[F(A) \wedge \star F(A)] - \Lambda \tr[F(A) \wedge F(A)] +  \eta(d_\omega e, d_\omega e).
\end{align}
In the above action $\star$ denotes the Hodge star internal to $\so(4)$. The original formulation of MacDowell--Mansouri gravity in \cite{MacDowellMansouri1977} was not yet connected to $BF$ theories and contained only the first term. As noted in \cite[10]{Freidel_2012} one can make use of the explicit form of $A$ and $F(A)$ to bring the action into the following form:
\begin{align}
\label{eq:MacDowelMansouriTopologicalTerms}
    S[\omega, e] &= \frac{V}{(V^2 - \Lambda^2)}\int_M \left(\tr[T_{\frac{\Lambda}{V}} (F(\omega)) \wedge e \wedge e] - \frac{1}{2} \tr[e^4] \right) \\
    &+ C_1 E(\omega) + C_2 P(\omega) + C_3 NY(\omega, e).
\end{align}
Here we denote by $E(\omega), P(\omega)$ and $NY(\omega, e)$ the Euler, Pontryagin and Nieh--Yan class respectively. The factors $C_i$ are combinations of $V$ and $\Lambda$\footnote{The $C_i$ also depend on the choice of a length scale $l$, which we have chosen to drop in our discussion. See however \cite[10]{Freidel_2012}.}. Thus one obtains, up to topological contributions, the Palatini--Cartan action for gravity with a cosmological constant.

MacDowell--Mansouri gravity does not rely on constraints enforced by Lagrange multipliers. Together with its likeness to a $BF+BB$ type theory, which is then recast into the second order formalism in a form similar to topological Yang--Mills theory, one might expect that some quantisation properties carry over. Note however that MacDowell--Mansouri gravity is not a topological theory: The additional term introduced in \eqref{eq:ModelMMgravity} explicitly breaks the shift symmetry into a diffeomorphism symmetry. As \cite{Freidel_2012} put it: The term ``introduces a direction in the internal space''. The authors also mention that attempts to mitigate this problem by perturbatively expanding around the topological kinetic term has proven to be problematic since the first, topological order has different symmetries from all subsequent ones (cf. \cite{Rovelli2006, Cattaneo1998}).

We find that \qGR is quite similar to MacDowell--Mansouri gravity. Both essentially start with the idea of embedding the symmetries of theories of gravity, in our work $\so(3,1)$, in a larger Lie algebra. One can even see striking parallels between the respective curvature forms \eqref{eq:qBFcurvature} and \eqref{eq:MScurvature}. We have further shown in \ref{subsec:qBFboundary} and \ref{subsec:BVextensions} that the gravitational sector of \qGR explicitly breaks several shift symmetries into a diffeomorphism symmetry, just as MS gravity does. Furthermore, both theories retain their likeness to $BF+\Lambda BB$ type theories and thus allow for an interpretation of gravity as a deformation of a topological background theory. They also share the benefit of a constraint-free description of gravity, working solely with field equations in their primary fields.

The model of \qGR studied here does however display several features that MS gravity does not. First and foremost \qGR provides a genuine generalisation over MS gravity by allowing general deformation automorphisms $T$. Even with the restrictions on $T,$ and $\beta$ enforced to obtain BV equivalence to gH (cf. Theorem \ref{theo:WeakBVdiagram}) this still gives access to a vast family of theories. \qGR also avoids the (arbitrary) split of the Lie algebra performed in MS gravity, but it instead makes use of the (more natural) trivial quadratic extension of $\so(3,1)$ by a suitable orthogonal module $V$. While this \emph{does} introduce additional fields and symmetries, we have shown in this work that they can be eliminated within the BV formalism. Note as well, that the extension of $\so(3,1)$ chosen here is but the most trivial one, and other ones can be studied, see Appendix \ref{App:QuadLie}. This could be used to either formulate the most general action of gravity in first order variables (compare \cite{Rezende2009}) or rule out certain terms

All in all one could see the family of \qGR theories as a generalisation of MacDowell--Mansouri gravity. It allows for more general ``potential'' terms (compare \eqref{eq:ModelMMgravity}) while yielding interesting results for the role of torsion and the deformation of $BF+BB$ type theories. In a sense it combines the versatility of Krasnov's approach \cite{Krasnov2018} to modified gravity with the algebraically simple and constraint-free approach of MacDowell--Mansouri gravity. It would be particularly interesting to investigate which generalisations of the \qGR reproduce all topological terms found in \eqref{eq:MacDowelMansouriTopologicalTerms} such that it yields the most general form of gravity in tetrad variables (compare \cite{Rezende2009}).

\appendix
\section{Calculations for Theorem {\ref{theo:ConstraintAlgebraQBF}}}
\label{App:CalcBdyConstr}

Here we collect the calculations needed for the proof of Theorem \ref{theo:ConstraintAlgebraQBF}. We begin by stating the boundary $2$-form
\begin{align}
    \varpi^\partial_{qeGR} = \int_{\partial M} \killing{\delta B}{\delta A} + \eta(\delta \tau, \delta e) + \killing{\delta S}{\delta c}.
\end{align}
Now calculate the action of $\delta$ on the six constraint functions and derive the respective hamiltonian vector fields with respect to $\varpi^\partial_{qeGR}$.
\begin{itemize}
    \item For $I_\epsilon$
    \begin{align}
        \delta I_\epsilon = \int_{\partial M} \killing{\delta c}{ e \wedge \epsilon} + \killing{c}{\delta e \wedge \epsilon} + \eta([\delta A, \tau] - d_A \delta \tau, \epsilon).
    \end{align}
    The non-trivial components of the associated hamiltonian vector field $\mathbb{I}_\epsilon$ are:
    \begin{align}
        (\mathbb{I}_\epsilon)_B &= - \tau \wedge \epsilon, \qquad (\mathbb{I}_\epsilon)_S = e \wedge \epsilon, \qquad (\mathbb{I}_\epsilon)_\tau = - [c, \epsilon], \qquad (\mathbb{I}_\epsilon)_e = d_A \epsilon.
    \end{align}

    \item For $M_\zeta$
    \begin{align}
        \delta M_\zeta = \int_{\partial M} \killing{-d_A \delta A + (\id + \beta T) \delta c}{ \zeta}.
    \end{align}
    The non-trivial components of the associated hamiltonian vector field $\mathbb{M}_\zeta$ are:
    \begin{align}
        (\mathbb{M}_\zeta)_B &= d_A \zeta, \qquad (\mathbb{M}_\zeta)_S = (\id + \beta T) \zeta.
    \end{align}

    \item For $L_\rho$
    \begin{align}
        \delta L_\rho = \int_{\partial M} \eta([\delta A, e] - d_A \delta e + (1+\beta) \delta \tau, \rho).
    \end{align}
    The non-trivial components of the associated hamiltonian vector field $\mathbb{L}_\rho$ are:
    \begin{align}
        (\mathbb{L}_\rho)_B &= - \rho \wedge e, \qquad (\mathbb{L}_\rho)_e = (1+\beta) \rho, \qquad (\mathbb{L}_\rho)_\tau = d_A \rho.
    \end{align}

    \item For $K_\alpha$
    \begin{align}
        \delta K_\alpha = \int_{\partial M} \killing{(\id + \beta T)\delta B + [\delta A, S] - d_A \delta S + \delta e \wedge e}{\alpha}.
    \end{align}
    The non-trivial components of the associated hamiltonian vector field $\mathbb{K}_\alpha$ are:
    \begin{align}
        (\mathbb{K}_\alpha)_A &= (\id + \beta T) \alpha, \qquad (\mathbb{K}_\alpha)_B = - [\alpha, S], \qquad (\mathbb{K}_\alpha)_c = d_A \alpha, \qquad (\mathbb{K}_\alpha)_\tau = - [\alpha, e].
    \end{align}

    \item For $H_\theta$
    \begin{align}
        \delta H_\theta = \int_{\partial M} \killing{[\delta A, B] - d_A \delta B - \delta e \wedge \tau + e \wedge \delta \tau - [\delta c, S] + [c, \delta S]}{\theta}.
    \end{align}
    The non-trivial components of the associated hamiltonian vector field $\mathbb{H}_\theta$ are:
    \begin{align}
        (\mathbb{H}_\theta)_A &= d_A \theta, \qquad (\mathbb{H}_\theta)_B = [\theta, B], \qquad (\mathbb{H}_\theta)_e = [\theta, e],\\
        (\mathbb{H}_\theta)_\tau &= [\theta, \tau], \qquad (\mathbb{H}_\theta)_S = [\theta, S], \qquad (\mathbb{H}_\theta)_c = [\theta, c].
    \end{align}
\end{itemize}
Using the above one can calculate the Poisson-brackets between the generating function. For two $0$-forms $f, g$ the bracket takes the following form: $\{f, g\} = \imath_{X_f} \delta B$. Here $X_f$ denotes the hamiltonian vector field associated to $f$. We restrict here to the three  non-trivial brackets that might result in the constraint algebra not defining a Poisson subalgebra, namely $\{K_\alpha, L_\rho\}$, $\{K_\alpha, I_\epsilon\}$ and $\{I_\epsilon, L_\rho\}$.
\begin{itemize}
    \item For $\{K_\alpha, L_\rho\}$:
    \begin{align}
        \{K_\alpha, L_\rho\} &= \imath_{\mathbb{K}_\alpha} \int_{\partial M} \eta([\delta A, e] - d_A \delta e + (1+\beta) \delta \tau, \rho) \\
        &= \int_{\partial M} \eta([(\id + \beta T) \alpha, e]- (1+\beta) [\alpha, e], \rho) = \int_{\partial M} \eta([\beta (T - \id) \alpha, e], \rho). \\
    \end{align}
    
    \item For $\{K_\alpha, I_\epsilon\}$:
    \begin{align}
        \{K_\alpha, I_\epsilon\} &= \imath_{\mathbb{K}_\alpha} \int_{\partial M} \killing{\delta c}{ e \wedge \epsilon} + \killing{c}{\delta e \wedge \epsilon} + \eta([\delta A, \tau] - d_A \delta \tau, \epsilon) \\
        &= \int_{\partial M} \killing{d_A \alpha}{ e \wedge \epsilon} + \eta([(\id + \beta T) \alpha, \tau] + d_A [\alpha, e], \epsilon) \\
        &\overset{\partial}{=} L_{[\alpha, \epsilon]} - \int_{\partial M} \eta([\beta (\id - T) \alpha, \tau] , \epsilon). \\
    \end{align}
    
    \item For $\{I_\epsilon, L_\rho\}$:
    \begin{align}
        \{I_\epsilon, L_\rho\} &= \imath_{\mathbb{I}_\epsilon} \int_{\partial M} \eta([\delta A, e] - d_A \delta e + (1+\beta) \delta \tau, \rho) \\
        &= \int_{\partial M} \eta(- [F_A, \epsilon] + (1+\beta) [c, \epsilon], \rho) = \int_{\partial M} \killing{F_A + (1+\beta)c}{\epsilon \wedge \rho} \\
        &= M_{\epsilon \wedge \rho} + \int_{\partial M} \killing{\beta (\id - T)c}{\epsilon \wedge \rho}
    \end{align}
\end{itemize}

\section{Quadratic Lie Algebra Extensions}
\label{App:QuadLie}

Let $\gf$ be a finite-dimensional Lie algebra and $(\rho, V, \inp_V)$ an orthogonal $\gf$-module. Think of $V$ as an abelian metric Lie algebra. Combining the wedge product $\wedge$ on Chevalley--Eilenberg cochains $C^\bullet_{CE}(\gf; V)$ of $\gf$ and the non-degenerate symmetric pairing $\pair{\cdot}{\cdot}_V$ on $V$ one obtains a bilinear multiplication map
\begin{align}
    \langle \cdot \wedge \cdot \rangle \colon C_{CE}^i(\gf; V) \times C_{CE}^j(\gf; V) \atob{\wedge} C_{CE}^{i+j}(\gf; V \otimes V) &\atob{\pair{\cdot}{\cdot}_V} C_{CE}^{i+j}(\gf),\\
      (\alpha, \beta) &\lmap \langle \alpha \wedge \beta \rangle.
\end{align}

\begin{definition}[Quadratic Cochains and Cocycles]
  Let $k \in 2\NN_0$. Define the {group of quadratic $(p-1)$-cochains}\index{Quadratic cohomology!cochains} as
  \begin{equation}
    C_Q^{p-1}(\gf; V) \coloneqq C_{CE}^{p-1}(\gf; V) \oplus C_{CE}^{2p-2}(\gf; V).
  \end{equation}
  The group structure is induced by the following multiplication map:
  \begin{align}
    * \colon C_Q^{p-1}(\gf; V) \times C_Q^{p-1}(\gf; V) &\lra C_Q^{p-1}(\gf; V), \\
    (\alpha_1, \beta_1) \times (\alpha_2, \beta_2) &\lmap (\alpha_1, \beta_1) * (\alpha_2, \beta_2) \coloneqq (\alpha_1 + \alpha_2, \beta_1 + \beta_2 + \frac{1}{2}\langle \alpha_1 \wedge \alpha_2 \rangle).
  \end{align}
  In the same manner we define the {set of quadratic $p$-cocycles}\index{Quadratic cohomology!cocycles} as
  \begin{equation}
    \ZC_Q^{p}(\gf; V) \coloneqq \left\{ (\alpha, \beta) \in C_{CE}^p(\gf; V) \oplus C_{CE}^{2p-1}(\gf) \bar d\alpha = 0, d\beta = \frac{1}{2}\langle \alpha \wedge \alpha \rangle \right\}.
  \end{equation}
\end{definition}

\begin{lemma}[{\cite[Lemma 2.1]{KathOlbrich2006}}]
\label{lem:actionLemma}
  Let $k \in 2\NN_0$. Let further $(\alpha, \beta) \in C_{CE}^k(\gf; V) \oplus C_{CE}^{2k-1}(\gf)$ and $(\gamma, \delta) \in C_{CE}^{k-1}(\gf; V) \oplus C_{CE}^{2k-2}(\gf)$.
  One can define a right action of $C_{Q}^{k-1}(\gf; V)$ on $C_{CE}^k(\gf;V) \oplus C_{CE}^{2k-1}(\gf)$ acting as the identity on $k$-cocycles $\ZC^k_Q(\gf; V)$ by
  \begin{equation}
    (\alpha, \beta) (\gamma, \delta) \coloneqq \left( \alpha + d\gamma, \beta + d\delta + \lieprod{\left( \alpha + \frac{1}{2} d\gamma \right)}{\gamma} \right).
  \end{equation}
\end{lemma}

\begin{definition}
  Let again $k \in 2\NN_0$. Since $\ZC^{2k-2}(\gf)$ acts trivially on $C_{CE}^k(\gf; V) \oplus C_{CE}^{2k-1}(\gf)$ using the action defined in Lemma \ref{lem:actionLemma}, it amounts to an action of the {group of quadratic $k$-coboundaries}\index{Quadratic cohomology!coboundaries}
  \begin{equation}
    \BC_Q^k(\gf; V) \coloneqq (C_{CE}^{k-1}(\gf; V) \oplus \BC_{CE}^{2k-1}(\gf), *),
  \end{equation}
  equipped with the following multiplication map:
  \begin{align}
    * \colon \BC_Q^k(\gf; V) &\times \BC_Q^k(\gf; V) \lra \BC_Q^k(\gf; V), \\
    (\alpha_1, \beta_1) &\times (\alpha_2, \beta_2) \lmap (\alpha_1, \beta_1) * (\alpha_2, \beta_2) \coloneqq \left(\alpha_1 + \alpha_2, \beta_1 + \beta_2 + \frac{1}{2}d \lieprod{\alpha_1}{\alpha_2} \right).
  \end{align}
\end{definition}

\begin{definition}[Quadratic Cohomology]
  Let $k \in 2\NN_0$. Define the {$k$-th quadratic cohomology}\index{Quadratic cohomology} of $\gf$ as the quotient space
  \begin{equation}
    \HH_Q^k(\gf;V) \coloneqq \ZC_Q^k(\gf; V) / \BC_Q^k(\gf; V),
  \end{equation}
  where the quotient is understood as the quotient space defined by the action of $\BC_Q^k(\gf; V)$ on $\ZC_Q^k(\gf; V)$. We denote a cohomology class as $[\alpha, \beta] \in \HH_Q^k(\gf;V)$.
\end{definition}

We will now describe a procedure pioneered in \cite{KathOlbrich2006} to extend any pair of a Lie algebra and an orthogonal module on it to a metric Lie algebra. For this we closely follow \cite[chapter 3]{KathOlbrich2006}.

\begin{remark}
  The above notion of a quadratic extension coincides with the categorial notion of an algebra extension where an extension of an algebra $\mathfrak{a} \in \mathrm{Alg}$ over a field $\kk$ is given by an epimorphism $\widetilde{\mathfrak{a}} \atob{p} \mathfrak{a}$. If the kernel of $p$ exists, as it does for most cases, this induces a short exact sequence
  $$ \ker(p) \lra \widetilde{\mathfrak{a}} \atob{p} \mathfrak{a}. $$
\end{remark}

The following extension procedure for Lie algebras makes use of their quadratic cohomology and provides a generalisation to the trivial standard extensions introduced in Definition \ref{def:StandardExtension_triv}.

\begin{definition}[Standard Extensions]
\label{def:StandardExtension_ntriv}
  Let $(\alpha, \beta) \in \ZC^2_Q(\gf;V)$ and define the vector space $\df \coloneqq \gf^* \oplus V \oplus \gf$. We define an inner product $\inp_\df \colon \df \times \df \lra \kk$ by
  \begin{equation}
    (z_1 + v_1 + g_1) \times (z_2 + v_2 + g_2) \lmap \pair{v_1}{v_2}_V + z_1(g_2) + z_2(g_1).
  \end{equation}
  We further denote by $[\cdot, \cdot]_\df \colon \df \times \df \lra \df$ the unique antisymmetric bilinear assignment such that
  \begin{enumerate}
    \item $\pair{[X,Y]_\df}{Z}_\df = \pair{X}{[Z,Y]_\df}_\df$ for all $X,Y,Z \in \df$,
    \item $[\df, \gf^*]_\df \subset \gf^*$, $[\gf^*, \gf^*]_\df = 0$, $[V, V]_\df \subset \gf^*$,
    \item $[g_1, g_2]_\df = \beta(g_1, g_2, \cdot) + \alpha(g_1, g_2) + [g_1, g_2]_\gf$ for all $g_1, g_2 \in \gf$,
    \item $\pair{[g, v_1]_\df}{v_2}_\df = \pair{\rho(g) v_1}{v_2}_\df$ for all $g \in \gf$ and $v_1, v_2 \in V$.
  \end{enumerate}
  In the above equations the canonical inclusions of $\gf^*, \gf$ and $V$ into $\df$ were omitted to clean up the equations. We call the standard extension arising from the metric Lie algebra $\df_{\alpha, \beta}(\gf, V, \rho) \coloneqq (\df, [\cdot, \cdot]_\df, \inp_\df)$ the {standard extension by $(\alpha, \beta)$ of $\gf$ by $V$}\index{Quadratic extension!standard extension} and denote it by $(\df_{\alpha, \beta}(\gf, V, \rho), \gf^*, i, p)$ or $\df_{\alpha, \beta}(\gf, V, \rho)$ for short.
  Note that $i \colon V \lra V \oplus \gf$ and $p \colon V \oplus \gf \lra \gf$ are given by the canonical injection and projection. For $(\alpha, \beta) = (0,0)$, call the arising standard extension the {trivial standard extension}\index{Quadratic extension!trivial quadratic extension}.
\end{definition}The definition of standard extensions is justified by the following result:

\begin{proposition}[{\cite[Proposition 3.1, 3.2]{KathOlbrich2006}}]
  The antisymmetric bilinear assignment $[\cdot, \cdot]_\df$ constructed in Definition \ref{def:StandardExtension_triv} is unique. For $(\alpha, \beta) \in \ZC^2_Q(\gf;V)$ the tuple $\df_{\alpha, \beta}(\gf, V, \rho) \coloneqq (\df, [\cdot, \cdot]_\df, \inp_\df)$ defines a metric Lie algebra and a quadratic extension of $\gf$ by $(\rho, V, \inp_V)$.
\end{proposition}

The following coherence result shows that, up to isomorphisms, we only need to consider standard extensions:

\begin{theorem}[{\cite[Theorem 3.1]{KathOlbrich2006}}]
  The equivalence classes of quadratic extensions of a Lie algebra $\gf$ by an orthogonal $\gf$-module $(\rho, V, \inp_V)$ are in one-to-one correspondence quadratic $2$-cohomology classes in $\HH_Q^2(\gf; V)$.
\end{theorem}

Thus in particular one can talk about \emph{the} trivial standard extension of a Lie algebra by an orthogonal module.

\printbibliography

\end{document}